%% file: BregmanPowerDivergenceV6.tex
\documentclass[11pt]{article}
\usepackage{accents}
\usepackage{fullpage,amssymb,graphicx,hyperref}
\usepackage[ruled]{algorithm2e}

\graphicspath{{Fig/}{Fig2/}{Fig3/}{Fig5/}{Fig6/}{./}}

\input{macrosBDPD.tex}

\newtheorem{Problem}{Problem}
\newtheorem{Example}{Example}

\begin{document}
 
\title{Bregman Power Divergences}

\author{Frank Nielsen\\  Sony Computer Science Laboratories, Inc., Tokyo, Japan.
}

 \date{}

% make the title area
\maketitle

\sloppy

% As a general rule, do not put math, special symbols or citations
% in the abstract or keywords.
\begin{abstract}
Power distances between weighted Euclidean points are fundamental dissimilarity measures in computational geometry, used notably for constructing power diagrams and their dual weighted regular triangulations. 
Bregman divergences, on the other hand, are fundamental deviation measures in information theory and serve as canonical divergences of dually flat spaces in information geometry.
In this work, we introduce Bregman power divergences, an extension of power distances to Bregman divergences, and provide two complementary geometric interpretations  of that measure based on (i) the Bregman potential lifting and on (ii) the Bregman Pythagorean theorem. 
We show that Bregman power divergences between weighted parameters amount to equivalent ordinary power distances expressed using the mixed-parameterization induced by the convex duality of Bregman generators.
This equivalence enables Bregman proximity problems to be solved as Euclidean proximity tasks at the expense of using an extra dimension.
As an illustration of this equivalence, we reduce the Bregman diameter and closest-pair problems to Euclidean power distance problems.
In the second part, we further elucidate the correspondence between Bregman power divergences and power distances through the equivalence of Bregman potential lifting and paraboloid lifting techniques. Last, we consider the space of Bregman and power spheres, and show how to construct a Bregman circumsphere passing through prescribed parameters as a corresponding power circumsphere, called an orthosphere.
We further discuss several applications. 
\end{abstract}


\input{sec-contributions}

\input{sec-background}

\input{sec-proximity}

\input{sec-Voronoi}

\input{sec-SpaceSphere}
\input{sec-conclusion}

\section*{Notations}\label{sec:nota}
\input{notations2.tex}

\appendix

%%%%%%%%%%
\section{Bregman centroids and barycenters}\label{sec:barycenter}
%%%%%%%%%%

Let $\{\hat{p}_1=(p_1,w_1),\ldots,\hat{p}_n=(p_n,w_n)\}$ be a set of weighted points.
Define the power barycenter with respect to a positive weight vector $\lambda\in\Delta_n$ (standard $(n-1)$-dimensional simplex), as  
\begin{eqnarray*}
\arg\min_p && \sum_i \lambda_i \sigma(p,\hat{p}_i),\\
&=& \sum_i \lambda_i \, \left(\inner{p-p_i}{p-p_i}-w_i\right),\\
&\equiv & E(p):=\inner{p}{p}-\sum 2\lambda_i\inner{p}{p_i}.\\
\end{eqnarray*}
The power barycenter is called the power centroid when the weight is uniform, i.e., $\lambda=(\frac{1}{n},\ldots,\frac{1}{n})$.
Since $\nabla_p E(p)=2p-2\sum_i \lambda_i p_i$, the minimum satisfies $\nabla_p E(p)=0 \Rightarrow p=\sum \lambda_i {p_i}$, and the minimum is unique 
because $\nabla_p^2 E(p)=2I\succ 0$ (positive-definite Hessian).

Notice that if we ask for $\arg\min_{p,w}  \sum_i \lambda_i \sigma((p,w),\hat{p}_i)$, then there is no minimum as $w\rightarrow\infty$.

The left Bregman centroid of a set of $n$ parameters $\{\theta_i\}$ is the minimizer of
\begin{eqnarray*}
\arg\min_\theta && \sum_i \lambda_i B_F(\theta:\theta_i),\\
\arg\min_\theta && \sum_i \frac{1}{2} \lambda_i \sigma(\hat\theta,\weta_i),\\
\arg\min_\theta &&  -w_F(\theta ) + \sum_i \lambda_i \sigma(\theta,\weta_i),\\
&\equiv& E(\theta):= -\inner{\theta}{\theta}+2F(\theta)+ \inner{\theta}{\theta}-2\sum_i \lambda_i \inner{\theta}{\eta_i}.
\end{eqnarray*}

We have $\nabla_\theta E(\theta)= 2\nabla F(\theta)-2\sum_i \lambda_i\eta_i$.
Hence, the minimum is obtained for $\nabla_\theta E(\theta_L)=0\Rightarrow \nabla F(\theta_L)=\sum_i \lambda_i\eta_i$, i.e., $\theta_L=(\nabla F)^{-1}\left(\sum_i \lambda_i \nabla F(\theta_i)\right)$. This minimum is unique because $\nabla_\theta^2 E(\theta)= 2\nabla^2 F(\theta)\succ 0$ because $F(\theta)$ is strictly convex.
Thus we have rederived the result presented in~\cite{nielsen2009sided} by using the equivalence of Bregman divergences with power divergences.
Notice that by symmetry, the right Bregman centroid $\theta_R$ is similarly found as $\eta_R=(\nabla F^*)^{-1}\left(\sum_i \lambda_i \nabla F^*(\eta_i)\right)$, i.e.,
 $\theta_R=\nabla F^*(\eta_R)=\sum_i \lambda_i \nabla F^*(\theta_i)=\sum_i \lambda_i \theta_i$, the center of mass~\cite{CBD-2005}.

Consider the maximum likelihood estimator (MLE) for an exponential family $\{p_\theta\}$ with cumulant function $F(\theta)$ given $n$ i.i.d. samples $x_1,\ldots, x_n$.
Using the Bregman-exponential family duality~\cite{CBD-2005,nielsen2012k} $\log p_\theta(x)=-B_{F^*}(x:\eta)+F^*(x)$, we have
\begin{eqnarray*}
\argmax_\theta \prod_i p_\theta(x_i) &=& \argmax_\theta \sum_i \log p_\theta(x_i),\\
&\equiv&    \argmin_\theta \sum_i Y_{F^*}(x_i:\theta),\\
&=& \argmin_\theta \sum_i \sigma(\hat\theta,\check x_i).
 \end{eqnarray*}

Thus the MLE amounts to a power centroid:

\begin{Property}
The  maximum likelihood estimator (MLE) for an exponential family amounts to a power centroid.
\end{Property}

%%%%%%%%%
\section{Analysis of the Frank--Wolfe power MEB algorithm}\label{sec:FWapp}
%%%%%%%%%

The power MEB minimizes the following objective function:
$$
\min_{x\in\bbR^d} L(x), \quad L(x)\eqdef\max_{i\in [n]} \|x-p_i\|^2-w_i.
$$
Let us rewrite $L(x)$ as follows:
$$
L(x)=\|x\|^2+\max_i \left(\beta_i-2\inner{p_i}{x}\right), \quad \beta_i=\|p_i\|^2-w_i. 
$$
Then the power MEB is modeled as 
\begin{eqnarray*}
&& \min_{x,t} \|x\|^2+t,\\
&& t\geq \beta_i-2\inner{p_i}{x}.
\end{eqnarray*}

Let $\lambda_i$'s be the Lagrangian multipliers with $\lambda\in\Delta_n$:
$$
L(x,t,\lambda)=\|x^2\|+t+\sum_i \lambda_i(\beta_i-t-\inner{p_i}{x}).
$$

 The dual problem writes as
\begin{equation}
\max_{\lambda\in\Delta_n} G(\lambda), \quad G(\lambda)=\sum_i\lambda_i\beta_i-\|\sum_i\lambda_i p_i\|^2.
\end{equation}

Since the primal optimization function $L(x)$ is convex and Slater's condition holds, strong duality ensures that we have
$$
\min_{x\in\bbR^d} L(x)=\max_{\lambda\in\Delta_n} G(\lambda).
$$
Observe that left-hand-side minimization problem is on $\bbR^d$ while right-hand-side maximization problem in on the standard simplex $\Delta_n$.

We now sketch study the Frank--Wolfe algorithm for the power MEB problem:
Although there exists many refined techniques~\cite{braun2022conditional}, we consider the vanilla Frank--Wolfe algorithm/analysis below.

At iteration $t$, let $c_t=\sum_i \lambda_i p_i$ and $f_t=\argmax_{i\in[n]} \beta_i-2\inner{p_i}{c_t}$.
Thus the Frank-–Wolfe oracle chooses the point of maximum power distance and generalizes the BC algorithm.
 
To prove convergence of the Frank--Wolfe algorithm, we consider the curvature $\kappa$ introduced by Jaggi~\cite{jaggi2011sparse,jaggi2013revisiting} which is related to a Bregman divergence.  We skip details but get finite bounded curvature for the power MEB ($\kappa\leq 8R^*$ where $R^*={r^*}^2$). Hence, the Frank-Wolfe power MEB approximation algorithm returns a $(1+\epsilon)$-approximation of the power circumcenter after $T=O(\frac{1}{\epsilon})$ steps.

\section*{Acknowledgment}

The author gratefully acknowledges Prof. Ting-Kam Leonard Wong and Dr. Amanjit Singh Kainth of the University of Toronto, Canada, for fruitful discussions and for sharing insights into their work on Bregman–Wasserstein divergences~\cite{BW-2025}, which provided the initial inspiration for this paper.

\bibliographystyle{plain}
 \bibliography{BDPD}

\tableofcontents

% that's all folks
\end{document}

%% file: macrosBDPD.tex
\def\calK{\mathcal{K}}

\newenvironment{proof}{\paragraph{Proof:}}{\hfill$\square$}

\def\Bi{\mathrm{Bi}}
\def\Burg{{\mathrm{Burg}}}
\def\termdef#1{{\em #1}}

\def\Orient{{\mathrm{Orient}}}
\def\Orthosphere{{\mathrm{Orthosphere}}}

\def\sign{\mathrm{sign}}
\def\Shannon{{\mathrm{Shannon}}}
\def\hatq{{\hat{q}}}
\def\bbC{\mathbb{C}}
\def\calW{\mathcal{W}}
\def\vN{{\mathrm{vN}}}
\def\calV{\mathcal{V}}
\def\calA{\mathcal{A}}
\def\calD{\mathcal{D}}
\def\GL{{\mathrm{GL}}}

\def\calB{\mathcal{B}}
\def\tr{\mathrm{tr}}
\def\det{\mathrm{det}}
\def\st{\ :\ }

\def\BregmanCircumSphere{{\mathrm{BregmanCircumSphere}}}

\def\hatp{{\hat{p}}}

\def\wtheta{{\hat{\theta}}}
\def\weta{{\check\eta}}

\def\wwcalA{ {{\hat{\hat{\mathcal{A}}}}} }
\def\wwcalB{ {{\check{\check{\mathcal{B}}}}} }

\def\wwtheta{{{\hat{\hat\theta}}}}
\def\wweta{{{\check{\check\eta}}}}

\def\calB{\mathcal{B}}

\def\sphere{{\mathrm{sphere}}}
\def\ball{{\mathrm{ball}}}

\def\bbR{\mathbb{R}}
\def\Bi{\mathrm{Bi}}
\def\eqdef{{:=}}
\def\st{{\ :\ }}

\def\Inner#1#2{{\left\langle #1,#2\right\rangle}}
\def\inner#1#2{{\langle #1,#2\rangle}}
\def\calT{\mathcal{T}}
\def\calE{\mathcal{E}}
\def\calF{\mathcal{F}}
\def\calH{\mathcal{H}}
\def\KL{\mathrm{KL}}

\def\calP{\mathcal{P}}

\def\calP{\mathcal{P}}
\def\Vor{{\mathrm{Vor}}}
\def\argmax{{\mathrm{argmax}}}
\def\argmin{{\mathrm{argmin}}}

\def\calS{{\mathcal{S}}}

\newtheorem{Definition}{Definition}

\newtheorem{Remark}{Remark}
\newtheorem{Property}{Property}
\newtheorem{Proposition}{Proposition}
\newtheorem{Corollary}{Corollary}

%% file: sec-contributions.tex
%%%%%%%%%%%%%%%%%%%%%%
\section{Introduction}
%%%%%%%%%%%%%%%%%%%%%% 

On one hand, power distances are a common dissimilarity measure between weighted points rooted in computational geometry~\cite{de2008computational} which found many applications ranging from 3D meshing~\cite{cheng2013delaunay} to analysis of neural networks~\cite{balestriero2019geometry}.
On the other hand, Bregman divergences are the workhorses in machine learning and information geometry~\cite{IG-2016}.
Bregman divergences can be manipulated equivalently as Fenchel-Young divergences~\cite{azoury2001relative} using the mixed parameterization induced by convex duality of Bregman generators.
This paper focuses on demonstrating some key advantages of using the power distance representation of Bregman/Fenchel-Young divergences as power distances.

We summarize our main contributions as follows:

\begin{itemize}

\item We define Bregman power divergences which unify power distances with Bregman divergences in \S\ref{sec:bdpdmain}, and show how Bregman power divergences can be reduced canonically to  ordinary power distances by coupling the point weights in the mixed parameterization with the convex conjugate Bregman potential functions (Proposition~\ref{prop:bpdequivpd}).
We then show how to interpret those Bregman power divergences  either (i) as a vertical gap read from the graph of the potential function  (Figure~\ref{fig:bpd})  or (ii) from an information-geometric point of view   (Figure~\ref{fig:bpdeta}).

	\item We show in \S\ref{sec:NN} how to perform Bregman proximity queries as equivalent Euclidean proximity queries using an extra dimension.
	We illustrate how to use this proximity query equivalence by considering computing the Bregman diameter or the Bregman closest pair.
	The main advantage of using the power distance representation of Bregman divergences is to consider these Bregman proximity problems as Euclidean proximity problems on bichromatic point sets: One (red) set plays the  role of the query set while the other (blue) set is the neighbor data set.
	Our approach is practical for solving problems based on Bregman proximity queries as proximity query data structures and software libraries have been well-studied in computational geometry~\cite{de2008computational}.
	
	\item In \S\ref{sec:Vor}, we state the equivalence between Bregman potential liftings and the canonical paraboloid lifting.
This section provides another complementary viewpoint of the property that Bregman power divergences can be expressed as double-sided power distances.
 	
 \item We recall the spaces of Bregman spheres~\cite{BVD-2010} in \S\ref{sec:spacesphereortho}.
The space of spheres allows one to interpret geometrically the construction of Bregman and power circumspheres passing through a given set of parameters and link these constructions to determinant calculations.
We show how to convert the calculation of Bregman circumspheres as equivalent power circumspheres (called orthospheres) in Proposition~\ref{prop:equivcircumcenter}.
We also report exact formula for the Bregman circumcenter and the Bregman radius for the 1D case in Proposition~\ref{prop:BMEB1D},  and interpret those quantities geometrically from the potential function graph. 

\end{itemize}

We conclude this work by summarizing and discussing the key results in \S\ref{sec:concl}.
The core text is followed by an Appendix which consists in two parts: The first part revisits the Bregman barycenter problems from the viewpoint of the equivalence with power distance barycenters in \S\ref{sec:barycenter}. The second part analyzes in \S\ref{sec:FWapp} the numerical Frank--Wolfe approximation algorithm used to approximate the minimum enclosing power ball.
A summary of main notations is given at the end of this article.

%% file: sec-background.tex
%%%%%%%%%
\section{Power distances and Power Bregman divergences}\label{sec:bdpdmain}
%%%%%%%%%

%%%%%%%%%%%
\subsection{Overview}
%%%%%%%%%%%
In this section, we first begin by defining the power distances in \S\ref{sec:pd} (both single-sided and double-sided versions) and the Bregman and Fenchel-Young divergences in \S\ref{sec:pd}. Then we show how to express Bregman divergences as equivalent double-sided power distances in \S\ref{sec:bdpd} and introduce a generalization of power distances in \S\ref{sec:bpd} that we term Bregman power divergences.
The main property is that Bregman power divergences can always be reduced to power distances by coupling weights to the Bregman generators (Proposition~\ref{prop:bpdequivpd})
and that conversely power distances can be interpreted as arbitrary Bregman power divergences provided some conditions on the domains of Bregman generators are satisfied (Proposition~\ref{prop:pd2bpd}).
We then show how to interpret those Bregman power divergences in \S\ref{sec:interpretsdbpd} either as a vertical gap read from the graph of the potential function (Figure~\ref{fig:bpd}) in
 \S\ref{sec:interpretdsbpd} or from an information-geometric point of view in \S\ref{sec:iginterp} (Figure~\ref{fig:bpdeta}).
The next section \S\ref{sec:NN} will demonstrate a first use of the equivalence between Bregman power divergences and power distances for solving efficiently Bregman  proximity queries and some related problems: Namely, finding the Bregman diameter  and the Bregman closest pair in a set of parameters.

%%%%%%%%%%%%
\subsection{Power distances}\label{sec:pd} % : Single-sided and Double-sided
%%%%%%%%%%%%

\begin{Definition}[\cite{aurenhammer1987power}]\label{def:pd}
The power distance between a point $x\in\bbR^m$ and a weighted point $\hatp=(p,w)$ (with $p\in\bbR^m$ and weight $w\in\bbR$) is defined by
\begin{equation}\label{eq:sspd}
\sigma(x,\hatp) \eqdef \|x-p\|^2-w.
\end{equation}
\end{Definition}

Note that although this dissimilarity is called a power distance, it is not a proper metric distance 
since  it can be potentially negative and  violate the triangle inequality of metric distances.

Power distances have been widely  used in computational geometry for 3D mesh reconstruction from point clouds~\cite{cheng2013delaunay} and in topological data analysis~\cite{PowerMEB-ESA-2024} (TDA)
 where they have also been extended to infinite-dimensional reproducible kernel Hilbert spaces (RKHSs).
Power distances also play important roles in semi-discrete optimal transport~\cite{levy2015numerical} and when studying neural networks~\cite{balestriero2019geometry} and Large Language Models~\cite{ma2026laguerre} (LLMs).

When the weight is non-negative, one can visualize the power distance by representing weighted points $\hatp$ as corresponding balls/spheres as follows:
Let $B(\hatp)=\ball(p,\sqrt{w})$ where $\ball(c,r)\eqdef\{x \st \|c-x\|\leq r\}$ denotes a closed Euclidean ball of center $c$ and radius $r$.
Denote by $\partial B$ the boundary sphere of a ball $B$ and by $B^\circ$ the open ball $B\backslash \partial B$.
The power distance $\sigma(x,\hat p)$ vanishes when $x\in S(\hatp)=\partial B(\hatp)$, amounts 
to the square length of the line segments tangent to $x$ and $S(\hatp)$ by Pythagoras' theorem (Figure~\ref{fig:vizPD}), 
and is negative when $x\in B^\circ(\hatp)$.
 
A point $t\in S(\hatp)$ defines a tangent line segment passing through $x$ when $(p-t)\cdot (x-t)=0$ where $\cdot$ denotes the scalar product (dot product, the Euclidean inner product).
Let $\calT\eqdef\{t \ : \ (p-t)\cdot (x-t)=0 \}$ denote the set of tangent points of $S(\hatp)$ passing through $x$.
$\calT$ is a $(m-1)$-dimensional sphere of $\bbR^m$ (thus a $0$-sphere consisting of two points when $m=2$).
To summarize, we have
\begin{equation}
\sigma(x,\hatp) = \left\{
\begin{array}{ll}
\|x-t\|^2>0, & x\not\in B(\hatp), t\in\calT,\\
=0, & x\in S(\hatp),\\
<0, & x\in B^\circ(\hatp).
\end{array}
\right.
\end{equation}

\begin{Remark}
When $w<0$, we may consider the sphere $S(\hatp)$ with imaginary radius~\cite{berger1984space,devillers1992space} $\sqrt{w}=i\sqrt{|w|}\in\bbC$, 
and every point $x$ falls outside the ball $B(\hatp)$.
\end{Remark}

\begin{Remark}
A  general geometric interpretation for the power distance $\sigma(x,\hatp)$ in 2D is reported in~\cite{devillers1992space} when $x\not\in B(\hatp)$:
Let $(p_1p_2)$ be a line passing through $x$ and intersecting $S(\hatp)$ in the two points $p_1$ and $p_2$. Then we have
$\sigma(x,\hatp)=(p_1-x)\cdot(p_2x)=\|p_1-x\| \, \|p_2-x\|$.
Thus when $p_1=p_2=t$, we get $\sigma(x,\hatp)=\|x-t\|^2$.
\end{Remark}

\begin{figure} 
%\centering
\includegraphics[width=0.35\textwidth]{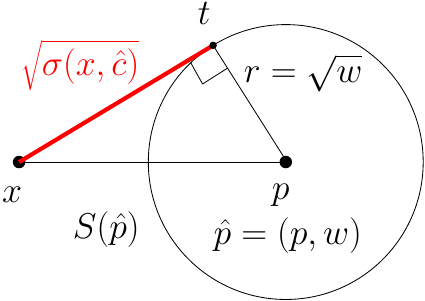}
 
\caption{Visualizing the power distance of a point $x$ to a non-negative weighted point $\hatp =(p,w)$ represented as a sphere $S(\hatp)$.\label{fig:vizPD}}
\end{figure}

The power distance of Eq.~\ref{eq:sspd} can further be extended to both weighted point arguments as follows:

\begin{Definition}[(double-sided power distance)~\cite{cheng2013delaunay}]\label{def:dspd}
The power distance between two weighted points $\hatp=(p,w)$ and $\hatp'=(p',w')$ is defined by
\begin{equation}\label{eq:dspd}
\sigma(\hatp,\hatp') \eqdef \|p-p'\|^2-w-w'.
\end{equation}
\end{Definition}

\begin{Remark}[Offsetting weights]
In algorithms, one often uses predicates which rely on power distances like checking whether $\sigma(\hatp_1,\hatp_2)<\sigma(\hatp_3,\hatp_4)$? or not.
By adding a constant $C\geq 0$ to the power distance $\sigma$, we do not change the predicate signs and  we may consider the weights $w$ and $w'$ non-negative 
provided that we chose $C$ sufficiently large enough (say, for a prescribed weighted point set).
Thus by offsetting all poiny weights by $C/2$, we may visualize the offset weighted points as corresponding spheres.
\end{Remark}

In the remainder, we shall term the power distance of Eq.~\ref{eq:sspd} the \termdef{single-sided power distance} and the power distance of Eq.~\ref{eq:dspd} the  \termdef{double-sided power distance}.
Note that $\sigma(\hatp,\hatp')$ is not anymore related to the lengths of the internal or external tangent line segments of $S(\hatp)$ and $S(\hatp')$ when $w,w'>0$.
A visualization of the double-sided power distance relying on the paraboloid lifting transform~\cite{cheng2013delaunay} is postponed to \S\ref{sec:bpd}.
Notice that the double-sided power distance can be turned into an equivalent single-sided power distance:
$\sigma(\hatp,\hatp')=\sigma(p,(p',w'+w))=\sigma((p,w+w'),p')$. 
The definition of double-sided power distance proves useful when the weights are coupled to the points as the shall see in the scope of this work.

%%%%%%%%%%%%
\subsection{Bregman divergences and convex duality}\label{sec:bd}
%%%%%%%%%%%%%

Let $F$ be a strictly convex and differentiable function defined on an open convex domain $\Theta\subset\bbR^m$ equipped with an inner product  (e.g., vector inner product $\inner{u}{v}\eqdef \sum_{i=1}^m u_iv_i$).
Let $\nabla F(\theta)$ denote the gradient vector at $\theta$ (i.e., column vector of partial derivatives $\left[\frac{\partial}{\partial{\theta_i}} F(\theta)\right]_i$).

\begin{Definition}[\cite{Bregman-1967}]\label{def:bd}
The Bregman divergence $B_F(\theta_1:\theta_2)$ between parameter $\theta_1$ and parameter $\theta_2$ of $\Theta$ is defined by
\begin{equation}
B_F(\theta_1:\theta_2) \eqdef F(\theta_1)-F(\theta_2)-\inner{\theta_1-\theta_2}{\nabla F(\theta_2)}.
\end{equation}
\end{Definition}
The ``:'' notation delimiting arguments introduced by Amari~\cite{IG-2016} emphasizes the fact that $B_F(\theta_1:\theta_2)\not= B_F(\theta_2:\theta_1)$ except for quadratic generators $F(\theta)=P\theta$ where $P$ is a symmetric positive-definite matrix~\cite{BVD-2007}.
In information theory~\cite{cover1999elements}, the delimiter is chosen as ``$\parallel$'' instead of the comma ``,'' in order to avoid confusion with joint random variables separated by a comma.

Bregman divergences (BDs) unify the squared quadratic Euclidean distance (obtained with $F(\theta)=\inner{\theta}{\theta}$)
with relative entropies like the Kullback-Leibler divergence (obtained with $F(\theta)=\sum_i \theta_i\log\theta_i$, Shannon negentropy) or the Itakura-Saito divergence (obtained with $F(\theta)=-\sum_i\log\theta_i$, Burg negentropy).
Smooth families of Bregman generators $F_\alpha(\theta)$ and $F_\beta(\theta)$ define the prominent classes of Bregman  $\alpha$-divergences~\cite{amari2009alpha} and $\beta$-divergences~\cite{hennequin2010beta}. % which were characterized axiomatically by Csisz\'ar~\cite{csiszar1991least}.

Denote by $\calH\eqdef\{\eta(\theta)=\nabla F(\theta) \st \theta\in\Theta\}$ the gradient domain.
The Bregman generator $F$ is said of Legendre-type~\cite{rockafellar1997convex} when 
$\|\nabla F(\theta)\|\rightarrow +\infty$ whenever $\theta\rightarrow\partial\Theta$.
For Legendre-type generators, the Bregman divergence $B_F(\theta:\theta')$ amounts equivalently to a dual Bregman divergence on corresponding parameters of $\calH$ after swapping the argument order: 
$$
B_F(\theta:\theta')=B_{F^*}(\eta':\eta),
$$
where the convex conjugate function $F^*$ is defined on domain $\calH$ as follows:
$$
F^*(\eta)=\inner{(\nabla F)^{-1}(\eta)}{\eta}-F((\nabla F)^{-1}(\eta)).
$$

We thus get a dual parameterization of parameters for expressing Bregman divergences: $\eta(\theta)=\nabla F(\theta)=(\nabla F^*)^{-1}(\eta)\in\calH$ 
and $\theta(\eta)=\nabla F^*(\eta)=(\nabla F)^{-1}(\eta)\in\Theta$. 
The conversion functions $\eta(\cdot)$ and $\theta(\cdot)$ are diffeomorphisms between domains $\Theta$ and $\calH$, and $\calH$ and $\Theta$, respectively. 
Furthermore, we have the involutive biconjugation property for Legendre-type functions: ${F^*}^*=F$ (Fenchel-Moreau's theorem).
When $F$ is strictly convex and smooth but not of Legendre-type, $\calH$ may not be convex and hence $F^*$. See~\cite{LegendreType-1967} for such a case study.

It follows from the convex duality of Legendre-type functions that we can rewrite a Bregman divergence as an equivalent Fenchel-Young divergence~\cite{acharyya2013learning} using the mixed primal-dual $\theta$-parameterization and $\eta$-parameterization:
$$
B_F(\theta:\theta')=Y_F(\theta:\eta'),
$$
where $\eta'=\nabla F(\theta')$ and the corresponding Fenchel-Young divergence $Y_F$ is defined as follows:

\begin{Definition}[\cite{acharyya2013learning}]\label{def:fyd}
The Fenchel-Young divergence between a parameter $\theta\in\Theta$ and a parameter $\eta'\in\calH$ is defined by:
\begin{equation}\label{eq:fyd}
Y_F(\theta:\eta') \eqdef F(\theta)+F^*(\eta')-\inner{\theta}{\eta'}.
\end{equation}
\end{Definition}

In information geometry~\cite{IG-2016}, dual Bregman divergences and dual Fenchel-Young divergences are coordinatized expressions of the canonical geometric divergences between points lying on dually flat spaces (also called Hessian manifolds~\cite{shima2007geometry} or Bregman manifolds).

%%%%%%%%%%%%
\subsection{Bregman divergences as power distances}\label{sec:bdpd}
%%%%%%%%%%%%

Our first result is to rewrite the Fenchel-Young divergence (and hence its corresponding dual Bregman divergences) as an equivalent half double-sided power distance (PD) between two mixed-parameterized weighted points:

\begin{Property}\label{prop:equivydpd}
The Fenchel-Young divergence $Y_F(\theta:\eta')$ between $\theta\in\Theta$ and $\eta'\in\calH$ amounts to an equivalent half (double-sided) power distance:
\begin{equation}\label{eq:mpeucl}
Y_F(\theta:\eta') = \frac{1}{2}\, \sigma(\wtheta,\weta'),
\end{equation} 
where $\wtheta\eqdef (\theta,\varpi=\|\theta\|^2 - 2\, F(\theta))$, $\weta'\eqdef(\eta',\varpi'=\|\eta'\|^2 - 2\, {F^*}(\eta'))$.
\end{Property}

We use the notation $\varpi$ instead of $w$ to emphasize that the weights are coupled to the parameter arguments.
 
\begin{proof}
Since $\|\theta-\eta'\|^2=\|\theta\|^2+\|\eta'\|^2-2\inner{\theta}{\eta'}$, we replace the inner product term
$-\inner{\theta}{\eta'}$ in the expression of the Fenchel-Young divergence of Eq.~\ref{eq:fyd}  by 
$\frac{1}{2} \|\theta-\eta'\|^2 - \frac{1}{2}\|\theta\|^2-\frac{1}{2}\|\eta'\|^2$, and we identify the expression with the Eq.~\ref{eq:dspd} of the double-sided power distance.
%\end{IEEEproof} 
\end{proof}
 
Let $w_F(\theta)\eqdef\|\theta\|^2 - 2\, F(\theta)$ and $w_{F^*}(\eta)\eqdef\|\eta\|^2 - 2\, F^*(\eta)$ be the  parameter to weight coupling functions.
We use the circumflex accent notation in  $\wtheta$ and its  upside down notation in $\weta$ to distinguish the dual mechanisms to lift parameters to weighted points using an extra dimension. %Both $\wtheta$ and $\weta$ belong to $\hat\bbR^{m+1}=\bbR^m\times\bbR$.
That is, we have $\wtheta=(\theta,w_F(\theta))\in\bbR^{m+1}$ and $\weta=(\eta,w_{F^*}(\eta))\in\bbR^{m+1}$.

It follows that we  can express Bregman divergences as equivalent half power distances:
\begin{equation}\label{eq:bdpdformula}
B_F(\theta:\theta') = \frac{1}{2} \sigma\left((\theta,w_F(\theta)),(\nabla F(\theta'),w_{F^*}(\nabla F(\theta')))\right).
\end{equation}

\begin{Remark}
Bregman divergences have been defined for other non-vector type data like for matrices~\cite{dhillon2008matrix} or for functions (via Fr\'echet derivatives in Lebesgue spaces~\cite{frigyik2008functional} or in Hilbert spaces~\cite{tsuchida2026generalising}) where the inner product used to define them depend on the type of arguments. We can similarly define generic power distances in inner product spaces equipped 
with an inner product $\inner{\cdot}{\cdot}_V$ in a space $V$ as follows:
\begin{eqnarray}
\sigma(\hat{p},\hat{p}') &\eqdef& \inner{p-p'}{p-p'}_V - w - w',\\
&=& \|p-p'\|_V^2 -w -w',
\end{eqnarray}
where the corresponding norm is indued by the innerproduct: $\|p\|_V=\sqrt{\inner{p}{p}_V}$.
For example, consider the space of symmetric positive-definite matrices equipped with the Frobenius inner product 
$\inner{P}{P'}_F\eqdef\tr(P^\top P')$ with matrix Bregman generators~\cite{dhillon2008matrix}. 
One such matrix Bregman divergence is the von Neumann (vN) divergence $D_\vN$ (quantum relative entropy) induced by 
the generator $F_\vN(X)\eqdef\tr(X\log X-X)$ with dual parameterization $Y=\nabla F_\vN(X)=\log X$, $\nabla F^*_\vN(Y)=(\nabla F)^{-1}(Y)=\exp(Y)$ and convex conjugate $F^*_\vN(Y)=\tr(\exp(Y))$.
We have
\begin{eqnarray*}
D_\vN(P:P') &=& B_{F_\vN}(P:P') = \tr(P(\log P-\log P')-P+P',\\
&=& B_{F_\vN^*}(\log P':\log P),\\
&=& \frac{1}{2} \sigma(\hat{P},\check{P'}),
\end{eqnarray*}
where $B_{F_\vN^*}(Y:Y')=\tr(\exp(Y)-\exp(Y')-\exp(Y')(Y-Y'))$ 
is the dual Bregman divergence and $\hat{P}=(P,w_{F_\vN}(P)=\|P\|_F^2 -2\, F_\vN(P))$ 
and
 $\check{P'}=(\log P',w_{F_\vN^*}(P')=\|\log P'\|_F^2- 2\, F_\vN^*(\log P')$.
The matrix power distance between $\hat{X}=(X,w)$ and $\hat{X'}=(X,w')$ is defined by $\sigma(\hat{X},\hat{X'})=\|X-X'\|_F^2-w-w'$.
\end{Remark}

\begin{Remark}
The Kullback-Leibler divergence (KLD) between two densities $p_{\theta_1}$ and $p_{\theta_2}$ of a  exponential family with cumulant function $F(\theta)$ amounts to a reverse Bregman divergence between their parameters~\cite{CBD-2005}: $D_\KL(p_{\theta_1}:p_{\theta_2})=B_F(\theta_2:\theta_1)$.
For example, consider the exponential family of zero-centered Gaussian distributions with cumulant function $F(\theta)=-\log \det \theta$ (matrix Burg negentropy~\cite{davis2006differential}) for $\theta$ belonging to the $m$-dimensional symmetric positive-definite cone. The gradient map is $\eta(\theta)=\nabla F(\theta)=-\theta^{-1}$. 
We have $B_F(\theta_1:\theta_2)=\tr(\theta_1\theta_2^{-1})-\log \det(\theta_1\theta_2^{-1}) -m$ (matrix version of the Itakura-Saito divergence).
Hence, the KLD between  $p_{\theta_1}$ and $p_{\theta_2}$ can be expressed as an equivalent power distance.
\end{Remark}

\begin{Remark}
The rewriting inner product ``trick'' $-\inner{\theta}{\eta'}=\frac{1}{2} \|\theta-\eta'\|^2 - \frac{1}{2}\|\theta\|^2-\frac{1}{2}\|\eta'\|^2$   was used in optimal transport to convert Bregman-Wasserstein distances~\cite{BW-2025} into equivalent ordinary Wasserstein distances.
The Wasserstein distance~\cite{gabriel2019computational} $\calW_2(\mu,\mu')$ between two probability measures $\mu$ and $\mu'$  is defined as
$$
\calW_2^2(\mu,\mu') \eqdef \inf_{\pi\in\Pi(\mu,\mu')} \int_{\bbR^m\times\bbR^m} \|x-x'\|^2 \, \mathrm{d}\pi(x,x'),
$$
where $\Pi(\mu,\mu')$ denotes the set of transport plans with marginal distributions $\mu$ and $\mu'$.
The Bregman-Wasserstein distance~\cite{BW-2025} is defined similarly by replacing the squared Euclidean cost function by a Bregman divergence cost: 
$$
\calW_F(\mu,\mu')\eqdef \inf_{\pi\in\Pi(\mu,\mu')} \int_{\Theta\times\Theta} B_F(\theta:\theta') \, \mathrm{d}\pi(\theta:\theta').
$$
By plugging the equivalence of BDs with PDs, it holds that
\begin{eqnarray*}
\calW_F(\mu,\mu') &=& \frac{1}{2}\left(
\inf_{\pi\in\Pi(\mu,\nu')}  \int_{\Theta\times\calH} \|\theta-\eta'\|^2\, \mathrm{d}\pi(\theta:\eta')
- \int_\theta w_F(\theta)\, \mathrm{d}\mu - \int_\eta w_{F^*}(\eta)\, \mathrm{d}\nu'
\right),\\
&=& \frac{1}{2} \, \left( \calW_2^2(\mu,\nu') - W_F(\mu) -W_{F^*}(\nu') \right),
\end{eqnarray*}
where $\nu'$ is the push-forward measure of $\mu'$ via the gradient map $\nabla F$. 
Since $W_F(\mu)\eqdef \int_\theta w_F(\theta)\, \mathrm{d}\mu$ and $W_{F^*}(\nu')\eqdef \int_\eta w_{F^*}(\eta)\, \mathrm{d}\nu'$ are constants, minimizing the coupling in the Bregman-Wasserstein distance $\calW_F$ amounts to minimize the coupling in the $2$-Wasserstein $\calW_2$ distance between $\mu$ and $\nu'$.
\end{Remark}

In this work, we make the formal connection between Bregman divergences and half power distances (Eq.~\ref{eq:bdpdformula}) which emphasize the role of the inner products, and discuss several applications bridging information geometry with computational geometry.

Notice that although the weights of $\wtheta$ and $\weta$ may be arbitrary in $\bbR$, we get the following property:

\begin{Property}\label{prop:dualweightprop}
The sum $w_F(\theta)+w_{F^*}(\eta')$ of mixed-parameterized weights  satisfy the following inequality:
\begin{equation}
w_F(\theta)+w_{F^*}(\eta') \leq \|\theta-\eta'\|^2,
\end{equation}
with equality if and only if $\eta'=\nabla F(\theta)=\eta$.
\end{Property}
 
\begin{proof}
By Fenchel-Young inequality, we have $F(\theta)+F^*(\eta')\geq \inner{\theta}{\eta'}$ (with equality if and only if $\eta'=\nabla F(\theta)$).
Therefore it holds that
\begin{eqnarray*}
w_F(\theta)+w_{F^*}(\eta') &=& \|\theta\|^2-2\, F(\theta)+\|\eta'\|^2-2\, F^*(\eta'),\\
&\leq& \|\theta\|^2+\|\eta'\|^2-2\inner{\theta}{\eta'},\\
&\leq& \|\theta-\eta'\|^2.
\end{eqnarray*}
The inequality is tight when $\eta'=\nabla F(\theta)=\eta$ since we have $w_F(\theta)+w_{F^*}(\eta)=\|\theta\|^2+\|\eta\|^2-2(F(\theta)+F^*(\eta))=\|\theta\|^2+\|\eta\|^2-2\inner{\theta}{\eta}=\|\theta-\eta\|^2$.

Another short proof is to use the fact that $B_F(\theta:\theta')\geq 0$ and  since
$B_F(\theta:\theta')=\frac{1}{2} \sigma(\wtheta,\weta')=\frac{1}{2}  \left( \|\theta-\eta'\|^2-(w_F(\theta)+w_{F^*}(\eta')) \right) \geq 0$, we conclude that 
necessarily we have $\|\theta-\eta'\|^2\geq w_F(\theta)+w_{F^*}(\eta')$.
\end{proof}

We can extend the equivalence of Bregman/Fenchel-Young divergence to duo Bregman/Fenchel-Young pseudo-divergences~\cite{DuoBD-2022}:

\begin{Definition}[\cite{DuoBD-2022}]\label{def:duobdfd}
Let $F_1$ and $F_2$ be two Bregman generators such that $F_1\geq F_2$ (i.e., $F_2$ majorizes pointwise $F_1$).
The duo Bregman pseudo-divergence is defined by
\begin{equation}\label{eq:duobd}
B_{F_1,F_2}(\theta:\theta') \eqdef F_1(\theta)-F_2(\theta')-\inner{\theta-\theta'}{\nabla F_2(\theta')}.
\end{equation}
A duo Bregman pseudo-divergence amounts to a corresponding
 duo Fenchel-Young pseudo-divergence: $B_{F_1,F_2}(\theta:\theta') = Y_{F_1,F_2^*}(\theta:\eta')$ where $\eta'=\nabla F_2(\theta')$ and 
\begin{equation}\label{eq:duofd}
Y_{F_1,F_2}(\theta:\eta') \eqdef F_1(\theta)+F_2^*(\eta')-\inner{\theta}{\eta'}.
\end{equation}
\end{Definition}

Those dissimilarities are called pseudo-divergences because when $F_1>F_2$ (everywhere strict pointwise majorization), we have $B_{F_1,F_2}(\theta:\theta)>0$ so that the identity of the indiscernibles of divergences (smooth potentially asymmetric distances) is not satisfied.
Bregman pseudo-divergences can be used to get closed-form formula for the Kullback-Leibler divergence between two densities of nested truncated exponential families~\cite{DuoBD-2022} like truncated Gaussian distributions.

Let $\Theta_1$ and $\Theta_2$ denote the domains of $F_1$ and $F_2$, and $\calH_1$ and $\calH_2$ the corresponding gradient domains by obtained by the gradient maps
 $\nabla F_1$ and $\nabla F_2$, respectively.

\begin{Proposition}\label{prop:duobdpd}
We have for $\theta\in\Theta_1$ and $\eta'\in\calH_2$,
\begin{equation}
Y_{F_1,F_2}(\theta:\eta') = \frac{1}{2}\, \sigma\left( (\theta,w_{F_1}(\theta)) , (\eta',w_{F_2^*}(\eta')) \right).
\end{equation}
\end{Proposition}

Let $Q(\theta)=\frac{1}{2}\, \inner{\theta}{\theta}=\frac{1}{2}\|\theta\|^2$ denote the half paraboloid of revolution and let 
$Q_2(\theta)=2Q(\theta)=\inner{\theta}{\theta}$ be the paraboloid of revolution.
It is known that the only function that is self-conjugate~\cite{bauschke2012fenchel} is $Q$, i.e.,  $F=F^* \Leftrightarrow F=Q$.

\begin{Property}\label{prop:alternatingsign}
The weights $w_{F}(\theta)$ and $w_{F^*}(\eta)$ for $\eta=\nabla F(\theta)$ or $\theta=\nabla F^*(\eta)$ are of opposite signs when $F>Q$ or $F<Q$.
\end{Property}

\begin{proof}
Since the Legendre-Fenchel transform reverses order, i.e., $F_1\leq F_2 \Leftrightarrow F_2^*\leq F_1^*$, 
it follows that whenever we have $F\geq Q$ (i.e., $\forall \theta, F(\theta)\geq Q(\theta)=\frac{1}{2}\|\theta\|^2$) we
 also  have $F^*\leq Q^*=Q$ $F\geq Q$ (i.e., $\forall \eta, F(\eta)\leq Q(\eta)=\frac{1}{2}\|\eta\|^2$).
Thus we have $w_F(\theta)\geq 0 \Leftrightarrow w_{F^*}(\eta)\leq 0$ and $w_F(\theta)\leq 0 \Leftrightarrow w_{F^*}(\eta)\geq 0$ (opposite signs).
\end{proof}

%%%%%%%%%%%%
\subsection{Bregman power divergences as power distances}\label{sec:bpd}
%%%%%%%%%%%%

In this section, we shall describe a visualization of the double-sided power distance in a more general setting by generalizing the power distances to power Bregman divergences as follows:

\begin{Definition}\label{def:bpd}
The double-sided power Bregman divergence between two weighted parameters $\hat\theta=(\theta,w)$ and $\hat\theta'=(\theta',w')$ is defined by:
\begin{equation}\label{eq:bpd}
\sigma_F(\hat\theta: \hat\theta') \eqdef B_F(\theta:\theta') - w-w'.
\end{equation} 
\end{Definition}
When $F(\theta)=\inner{\theta}{\theta}=Q_2(\theta)$, we recover the ordinary power distance: 
$\sigma(\hat\theta: \hat\theta')=\sigma_{Q_2}(\hat\theta: \hat\theta')$.

\begin{Property}\label{prop:bpdequivpd}
The  power Bregman divergence $\sigma_F(\hat\theta:\hat\theta')$ amounts to an ordinary half power distance: 
$$
\sigma_F(\hat\theta: \hat\theta') =  \frac{1}{2}\, \sigma(\wwtheta,\wweta'),    
$$
where $\wwtheta=(\theta, 2w+w_F(\theta))$ and $\wweta'=(\eta, 2w'+w_{F^*}(\eta'))$.
\end{Property}

\begin{proof}
We have
\begin{eqnarray*}
\sigma_F(\bar\theta:\bar\theta')&=& B_F(\theta:\theta')- w-w' ,\\
&=& \frac{1}{2} \sigma(\wtheta,\weta')-w-w',\\
&=&  \frac{1}{2} \left( \|\theta-\eta'\|^2 -w_F(\theta)-w_{F^*}(\eta') -2w -2w'\right) ,\\
&=& \frac{1}{2}\, \sigma(\wwtheta,\wweta').
\end{eqnarray*}
%Notice that we use the double circumflex accent $\hat{\hat{\cdot}}$ and  the reverse double circumflex accent  $\check{\check{\cdot}}$  to indicate that we transform parameter weights into other parameter weights using the dual Bregman generators.
\end{proof}

Conversely, we may interpret ordinary double-sided power distances $\sigma(\hatp,\hatp')$ (i.e., Bregman power divergences with quadratic generator $Q_2$) 
 with $\hatp=(p,\pi)$ and $\hatp'=(p',\pi')$ 
as equivalent Bregman power divergences $\sigma_F(\wwtheta,\wweta')$ provided some conditions on the domains of the dual Bregman generators are satisfied.
Here, $\wwtheta=(\theta,w+\varpi)$ where $\varpi=w_F(\theta)$ and similarly $\wweta'=(\eta',w'+\varpi')$ where $\varpi'=w_{F^*}(\eta')$.
Indeed, when both $p\in\Theta$ and $p'\in\calH$,  we let $\theta=p$ and $\eta'=p'$ and we have

\begin{eqnarray*}
\sigma(\hatp,\hatp') &=& \sigma_F(\wwtheta,\wweta'),\\
 \|p-p'\|^2-\pi-\pi' &=& \|\theta-\eta'\|^2-w_F(\theta)-w-w_{F^*}(\eta')-w',
\end{eqnarray*}
with $\pi=w_F(\theta)+w=\varpi+w$ and $\pi'=w_{F^*}(\eta')+w'=\varpi'+w'$.
Hence, we let $w \eqdef \pi-w_F(\theta)$ and $w'\eqdef \pi'-w_{F^*}(\eta')$ to obtain the equivalence.

\begin{Proposition}\label{prop:pd2bpd}
The power distance $\sigma(\hatp,\hatp')$ with  $\hatp=(p,\pi)$ and $\hatp'=(p',\pi')$  can be interpreted as a Bregman power divergence $\sigma_F(\wwtheta:\wwtheta')$ 
provided that $p\in\Theta$ and $p'\in\calH$ by letting $\wwtheta=(p,w+w_F(p))$ and $\wwtheta'=(\nabla F^*(p'),w'+w_{F^*}(p'))$, and 
 $w=\pi-w_F(\theta)$ and $w'=\pi'-w_{F^*}(\eta')$.
\end{Proposition}

\begin{Remark}
In information geometry, the primal coordinate system $\theta$ may be reparameterized by an arbitrary affine map and the Bregman generator can also be defined up to another affine term. These degrees of freedom are called gauge transformation since the transformed dual Bregman/Fenchel-Young divergences amount to the same geometric dually flat divergence $\calD$. See~\cite{nielsen2025generalized} for details.  
\end{Remark}

%%%%%%%%%%%%
\subsection{Interpretation of double-sided Bregman power divergences}\label{sec:interpretdsbpd}
%%%%%%%%%%%%

Let us give a geometric interpretation of the Bregman power divergences which includes the ordinary double-sided power distance.
For a weighted parameter $\hat\theta=(\theta,w)$ (not necessarily a coupled weight), let 
\begin{equation}
\hat\theta^F \eqdef F(\theta)-w.
\end{equation}
 
Define the support hyperplane of the potential function graph $\calF\eqdef\{(\theta,F(\theta)) \ \theta\in\Theta\}$ at $\theta$ shifted by $w\in\bbR$:
\begin{equation}
h_{\hat\theta}^\calF(x) = \inner{x-\theta}{\nabla F(\theta)}+F(\theta)+w.
\end{equation}

We can also rewrite this function using the dual parameterization $\eta=\nabla F(\theta)$ (with $\hat\eta=(\eta,w)$) as follows:
\begin{equation}
h_{\hat\eta}^\calF(x) = \inner{x}{\eta}-F^*(\eta)+w.
\end{equation}

Equipped with these notations, we have the following proposition:

\begin{Proposition}\label{prop:bpdverticalgap}
The  double-sided Bregman power divergence $\sigma_F(\hat\theta_1:\hat\theta_2)=B_F(\theta_1:\theta_2)-w_1-w_2$ between two weighted parameters $\hat\theta_1=(\theta_1,w_1)$ 
and $\hat\theta_2=(\theta_2,w_2)$ can be expressed as the vertical gap
\begin{equation}
\sigma_F(\hat\theta_1:\hat\theta_2) = \hat\theta_1^F-h^\calF_{\hat\theta_2}(\theta_1).
\end{equation} 
\end{Proposition}

\begin{proof}
\begin{eqnarray*}
\hat\theta_1^F-h_{\hat\theta_2}^\calF(\theta_1) &=& F(\theta_1)-w_1-(\inner{\theta_1-\theta_2}{\nabla F(\theta_2)}+F(\theta_2)+w_2),\\
&=& B_F(\theta_1:\theta_2)-w_1-w_2,\\
&=:& \sigma_F(\hat\theta_1:\hat\theta_2).
\end{eqnarray*}
\end{proof}

Figure~\ref{fig:bpd} visualizes the vertical gap $\sigma_F(\hat\theta_1:\hat\theta_2)$ from the graph $\calF$ and the translated tangent hyperplane 
of $\calF$ at $\theta_2$.

\begin{figure}%
\includegraphics[width=0.75\columnwidth]{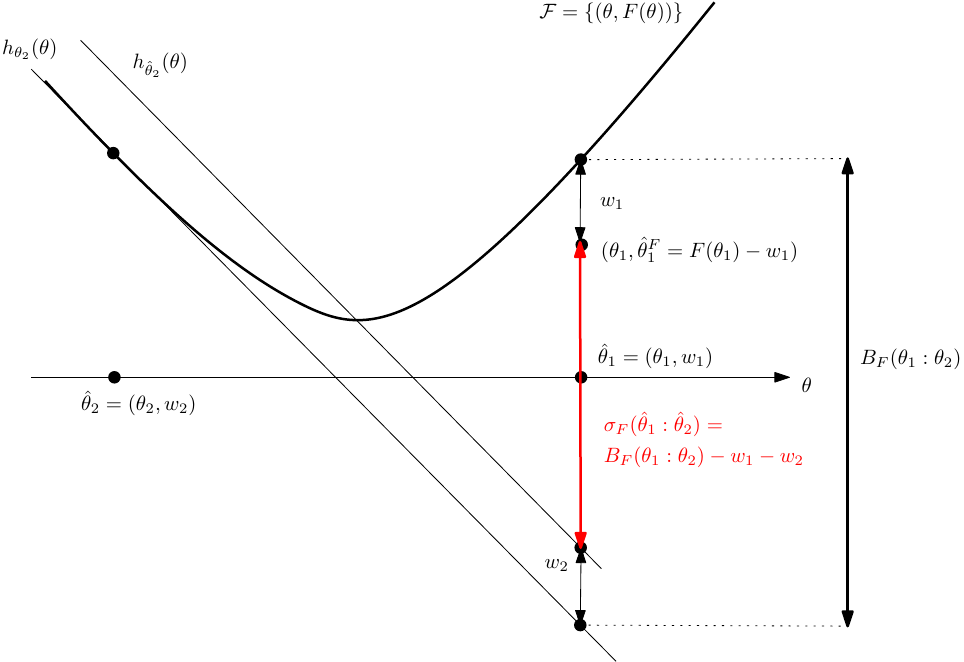}%

\caption{Visualizing the double-sided Bregman power divergence as a vertical gap (red) from the potential function.}%
\label{fig:bpd}%
\end{figure}

Notice that when $F(\theta)=\inner{\theta}{\theta}=Q_2(\theta)$ (the paraboloid) with $\eta=\nabla F(\theta)=2\theta$ and $F^*(\eta)=\frac{1}{4}\inner{\eta}{\eta}$, we get
\begin{eqnarray*}
\sigma_{Q_2}(\hat\theta_1:\hat\theta_2) &=& \sigma(\hat\theta_1,\hat\theta_2),\\
&=& \hat\theta_2^{Q_2}-h_{\hat\theta_1}(\theta_2),\\
&=& \hat\theta_1^{Q_2}-h_{\hat\theta_2}(\theta_1),\\
&=& \sigma_{Q_2}(\hat\theta_2:\hat\theta_1) = \sigma(\hat\theta_2,\hat\theta_1),
\end{eqnarray*}
where $h_{\hat\theta_1}(\theta_2)=2\inner{x-\theta}{\theta}+\inner{\theta}{\theta}+w$.

Thus we get an alternative visualization of the double-sided power distance as a vertical gap in $\bbR^{m+1}$ 
(compare with  the visualization of the single-sided power distance in Figure~\ref{fig:vizPD}).

%%%%%%%%%%%%%%%
\subsection{Interpretation of single-sided Bregman power divergences}\label{sec:interpretsdbpd}
%%%%%%%%%%%%%%%

We first recall some basics of Bregman manifolds in \S\ref{sec:bm} to introduce Bregman Pythagorean theorem,  
and then give an interpretation of the single-sided Bregman power divergence in \S\ref{sec:iginterp}.

%%%%%%%%%%%%%%%
\subsubsection{Dualistic information geometry of Bregman manifolds}\label{sec:bm}
%%%%%%%%%%%%%%%%
Although information geometry~\cite{IG-2016} historically considered differential geometric structures for statistical models with applications, the field
 discovered novel dualistic geometric structures that can also be used in other contexts than statistics like in time-series, in control systems, in mathematical programming, etc. 
The dually flat spaces~\cite{IG-2016} (studied in pure differential geometry in the broader scope of Hessian manifolds~\cite{shima2007geometry}) is such can example of information-geometric structures that can be built from convex functions.

In order to give an interpretation of the single-sided Bregman power divergences generalizing the interpretation of single-sided power distances, we concisely describe the dually flat spaces induced by Bregman generators/divergences without the burden of differential geometry (i.e., by hiding the concept of affine connections). The domain $\Theta$ 
represents the $\theta$-coordinates of points lying on a Bregman manifold $M$ which can also be indexed by the dual $\eta$-coordinates.
On $M$,  a Riemannian metric $g$ is expressed in the dual $\theta/\eta$ coordinate systems as either $\nabla^2 F(\theta)=[g_{ij}(\partial_i,\partial_j)]_{ij}$ or as 
$\nabla^2 F^*(\eta)=[g^{ij}(\partial^i,\partial^j)]_{ij}=(\nabla^2 F(\theta))^{-1}$, where  $\partial_i = \frac{\partial}{\partial\theta_i}$ and 
$\partial^i=\frac{\partial}{\partial\eta_i}$.
A vector $v$ in  a tangent plane $T_p M$ at point $p\in M$ can be expressed in the natural basis 
$\{ \partial_i = \frac{\partial}{\partial\theta_i} \mid_p\}$
 as $v=\sum_i v^i \partial_i$ or in the reciprocal basis  $\{\partial^i=\frac{\partial}{\partial\eta_i}\mid_p\}$ as $v=\sum_i v_i \partial^i$.
The vector components  $[v^i]$ and $[v_i]$ in the dual basis satisfy $[v_i]= \nabla^2 F(\theta(p)) [v^i]$.
The basis vectors of the dual basis are mutually orthogonal: $g(\partial_i,\partial^j)=\delta_i^j$ where the Kronecker symbol is defined by 
$\delta_i^j=\left\{\begin{array}{l}
1, i=j,\cr 0, i\not =j\end{array}\right.$. 
Two vectors $u,v\in T_p M$ are orthogonal iff. $g(u,v)=0$ or in a basis when
$[u^i]^\top \nabla F(\theta(p)) [v^i]=0$ which amounts to $[u^i] \cdot [v_i]=0$ where $\cdot$ denotes the Euclidean inner product.
Any two points $p$ and $q$ on $M$ can be either joined by a $\theta$-geodesic $\gamma_{pq}(t)$ such that $\theta(\gamma_{pq}(t))=(1-t)\theta(p)+t\theta(q)$ 
or by a $\eta$-geodesic $\gamma_{pq}^*(t)$ such that $\eta(\gamma_{pq}^*(t))=(1-t)\eta(p)+t\eta(q)$. 
That is, the dual geodesics are expressed as linear interpolations in the dual coordinate systems.
The divergence $\calD(p:q)$
 between two points $p$ and $q$ on the Bregman manifold $M$ is expressed in the dual coordinate systems as dual Bregman/Fenchel-Young divergences:
\begin{eqnarray*}
\calD(p:q) &=& B_F(\theta(p):\theta(q))=B_{F^*}(\eta(q):\eta(p)),\\
&=& Y_F(\theta(p):\eta(q))=Y_{F^*}(\eta(q):\theta(p)).
\end{eqnarray*}

(In information geometry, $\calD$ is called the dually flat divergence because it is geometrically induced by the dually flat connections $\nabla$ and $\nabla^*$ which induce the $\nabla$-affine coordinate system $\theta$ and $\nabla^*$-affine coordinate system $\eta$.)

A generalization of the Pythagoras theorem holds in dually flat spaces which follows from the following $3$-point identity:

\begin{Property}[Dual-geodesic $3$-point identity]\label{prop:3ptid}
Let $\gamma(t_1)=\gamma_{p,a}(t_1)$ and $\Gamma^*(t_2)=\Gamma^*_{p,b}(t_2)$ be dual geodesics (with $t_1,t_2\in\bbR$) emanating from point $p\in M$ 
with initial vectors $u$ and $v$ in the tangent plane $T_p M$, respectively (i.e., dual geodesics with initial value conditions).
Then we have
\begin{equation}\label{eq:threeptid}
\calD(\gamma(t_1):p)+\calD(p:\Gamma^*(t_2)) = \calD(\gamma(t_1):\Gamma^*(t_2)) + t_1t_2\, g_p(u,v).
\end{equation}
\end{Property}

\begin{proof}
To write concisely the proof, let us note for short $\theta=\theta(p)$, $\eta=\eta(p)$, $\gamma=\theta(\gamma_{p,a}(t_1))$ and $\Gamma^*=\eta(\Gamma_{p,b}^*(t_2))$.
Expressing the divergences $\calD$ using the equivalent Fenchel-Young divergences, we rewrite 
$\calD(\gamma(t_1):p)+\calD(p:\Gamma^*(t_2)) - \calD(\gamma(t_1):\Gamma^*(t_2))$ as follows:
\begin{eqnarray*}
&=& Y_F(\gamma:\eta)+Y_F(p:\Gamma^*)-Y_F(\gamma:\Gamma^*),\\
&=& F(\theta)+F^*(\eta)-\theta\cdot\eta+F(\gamma)+F^*(\eta)-\gamma\cdot\eta
-(F(\gamma)+F^*(\Gamma^*)-\gamma\cdot\Gamma^*),\\
&=& t_1\, t_2 \, [u^i]\cdot [v_j].
\end{eqnarray*}
Since $[u^i]\cdot [v_j]=g_p(u,v)$, we get the identity of Eq.~\ref{eq:threeptid}.
\end{proof}

\begin{Corollary}[Pythagoras' theorem]\label{cor:BregPytha}
Let $\gamma(t_1)=\gamma_{p,a}(t_1)$ and $\Gamma^*(t_2)=\Gamma^*_{p,b}(t_2)$ be dual geodesics orthogonal at $p$.
Then we have
\begin{equation}
\calD(\gamma(t_1):p)+\calD(p:\Gamma^*(t_2)) = \calD(\gamma(t_1):\Gamma^*(t_2)),
\end{equation}
or equivalently when $q=\gamma(t_1)$ and $r=\Gamma^*_{p,b}(t_2)$:
\begin{eqnarray}
B_F(\theta_q:\theta_p) + B_F(\theta_p:\theta_r)  &=&  B_F(\theta_q:\theta_r),\\
B_{F^*}(\eta_p:\eta_q) + B_{F^*}(\eta_r:\eta_p) &=& B_{F^*}(\eta_r:\eta_q).
\end{eqnarray}
\end{Corollary}

%%%%%%%%%%%%%%% 
\subsubsection{Information-geometric interpretation}\label{sec:iginterp}
%%%%%%%%%%%%%%%

Consider the single-sided Bregman power divergence between $\theta$ and a weighted parameter $\hat\theta'=(\theta',w')$:
\begin{equation}
\sigma_F(\theta: \hat\theta') = B_F(\theta:\theta') -w'.
\end{equation} 
We have $\sigma_F(\theta: \hat\theta')=\sigma_F(\hat\theta:\theta')=\sigma_{F^*}(\eta': \hat\eta)=\sigma_{F^*}(\hat\eta': \theta)$ where $\hat\theta'=(\theta',w)$, $\hat\eta=(\nabla F(\theta),w)$ and $\hat\eta'=(\nabla F(\theta'),w)$.

Let $S=S_F^l(\hat\theta_C)\eqdef\{\theta' \st B_F(\theta_C:\theta')=w\}$ be the Bregman sphere with left-sided center $\theta_C$ and radius $w$ (provided that $w\geq 0$).
Similarly, define the Bregman sphere with right-sided center: $S_F^r(\hat\theta_C)\eqdef\{\theta' \st B_F(\theta':\theta_C)=w\}$.

\begin{Property}[Corollary 3.11~\cite{amari2000methods}]\label{prop:Amari}
Any $\theta$-geodesic passing through $\theta_C$ intersects $S$ orthogonally with respect to the Hessian Riemannian metric induced by $F$.
\end{Property}

We can interpret the Bregman power divergence by generalizing the Pythagoras interpretation of power distances in the Euclidean space as follows:
Let $\calT=\{\theta_T\}$ be the set of parameters of $S$ (i.e., $B_F(\theta_C:\theta_T)=w$) such that the $\eta$-geodesic $\gamma^*(\theta,\theta_T)$ from $\theta$ to $\theta_T$ (boundary value conditions) is orthogonal at $\Theta_T$ to the $\theta$-geodesic $\gamma(\theta_T,\theta_C)$.
That is, $\theta_T$ satisfies $(\theta_C-\theta_T)^\top (\theta_T-\theta)=0$ and we thus   the Bregman Pythagorean identity holds:
$$
B_F(\theta_C:\theta_T)+B_F(\theta_T:\theta)=B_F(\theta_C:\theta).
$$
It follows that the power Bregman divergence is given by 
$$
\sigma_F(\theta_C:\theta)=B_F(\theta_T:\theta) =  B_F(\theta_C:\theta)-w.
$$

\begin{figure}%
%\center
\includegraphics[width=0.5\textwidth]{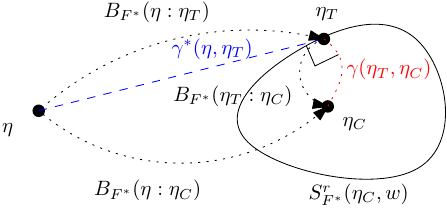}%
\caption{Interpretation of the single-sided Bregman power divergence in the $\eta$-coordinate system using the Bregman Pythagorean theorem.}%
\label{fig:bpdeta}%
\end{figure}

\begin{Proposition}\label{prop:pbd}
Let $\theta_T$ be a parameter on the left-sided Bregman sphere $S_F^l(\hat\theta_C)$ so that $\gamma^*(\theta,\theta_T) \perp_F \gamma(\Theta_T,\theta_C)$.
Then the power Bregman divergence  $\sigma_F(\theta_C:\theta)$ amounts to $B_F(\theta_T:\theta)$. 
\end{Proposition}

Notice that when $F$ is not a quadratic generator, we cannot anymore visualize Bregman divergences as lengths, hence the dashed arcs that displays only the orientations of the Bregman divergences (with arrows directed from first arguments to second arguments).

\begin{Remark}
The Bregman ball $B_F^l$ may not have a $\theta$-convex shape in the $\theta$-coordinate system (e.g., the Itakura-Saito ball obtained for Burg negentropy $F_\Burg(\theta)=-\sum_i \log \theta_i$).
But in the dual coordinate system $\eta$, those balls are $\eta$-convex.
This informally proves that in 2D, we have exactly two such tangent points $\Theta_T\in\calT$.
Figure~\ref{fig:bpdeta} illustrates this proposition in the $\eta$-coordinate system.
We omitted a similar figure in the $\theta$-coordinate system with a potentially non-convex left ball. 
\end{Remark}

%% file: sec-proximity.tex
%%%%%%%%%%%%%%%%%%%%%%%%%%%
\section{Proximity queries}\label{sec:NN}
%%%%%%%%%%%%%%%%%%%%%%%%%%%

\subsection{Outline}

The rationale of this section is that the power distance representation of Bregman divergences allows one to convert some Bregman problems into equivalent power distance problems for which we can refer to the rich literature of prior works and use the available off-the-shelf many libraries of computational geometry (e.g.,~\cite{Mount2010ANN}).
In this section, we first show how proximity queries with respect to power distances and Bregman divergences  can be handled as Euclidean proximity queries at the expense of increasing the dimension by one (Proposition~\ref{prop:proxqueryequiv}). We also consider the symmetrized Bregman proximity queries which requires doubling the dimension (Proposition~\ref{prop:sbdquery}).
We then illustrate how to use Bregman$\Leftrightarrow$Euclidean proximity queries for solving efficiently two Bregman tasks: 
Namely, the Bregman diameter problem (Proposition~\ref{prop:diam} and Corollary~\ref{cor:diam2d}) and
 the Bregman closest pair problem (Proposition~\ref{prop:cp} and Corollary~\ref{cor:cp2d}).

%%%%%%%%%%%%%
\subsection{Power distance queries}\label{sec:kNN} %
%%%%%%%%%%%%

Let $[n]=\{1,\ldots, n\}$ and $\hat\calP=\{\hatp_1=(p_1,w),\ldots, \hatp_n=(p_n,w_n)\}$ be a set of weighted points.
Given a weighted point query $\hatq=(q,\kappa)$, we seek the $k$-nearest and $k$-farthest neighbors of $\hatq$ in $\hat\calP$ with respect to the double-sided power divergence.
Wlog., consider the nearest neighbor (NN) query which reports the index $i^*$ of the NN of $\hat{q}$ in $\hat\calP$:
\begin{eqnarray*}
i^* &=& \arg\min_{i\in [n]} \sigma(\hatq,\hatp_i),\\
&=& \arg\min_{i\in [n]} \|q-p_i\|^2 - \kappa - w_i,\\
&=&  \arg\min_{i\in [n]} \|q-p_i\|^2  + W-w_i,
\end{eqnarray*}
for any $W\in\bbR$.
By choosing $W$ such that $W-w_i\geq 0$ for all $i$ (e.g., $W\geq \max_i w_i$), we get
$$
i^*= \arg\min_{i\in [n]} \left\| \left[\begin{array}{l}q\cr 0\end{array}\right] -  
\left[\begin{array}{l}p_i\cr \sqrt{W-w_i}\end{array}\right]
\right\|.
$$

Thus a power distance query amounts to an Euclidean distance query using an extra dimension.

\begin{Proposition}\label{prop:pdquery}
Proximity queries with respect to double-sided power distances amount to Euclidean proximity queries using an extra dimension.
\end{Proposition}

Notice that we presented Proposition~\ref{prop:pdquery} with respect to double-sided power distances.
In fact, it is known that Voronoi diagrams with respect to single-sided power distances (as known as power diagrams~\cite{aurenhammer1987power}) amount to Euclidean Voronoi diagrams using an extra dimension (e.g.,~\cite{levy2015numerical,balestriero2019geometry}).

\begin{Corollary}\label{prop:PDEVD}
Let $\hat\calP=\{\hat p_i=(p_i,w_i)\}$ be a set of $n$ weighted points with $p_i\in\bbR^m$ and $w_i\in\bbR$. 
Then the power diagram amounts to an ordinary Euclidean Voronoi diagram of points $\hat\calP^\shortparallel=\{p_i^\shortparallel=(p_i,\sqrt{W-w_i}\}$ in dimension $m+1$ clipped to the hyperplane $H_0: x_{m+1}=0$.
\end{Corollary}

Next, we shall restate Proposition~\ref{prop:pdquery} for Bregman divergences by using their power distance representation and illustrate the use of queries for solving several tasks with respect to BDs.

%%%%%%%%%%%%%
\subsection{Bregman divergence queries}\label{sec:bdkNN} %
%%%%%%%%%%%%

Let $\calA=\{\theta_1,\ldots,\theta_n\}$ be a set of $n$ parameters of $\Theta\subset\bbR^m$.
Given a query parameter $\theta\in\Theta$, we seek for the $k$-neighbors of $\theta$ with respect to a given  Bregman divergence $B_F$.
Here, $k$-neighbors is a shortcut for either considering the $k$-nearest or the $k$-farthest neighbors ($k$-NNs or $k$-FNs).
We shall report the $k=|I|$ indexes $I\subseteq [n]$ of parameters in $\calA$ so that $\theta_i$ belongs to the $k$-neighborhood of $\theta$ for all $i\in I$.

There exists a rich litterature of proximity data-structures tailored to Bregman divergences, e.g., \cite{cayton2008fast,nielsen2009bregman,nielsen2009tailored,cayton2009efficient,song2020brepartition} just to cite a few.

Since $B_F(\theta:\theta_i)=\frac{1}{2}\, \sigma(\wtheta:\weta_i)$ with $\wtheta=(\theta,w_F(\theta))$ and
 $\weta_i=(\eta_i=\nabla F(\theta_i),b_i\eqdef w_{F^*}(\eta_i))$, the Bregman queries $B_F(\theta:\theta_i)$
 amount equivalently to find the proximity neighbors with respect to

\begin{equation}\label{eq:mpdo}
\|\theta-\eta_i\|^2-b_i+B,
\end{equation}
where  $B\in\bbR$ is any prescribed constant.
Consider $B\eqdef\max_{i\in [n]} b_i$.
By offsetting the weights with $B$, we ensure that we have $B-b_i\geq 0$, and thus 
 we can  rewrite the dissimilarity of Eq.~\ref{eq:mpdo} as an equivalent squared Euclidean distance using an extra dimension: 

\begin{equation}\label{eq:pded}
\|\theta-\eta_i\|^2-b_i+B = \left\| \left[\begin{array}{l}\theta\cr 0\end{array}\right] -  
\left[\begin{array}{l}\eta_i\cr \sqrt{B-b_i}\end{array}\right]
\right\|^2.
\end{equation}

Let us define the following notations:
$\eta_i^{\shortparallel}\eqdef\left[\begin{array}{l}\eta_i\cr \sqrt{B-b_i}\end{array}\right]$ and 
$\theta^{=} \eqdef \left[\begin{array}{l}\theta\cr 0\end{array}\right]$.
Then we have 

\begin{equation}
\|\theta-\eta_i\|^2-b_i+B = \left\|\theta^{=} - \eta_i^{\shortparallel}\right|^2.
\end{equation}

It follows that the rank order of the increasing sorted list $( B_F(\theta:\theta_{(i)}) )_{i\in [n]}$ 
matches the rank order of the increasing sorted list $(\sigma(\theta,(\eta_{(i)},B-b_i)))_i$ or equivalently coincides with the rank order of the
 increasing sorted list $(\left\|\theta^{=}- \eta_i^{\shortparallel}\right|)_i$.

\begin{Proposition}\label{prop:proxqueryequiv}
Bregman proximity queries in a $m$-dimensional set $\calA$ including $k$-NN or $k$-FN queries amounts to corresponding Euclidean proximity queries in dimension $m+1$.
\end{Proposition}

\begin{Remark} 
More generally for any two weighted points $\wtheta=(\theta,w)$ and $\wtheta'=(\theta',w')$, we have for any constant $A\geq w+w'$:
$$
\sigma(\wtheta,\wtheta')+A
= \left\| \left[\begin{array}{l}\theta\cr 0\end{array}\right] -  
\left[\begin{array}{l}\eta_i\cr \sqrt{A-(w+w')}\end{array}\right]
\right\|^2.
$$
\end{Remark}

There exists many data-structures for reporting exact or approximate Euclidean nearest neighbor (NN) queries~\cite{Arya2009SpaceTime} and farthest neighbor (FN) queries~\cite{Daescu2006Farthest,Williams2018Difference}.
The reduction of Bregman queries to Euclidean queries allows one to use off-the-self software library like {\sc ANN}~\cite{Mount2010ANN} and also GPU-accelerated methods~\cite{Garcia2010KNearest,venkatasubba2025bang}.

Since Kullback-Leibler  divergence (KLD) queries are important in machine learning, let us explicit the separable and non-separable Bregman generators which yield both to KLD.

\begin{Example}
The Kullback-Leibler  divergence between two probability mass functions $p$ and $p'$ of the open $(m-1)$-dimensional standard simplex $\Delta_m$  
can be handled either as a $m$-variate separable Bregman divergence $B_{F_{\KL_m}^+}$ between positive arrays (i.e., positive measures of the orthant measure $\bbR^m_{>0}$) or as 
a $(m-1)$-variate non-separable Bregman divergence $B_{F_{\KL_{m-1}}}$:
\begin{eqnarray*}
D_\KL(p:p') &\eqdef & \sum_i p_i\log\frac{p_i}{p_i'},\\
&=& B_{F_{\KL_m}^+}(\lambda_i:\lambda_i'),\\
&=& B_{F_{\KL_{m-1}}}(\theta_i:\theta_i'),
\end{eqnarray*}
where $F_{\KL_m^+}(\lambda) \eqdef \sum_{i=1}^m \lambda_i\log\lambda_i-\lambda_i$, $\lambda(p)=p$ and $\lambda\in\Lambda=\bbR_{>0}^m$, 
and 
$F_{\KL_{m-1}}(\theta)\eqdef \sum_{i=1}^{m-1}(\theta_i\log\theta_i)+(1-\sum_{i=1}^{m-1}\theta_i)\log(1-\sum_{i=1}^{m-1}\theta_i)$, $\theta(p)=(\log\frac{p_1}{p_m},\ldots, \log\frac{p_{m-1}}{p_m})$, $\theta=\bbR^{m-1}$.
Thus we can perform KL proximity queries as equivalent Euclidean proximity queries in dimension $m$ or $m+1$ by preprocessing
the set $\calA$ into a set $\calB=\{\eta_i^{\shortparallel} \ :\ i\in [n]\}$ 
with either $\eta_{\KL_m^+}(\theta)=[\log\theta_i]_i$ or $\eta_{\KL_{m-1}}(\theta)=\left[\log\frac{\theta_i}{1-\sum_{j=1}^{m-1} \theta_j}\right]_i$ and corresponding offset values for parameter $B$.
\end{Example}

%%%%%%%%%%%%
\subsection{Symmetrized Bregman divergences}\label{sec:proxisbd}
%%%%%%%%%%%%

There are several ways to symmetrize and metricize Bregman divergences~\cite{Chen-BD1-2008,Chen-BD2-2008,acharyya2013bregman}.
Consider the common Jeffreys-Bregman symmetrization defined as follows

\begin{eqnarray*}
S_F(\theta,\theta')&\eqdef & B_F(\theta:\theta')+B_F(\theta':\theta),\\
&=& B_F(\theta:\theta')+B_{F^*}(\eta:\eta'),\\
&=& \inner{\theta-\theta'}{\nabla F(\theta)-\nabla F(\theta')}.
\end{eqnarray*}

Let us consider $\xi=(\xi_1,\xi_2)\in\bbR^{2m}$, $F_s(\xi_1,\xi_2)\eqdef F(\xi_1)+F^*(\xi_2)$, and
$\xi(\theta)=(\theta,\nabla F(\theta))$.
Then we express the symmetrized Bregman divergence as $S_F(\theta,\theta')= B_{F_s}(\xi(\theta):\xi(\theta'))$.
That is, the Jeffreys-Bregman divergence is a $m$-dimensional curved Bregman divergence because the parameter space $\Xi=\xi(\theta)\subset\bbR^{2m}$ 
 is curved in $\bbR^{2m}$ (except for quadratic generators).
However, $B_F(\xi:\xi')$ with $\xi,\xi'\in \Theta\times\calH$ is a proper product Bregman divergence.
Thus we may solve proximity queries with respect to symmetrized Bregman divergences as Euclidean proximity queries in dimension $2m+1$.

\begin{Proposition}\label{prop:sbdquery}
Proximity queries with respect to  $m$-variate Jeffreys-Bregman divergences $S_F$ can be reduced to Euclidean proximity queries in dimension $2m+1$.
\end{Proposition}

%%%%%%%%%%%%%
\subsection{Bregman diameters and Bregman closest pairs}
%%%%%%%%%%%%%%

Recall that $\calA=\{\theta_1,\ldots,\theta_n\}\subset\Theta\subset\bbR^m$ is a set of $n$ parameters.
Consider the following Bregman task:

\begin{Problem}[Bregman diameter]
Compute the Bregman diameter $\Delta_F(\calA)$ defined as:
\begin{equation}
\Delta_F(\calA) \eqdef \max_{i,j \in [n]} B_F(\theta_i:\theta_j).
\end{equation}
\end{Problem}

Solving the Bregman diameter problem amounts to report a pair of indexes $(i^*,j^*)\in [n]^2$ such that $\Delta=\Delta_F(\calA)= B_F(\theta_{i^*}:\theta_{j^*})$.
To compute $\Delta$, we may proceed by calculating for each $i \in [n]$, the farthest Bregman parameter to $\theta_i$ ($1$-FN query):
$$
\Delta_i=\max_{j \in [n]} B_F(\theta_i:\theta_j),\quad f_i^*=\argmax_{j \in [n]} B_F(\theta_i:\theta_j)
$$
and then report $\Delta=\max_{i \in [n]} \Delta_i$ with $i^*=\arg \max_{i \in [n]} \Delta_i$ and $j^*=f_{i^*}^*$.
%Let $f_i=\arg\max_{j \in [n]} B_F(\theta_i:\theta_j)$ denote the  Bregman farthest neighbor of $\theta_i$ in $\calA$ (any index $f_i$ in case of ties):
%$\max_{j \in [n]} B_F(\theta_i:\theta_j)= B_F(\theta_i:\theta_{f_i})$.

Since Bregman farthest  queries in dimension $m$ amounts to a farthest Euclidean queries in dimension $m+1$, 
we consider two sets: The query set $\calK=\{\theta_i^{=}\}$, and the neighbor data set $\calB=\{\eta_i^{\shortparallel}\}$.
We get the following proposition:

\begin{Proposition}\label{prop:diam}
The Bregman diameter $\Delta_F(\calA)$ of a set $\calA$ of $n$ $m$-dimensional parameters can be calculated in $O(P_{m+1}(n)+nQ_{m+1}(n))$ where $P_{m+1}(n)$ is the processing time of $\calB$ for FN queries $\calK$ and $Q_{m+1}(n)$ is the overall FN query time (e.g., worst-case time or amortized time $\tilde O(\cdot)$).
\end{Proposition} 

Note that in dimension $m=2$, we can directly proceed by computing the farthest power diagram~\cite{aurenhammer1987power} in $P_2(n)=O(n\log n)$ and perform FN queries as point cell location queries in $O(\log n)$ for an overall cost of $O(n\log n)$:

\begin{Corollary}\label{cor:diam2d}
The Bregman diameter of $n$ 2D parameters with respect to a given Bregman divergence can be computed in $O(n\log n)$.
\end{Corollary}

Since the Euclidean diameter $\Delta_2(\calA)\eqdef\max_{i,j} \|\theta_i-\theta_j\|$ 
 can be interpreted as the square root of the squared Euclidean Bregman diameter $\Delta_2(\calA)=\sqrt{\Delta_{Q_2}(\calA)}$, we refer to the existing literature~\cite{Finocchiaro2002Computing,liu2026optimal} for  complexity results about the well-studied Euclidean diameter problem and Euclidean $k$-NN queries.

%%%%%%%%%%%%
\subsection{Bregman and power closest pairs}\label{sec:cp}
%%%%%%%%%%%%

Consider the Bregman closest pair problem:

\begin{Problem}[Bregman closest pair]
Compute the Bregman closest pair defined as $\delta_F(\calA)\eqdef B_F(\theta_{i^*}:\theta_{j^*})$ where
\begin{equation}\label{eq:cpb}
(i^*,j^*) = \arg\min_{i,j \in [n], i\not =j} B_F(\theta_i:\theta_j).
\end{equation}
\end{Problem}

The smallest Bregman (oriented) inter-point distance is $\delta_F(\calA) = B_F(\theta_{i^*}:\theta_{j^*})$.
A trivial naive algorithm computes the closest pair using $O(n^2)$ evaluations of Bregman divergences.
We use the equivalence with Euclidean NN queries by considering the query set $\calK=\{\theta_i^{=}\}$ and the offset static parameter set $\calB=\{\weta_i\}_i$. 
%We need to query for the $2$-NN because the $1$-NN of $\wtheta_i$ in $\calB$ is $\weta_i$ and .
By using the reduction to Euclidean proximity queries, we consider the query set $\{\theta_i^=\}_i$ and the parameter set
 $\calB=\{\eta_i^{\shortparallel}\}$ to solve 
Eq.~\ref{eq:cpb}. We need to consider  $2$-NN queries $\theta_i$ because the first neighbor of $\theta_i$ in $\calB$ is $\eta_i^{\shortparallel}$ and $\sigma(\wtheta_i,\weta_i)=0$.

\begin{Proposition}\label{prop:cp}
The Bregman closest pair $\delta_F(\calA)$ of a set $\calA$ of $n$ $m$-dimensional parameters can be calculated in $O(P_{m+1}(n)+nQ_{m+1}(n))$ where $P_{m+1}(n)$ is the processing time of $\calA$ for $2$-NN queries and $Q_{m+1}(n)$ is the $2$-NN query time (e.g., worst-case or amortized, etc.).
\end{Proposition} 

When $m=2$, we can directly solve the closest pair problem by computing the $2$-order Bregman Voronoi diagram~\cite{BVD-2010} in $O(n\log n)$-time.

\begin{Corollary}\label{cor:cp2d}
The Bregman closest pair of $n$ 2D parameters with respect to a given Bregman divergence can be computed in $O(n\log n)$.
\end{Corollary}

%% file: sec-Voronoi.tex
%%%%%%%%%%%
\section{Bisectors and Voronoi diagrams}\label{sec:Vor}
%%%%%%%%%%%

Bregman Voronoi diagrams (BVDs) and their dual triangulations/complexes were studied in~\cite{BVD-2010}.
Furthermore, the authors of a\cite{BVD-2010} also stated results on additively-weighted BVDs, $k$-order BVDs, 
and $k$-bag BVDs where each Voronoi site defines its own Bregman divergence by using a positively weighted combinations of elementary BDs. 
All those Voronoi diagrams can be computed from equivalent affine Voronoi diagrams, i.e., Voronoi diagrams with hyperplane bisectors.
A key theorem of computational geometry is that any affine Voronoi diagram can be realized by an equivalent (clipped) power diagram (PD) under mild conditions~\cite{CVD-2006}. For example, the hyperbolic Voronoi diagram in Klein disk model is affine and can be realized as a power diagram clipped to the unit ball~\cite{nielsen2010hyperbolic}.

In this section, we first recall BVDs as PDs and consider also M\"obius BVDs which are both additively and multiplicatively weighted BVDs  as a generalization of ordinary Euclidean M\"obius diagrams~\cite{CVD-2006}.
We then state an equivalence between Bregman potential liftings and the paraboloid lifting (also called parabolic lifting in the literature) 
which provides another complementary viewpoint of the property that Bregman power divergences can be expressed as double-sided power distances.
The converse is also true provided that conditions on the dual convex conjugate function domains are satisfied.
Thus the paraboloid $Q_2$ is the standard canonical lifting transform in computational geometry~\cite{de2008computational} but other canonical convex potential functions can be used instead of $Q_2$.

%%%%%%%%%%%%
\subsection{Bregman Voronoi and Bregman power diagrams as ordinary power diagrams}\label{sec:bvdpd}
%%%%%%%%%%%%

\subsubsection{Identiying bisectors}

Recall that the left Bregman bisector 
$$
\Bi_F(\theta_i:\theta_j)\eqdef\{\theta\in\Theta \st Y_F(\theta:\eta_i)=Y_F(\theta:\eta_j)\}
$$ 
amounts to a clipped power bisector since the left Bregman bisector is $\theta$-affine:
\begin{equation}\label{eq:lbvd}
\Bi_F(\theta_i:\theta_j)=\{\theta\in\Theta \st \inner{\theta}{\eta_j-\eta_i}+F^*(\eta_i)-F^*(\eta_j)=0\},
\end{equation}
and the power bisector between two weighted points $\hatp_i=(p_i,w_i)$ and $\hatp_j=(p_j,w_j)$
$$
\Bi_\sigma(\hatp_i,\hatp_j)\eqdef\{p\in\bbR^m \st \sigma(p,\hatp_i)=\sigma(p,\hatp_j)\},
$$
 is also affine:
\begin{equation}\label{eq:pwd}
\Bi_\sigma(\hatp_i,\hatp_j) = \{p\in\bbR^m \st \inner{p}{2(p_j-p_i)}-w_i+w_j+\|p_i\|^2-\|p_j\|^2\}=0.
\end{equation}

When the weights of $\hatp_i$ and $\hatp_j$ are non-negative, 
the Equation~\ref{eq:pwd} can be derived by subtracting the quadratic equation of the sphere $S(\hatp_j)$ from the quadratic equation of the sphere $S(\hatp_i)$ where 
$S(\hatp_l): \|p_l-x\|^2-w_l=0$.
The power bisector hyperplane $\Bi_\sigma(\hatp_i,\hatp_j)$ is called the radical axis of   $S(\hatp_i)$ and $S(\hatp_j)$ or chordale~\cite{devillers1992space} (see Figure~\ref{fig:powerbisector}).
We denote by $\Vor_F^l(\calP)$ and $\Vor_F^r(\calP)$ the left and right Bregman Voronoi diagrams and by $\Vor_\sigma(\calP)$ the power diagram of sites $\calP$.

\begin{figure}%
\centering
\includegraphics[width=0.53\columnwidth]{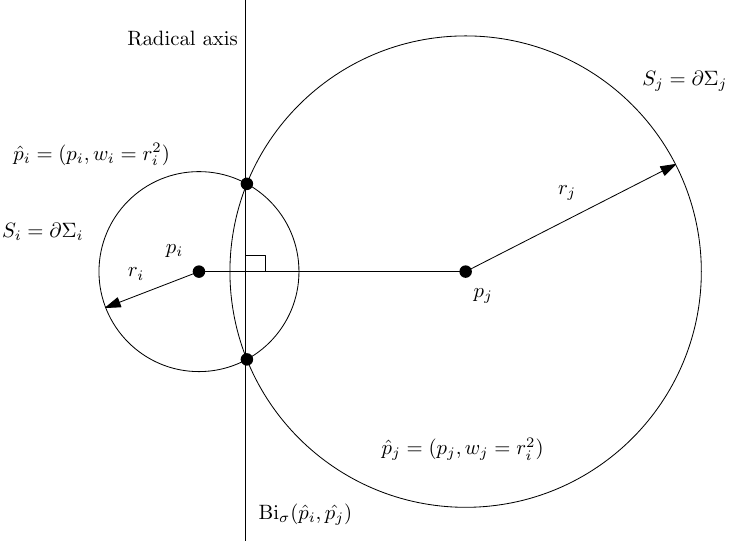}%
\includegraphics[width=0.42\columnwidth]{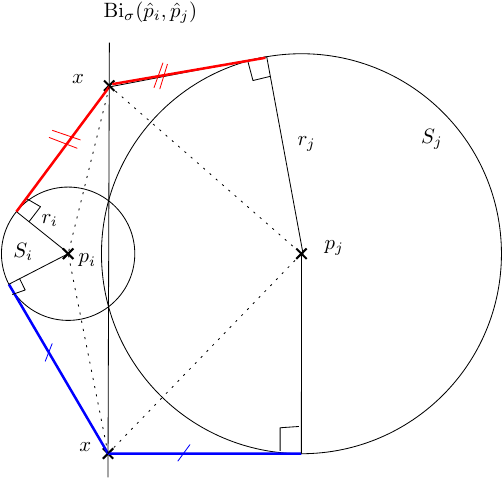
}
\caption{Left: The power bisector between two weighted points with non-negative weights is equivalent to the radical axis (chordale) of their corresponding spheres. 
The power bisector is perpendicular to the line segment joining the sphere centers.
Right: Power bisector is the locii of points $x$ at equal power distance to spheres $S_i=\hat{p}_i$ and $S_j=\hat{p}_j$. }%
\label{fig:powerbisector}%
\end{figure}

By dividing by $2$ Eq.~\ref{eq:pwd} and identifying its terms with Eq.~\ref{eq:lbvd}, we get the following proposition:

\begin{Proposition}[\cite{BVD-2010}]
The left Bregman Voronoi diagram of $\calA=\{\theta_1,\ldots,\theta_n\}$ amounts to a power diagram of 
a weighted point set $\check\calB=\{\hatp_i=(\eta_i,w_i=w_{F^*}(\eta_i)=\|\eta_i\|^2-2\, F^*(\eta_i))\}$ clipped to $\Theta$:
\begin{equation}
\Vor_F^l(\calA)=\Vor_\sigma(\check\calB)\cap\Theta.
\end{equation}
\end{Proposition}

Since power diagrams are also called Laguerre Voronoi diagrams~\cite{imai1985voronoi}, 
we shall call the equivalent power diagram sites $\hatp_i$'s the Laguerre parameters.

Notice that $\Bi_\sigma(\hatp_i,\hatp_j)=\Bi_\sigma((p_i,w_i+W),(p_j,w_j+W))$ for any weight offset $W\in\bbR$.
Thus we may offset the weights of $\check\calB$ by $W=\max_i w_i$ so that $W+w_i\geq 0$ for all $i\in [n]$;
This allows one to visualize the Voronoi sites $\hatp_i$ as corresponding spheres $S(\hatp_i)$.
See Figure~\ref{fig:powerbisectorex} for some illustrations of power diagrams.
Similarly, the right Bregman Voronoi diagram amounts to a power diagram for $\hat\calA=\{\hatp_i=(\theta_i,w_F(\theta_i))\}$ clipped by the gradient domain $\calH$:
\begin{eqnarray*}
\Vor_F^r(\calA) &=& \Vor_{F^*}^l(\calB),\\
&=& \Vor_\sigma(\hat\calA)\cap\calH,
\end{eqnarray*}
where $\calB=\{\eta_i\}$ (and since ${F^*}^*=F$).

\begin{Remark}
The Voronoi diagram (VD) with respect to an arbitrary dissimilarity $D(\cdot:\cdot)$ coincides with the VD  with respect to any scaled dissimilarity $\lambda D(\cdot:\cdot)$ for any $\lambda>0$.
Since a scaled Bregman divergence is a Bregman divergence for the scaled generator, $\lambda B_F=B_{\lambda F}$, we get $\Vor_F(\{\theta_i\})=\Vor_{\lambda F}(\{\theta_i\})$ for any $\lambda>0$.
Thus  the conversion of $\theta$-parameters to equivalent Laguerre parameters $\hatp$ is not unique (Figure~\ref{fig:lambda}) as we can choose any prescribed $\lambda>0$ and get the corresponding Laguerre parameters:
\begin{eqnarray*}
\hatp_i &=& \left(\lambda\eta_i, \lambda^2 \|\eta_i\|^2-2\lambda\, (F(\theta_i)-\inner{\theta_i}{\eta_i})\right),\\
&=&\left(\lambda\eta_i, \lambda^2 \|\eta_i\|^2-2\lambda\,F^*(\eta_i)\right).
\end{eqnarray*}

More generally, the Bregman divergence $B_F(\theta:\theta')$ encodes a geometric divergence $\calD(p_{\theta}:p_{\theta'})$ between corresponding points $p_{\theta}$ and $p_{\theta'}$ (contrast function) on a Hessian manifold $(M,g,\nabla,\nabla^*)$ where $\theta(\cdot)$ denotes a $\nabla$-affine coordinate system 
and $\eta(\cdot)$ denotes the dual $\nabla^*$-affine coordinate system.
The Bregman generator can be written in any $\nabla$-affine coordinate system. Thus for the Bregman BVDs, 
we may consider the following generic parameterization (gauge freedom~\cite{nielsen2025generalized}): 
$\lambda (F(A\theta+b)+\inner{c}{\theta}+d)$ where $\lambda>0$, $A\in\GL(d)$, $b,c\in\bbR^d$ and $d\in\bbR$.
\end{Remark}

\begin{figure}%
\centering
\begin{tabular}{lccc}
 &  $w_i=r_i^2$ & $w_i+W$ & $w_i=0$ \\
$n=2$ &\includegraphics[width=0.26\columnwidth]{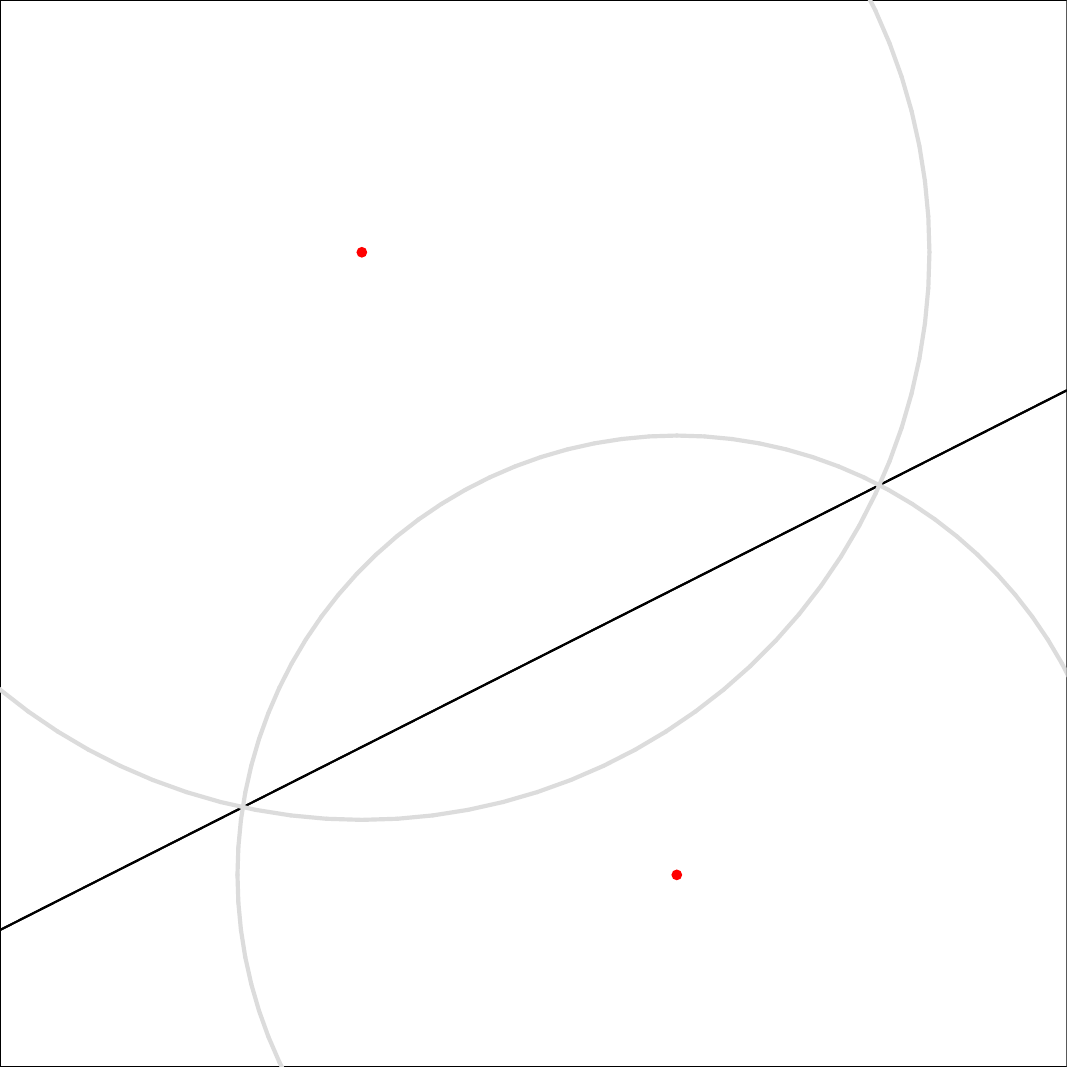} &
\includegraphics[width=0.26\columnwidth]{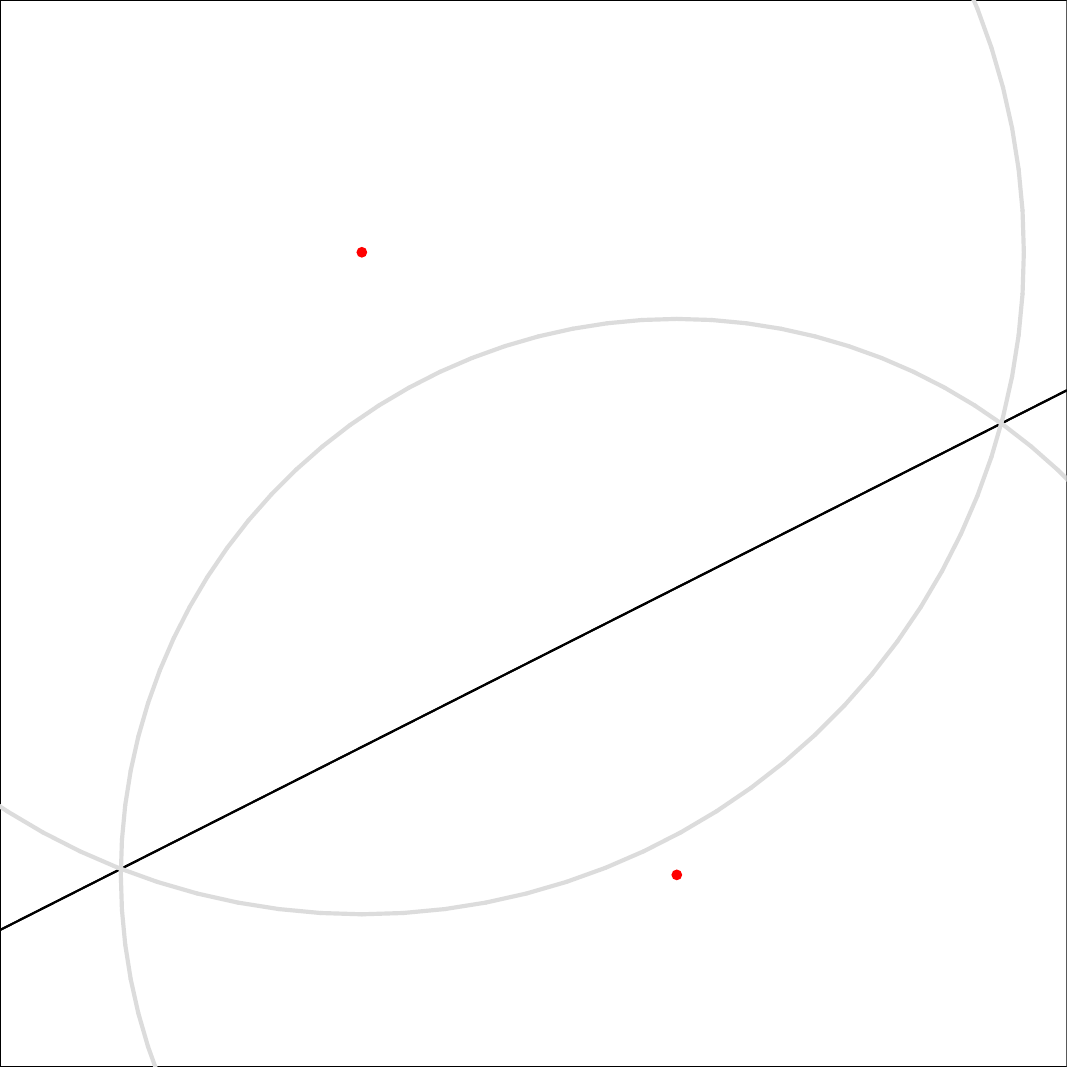} &
\includegraphics[width=0.26\columnwidth]{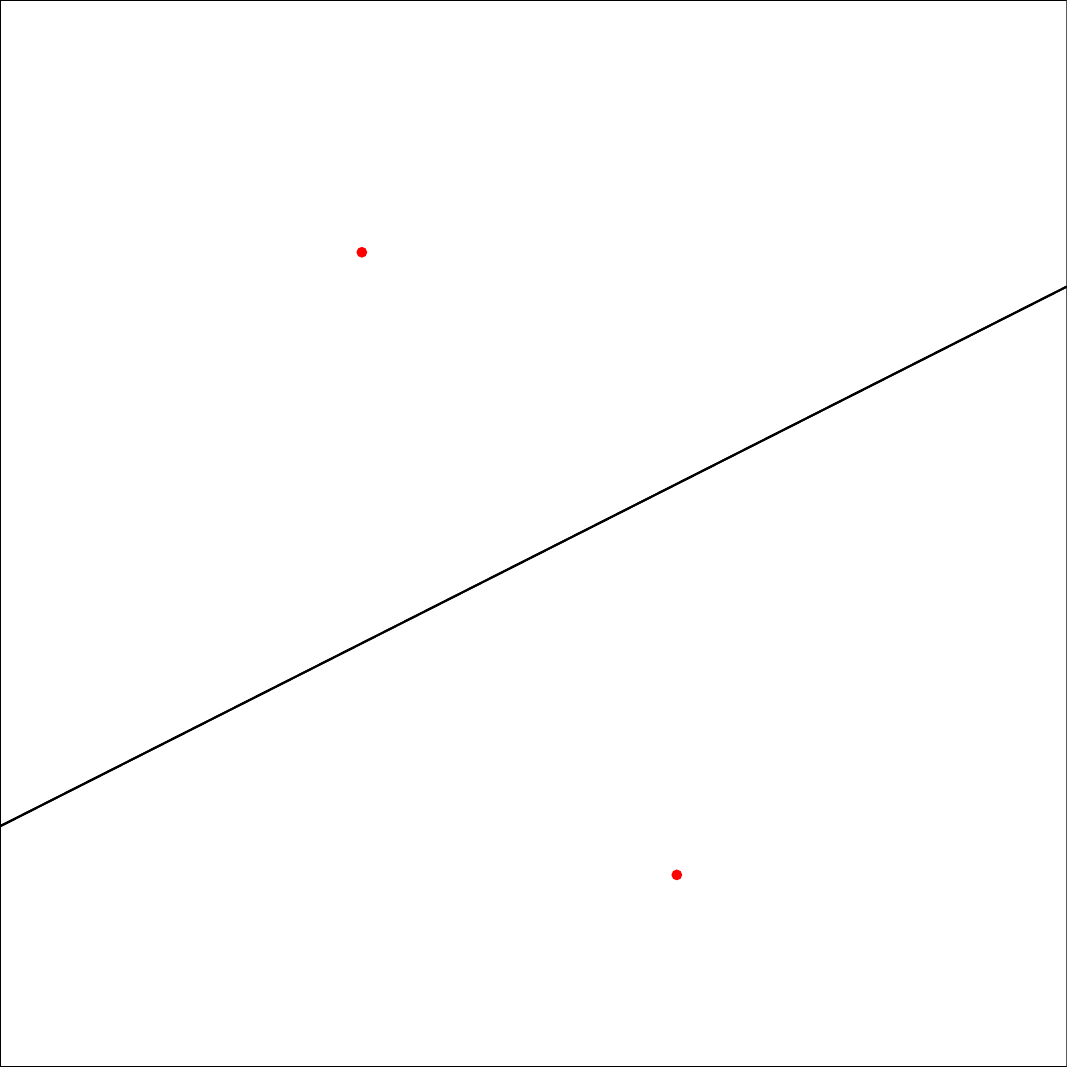} \\
$n=8$ & \includegraphics[width=0.3\columnwidth]{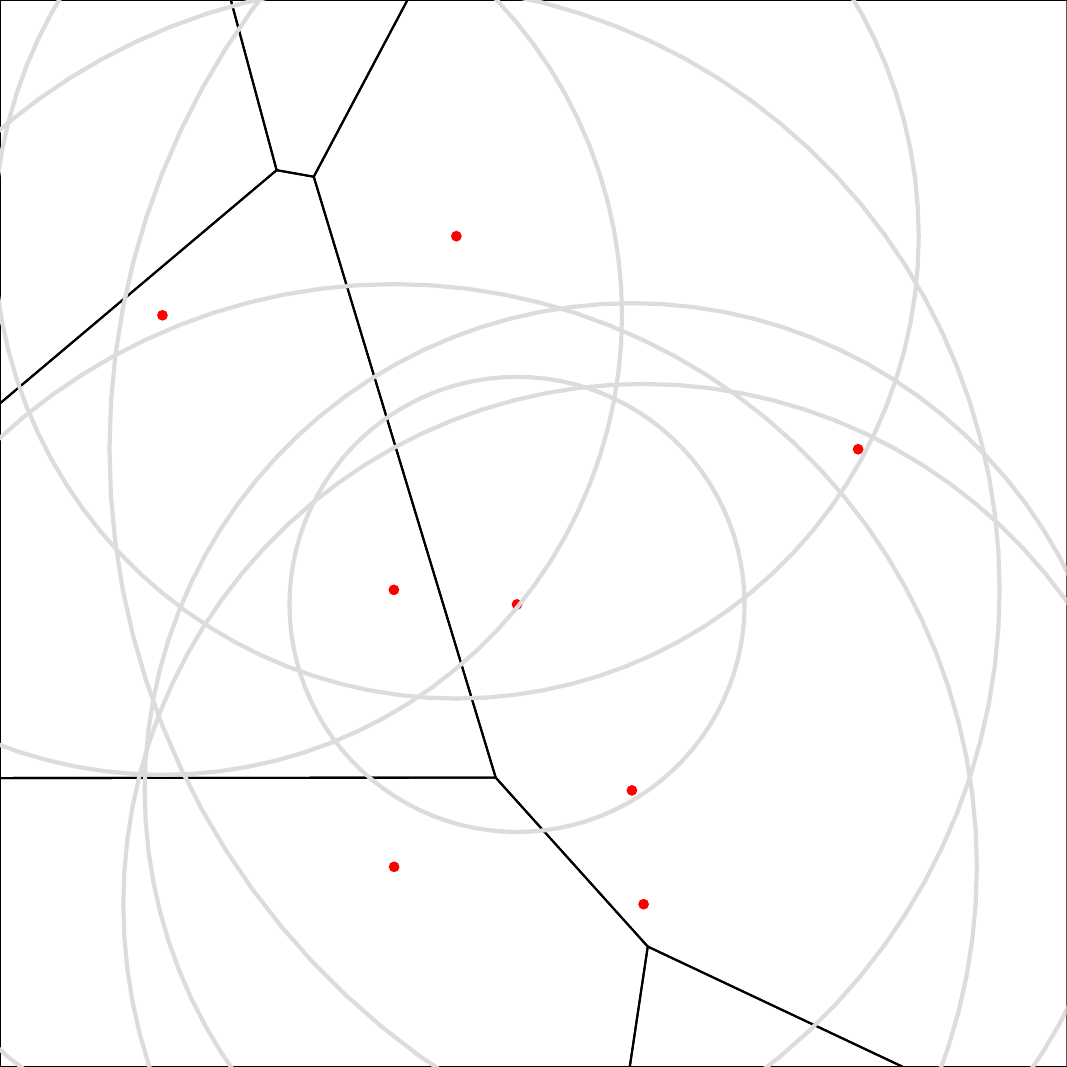} &
\includegraphics[width=0.26\columnwidth]{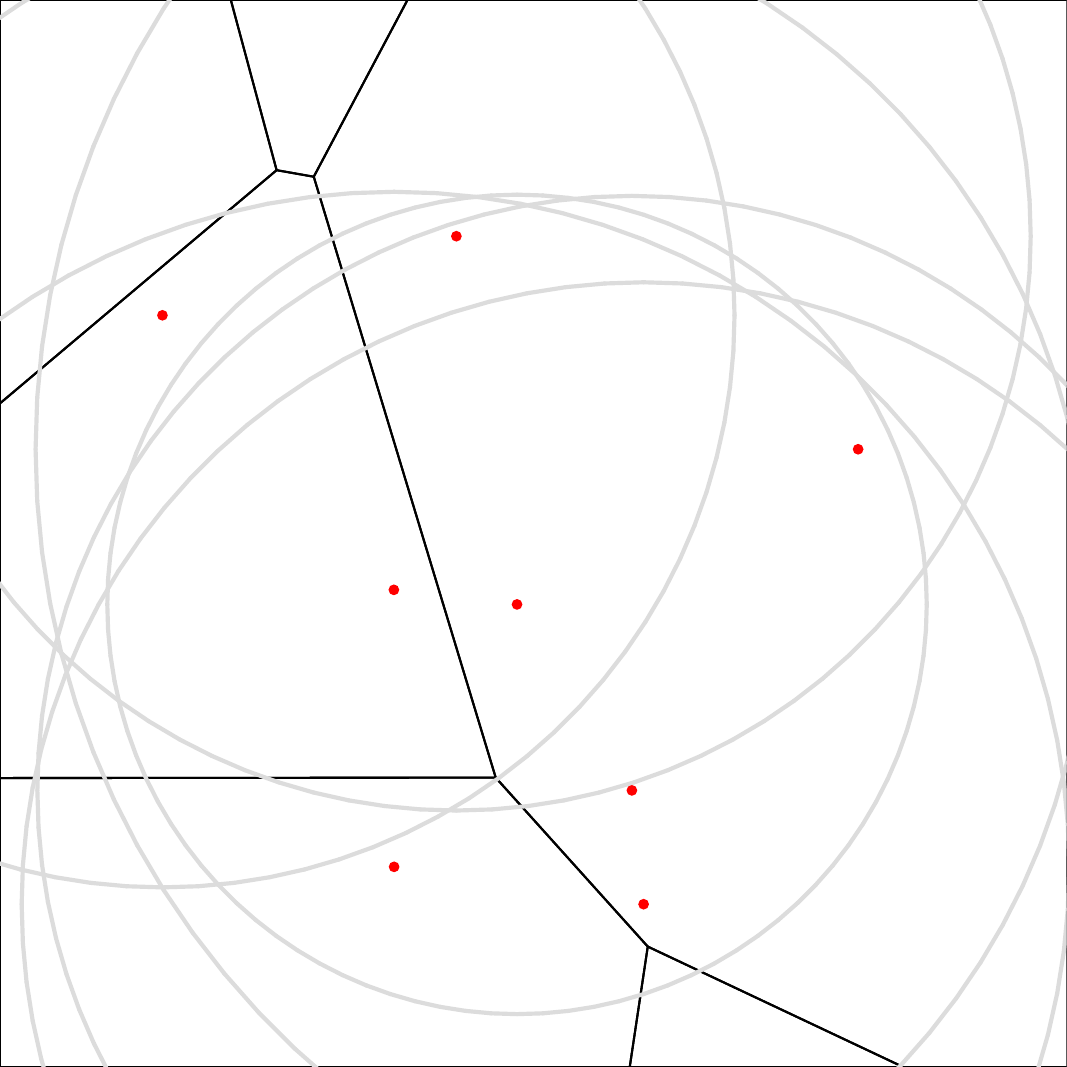} &
\includegraphics[width=0.26\columnwidth]{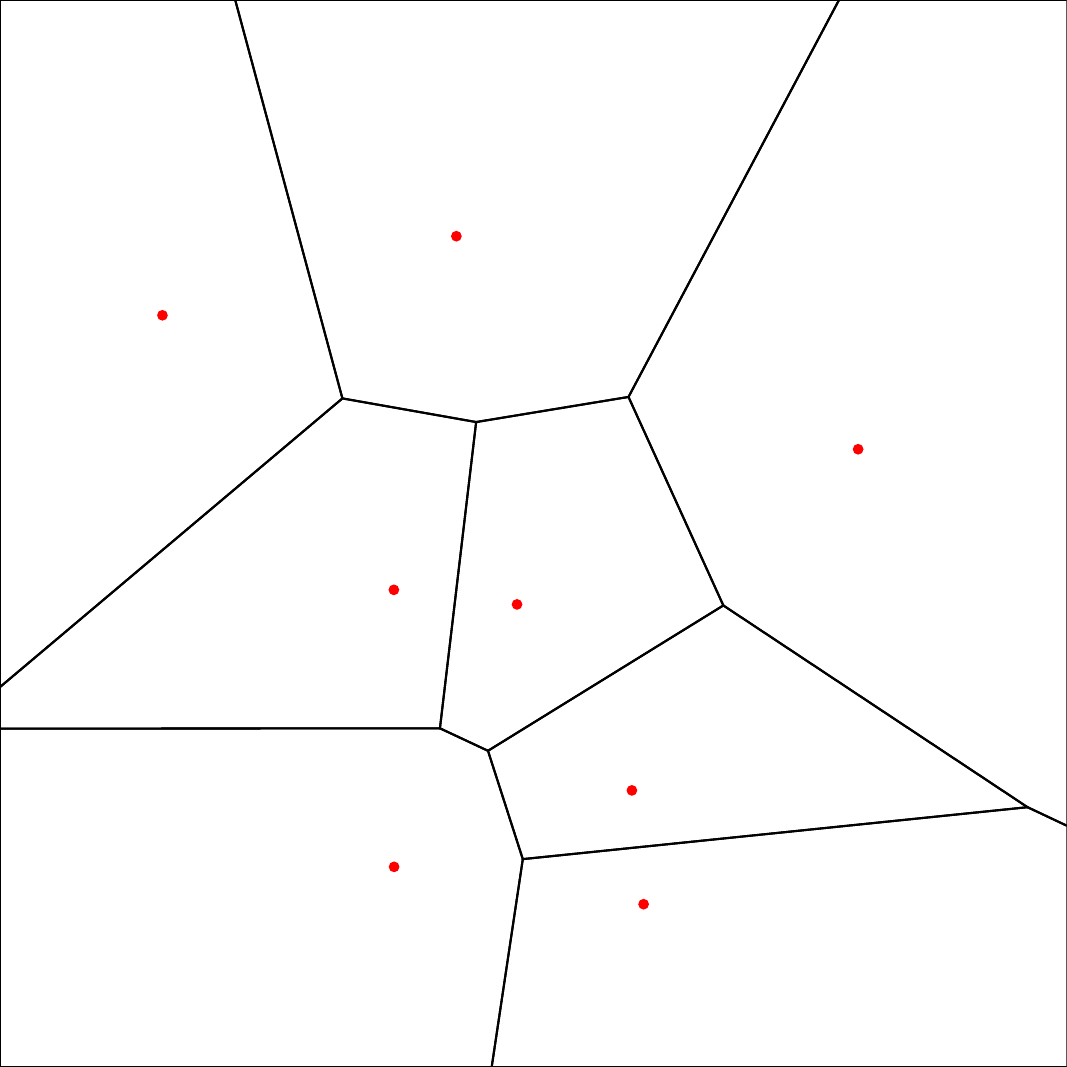} \\
& (a) & (b) & (c)
\end{tabular}

\caption{Top: Power bisector between (a) two $n=2$ weighted points with non-negative weights corresponding to squared radius, (b) weights shifted by some amount $W$ showing that power bisectors coincide with (a), 
and (c) vanishing weights yielding the ordinary Euclidean bisector.
Bottom: (a) Power diagram (PD) for $n=8$ points with some empty Voronoi cells, (b) PD with weights shifted by the same amount $W$ illustrating the PDs coincide, and (c) PD with all vanishing weights yielding ordinary VD with all non-empty cells.}
\label{fig:powerbisectorex}%
\end{figure}

Figure~\ref{fig:lambda} displays some examples showing that the left Bregman VDs are equivalent to power diagrams clipped to the domain $\Theta$ for different.
The results extends to Bregman power diagrams defined as Voronoi diagrams with respect to Bregman power divergences (additively weighted Bregman Voronoi diagrams):

\begin{Proposition}
The left and right Bregman power diagrams of $\hat\calA\eqdef \{(\theta_1,a_1),\ldots,(\theta_n,a_n)\}$ amounts to a power diagram of 
a weighted point set 
$\wwcalB\eqdef \{\hatp_i=(\eta_i,w_i=w_{F^*}(\eta_i)-a_i\}$ 
clipped to $\Theta$ and to a power diagram of 
a weighted point set
 $\wwcalA\eqdef \{\hatp_i=(\theta_i,w_i=w_{F}(\theta_i)-a_i\}$ 
clipped to $\calH$:
\begin{eqnarray*}
\Vor_F^l(\hat\calA) &=& \Vor_\sigma(\wwcalB)\cap\Theta,\\
\Vor_F^r(\hat\calA) &=& \Vor_\sigma(\wwcalA)\cap\calH.
\end{eqnarray*}
\end{Proposition}

Since power diagrams can be calculated from ordinary Voronoi diagrams using an extra dimension, we get the following corollary:

\begin{Corollary}
Bregman Voronoi diagrams and Bregman power diagrams (also known as additively weighted Bregman Voronoi diagrams~\cite{BVD-2010}) be calculated as ordinary Euclidean Voronoi diagrams.
\end{Corollary}

Thus we can use software libraries like {\sc qhull}~\cite{barber1996quickhull} of computational geometry to compute those Voronoi diagrams.

\begin{figure}%
\centering
\begin{tabular}{ccc}
\fbox{\includegraphics[width=0.28\columnwidth]{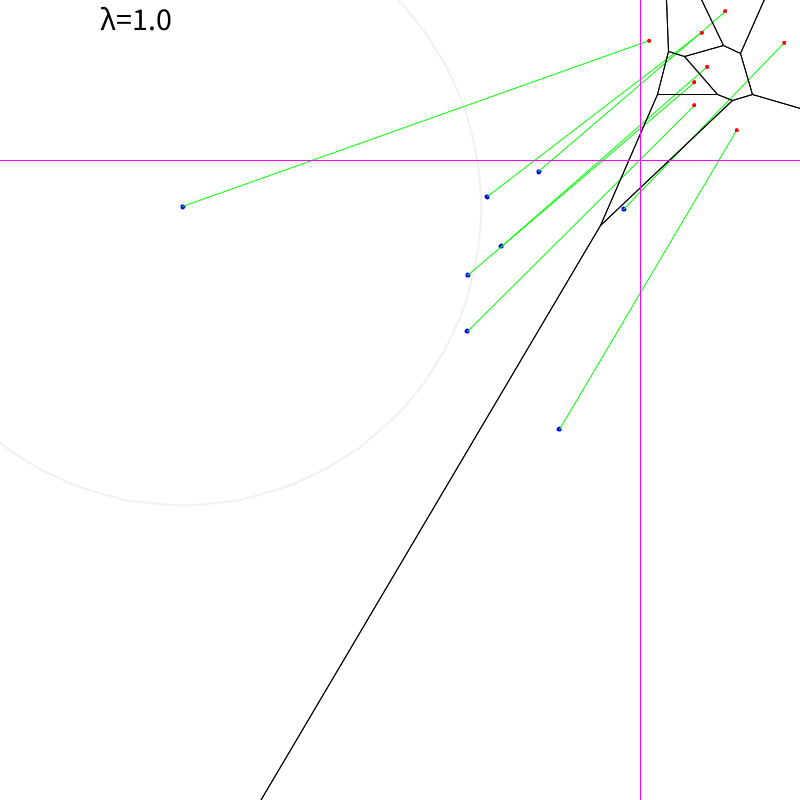}} &
\fbox{\includegraphics[width=0.28\columnwidth]{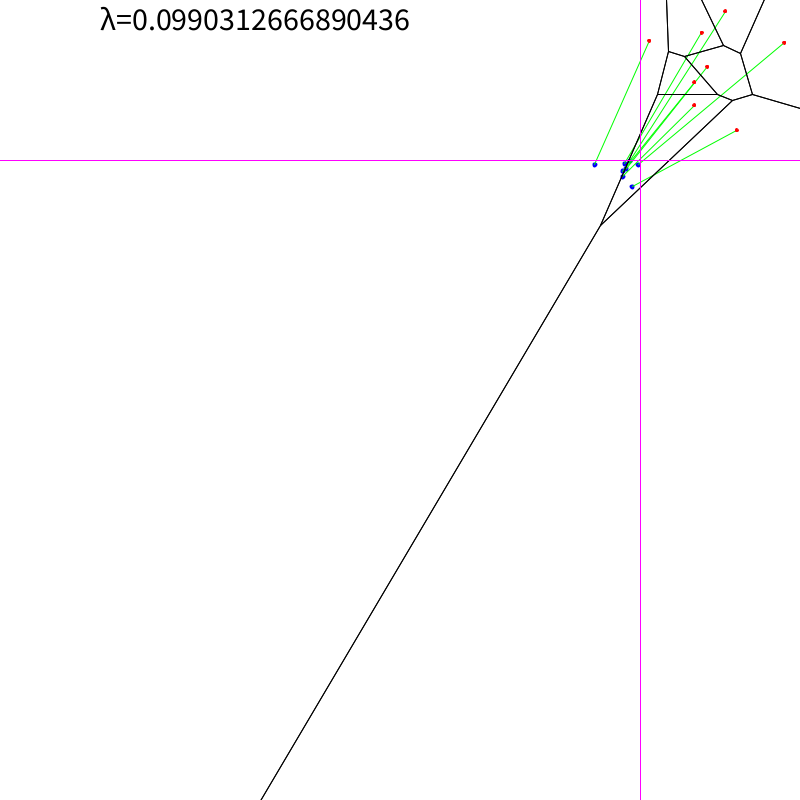}} &
\fbox{\includegraphics[width=0.28\columnwidth]{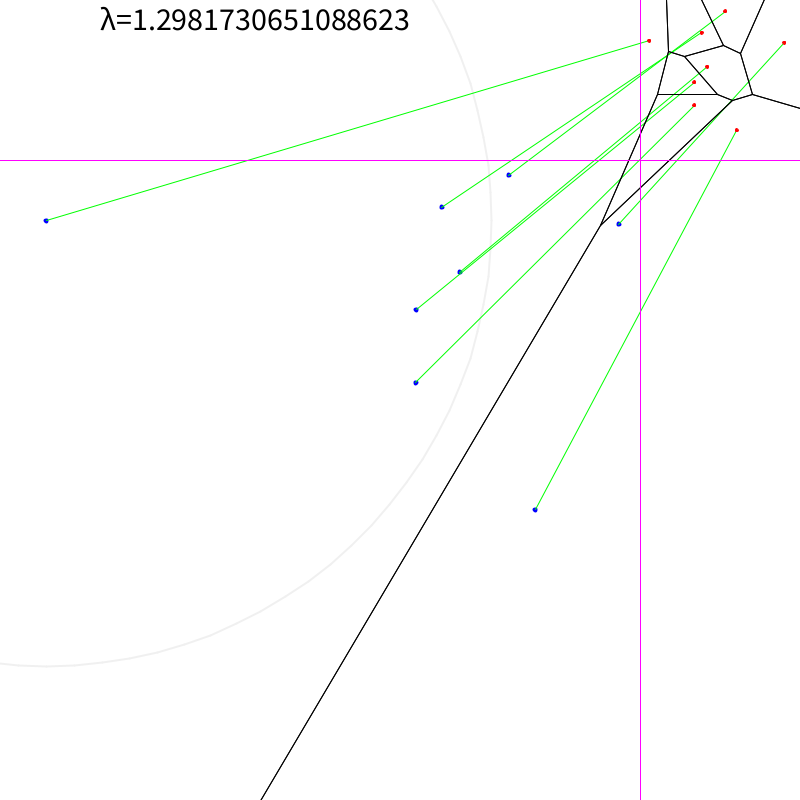}} 
\end{tabular}
\caption{Bregman Voronoi diagrams (eKLD) of $n=8$ points with equivalent power diagrams obtained for different values of $\lambda$.\label{fig:lambda}}
\end{figure}

%%%%%%%%%%%%
%\subsection{Bregman power diagrams from PD equivalence}\label{sec:directbvdpd}
%%%%%%%%%%%%

\subsubsection{Derivation from the power distance representation}

Using the equivalence of a Bregman divergence with its mixed primal-gradient power distance, we can directly prove that left Bregman bisector amounts
 to a corresponding power bisector as follows:
\begin{eqnarray*}
\Bi_F(\theta_i,\theta_j) &=& \{ \theta\in\Theta \st B_F(\theta:\theta_i) = B_F(\theta:\theta_i)\},\\
& =& \{\wtheta=(\theta,w_F(\theta)) \st \frac{1}{2}\sigma(\wtheta:\weta_i) = \frac{1}{2}\sigma(\wtheta:\weta_j) \},\\
&=& \{ \theta\in\Theta \st  \sigma(\theta:\weta_i) = \sigma(\theta:\weta_j)\},\\
&=& \Bi_\sigma(\weta_i,\weta_j)\cap \Theta
\end{eqnarray*}
where $\weta_l=(\eta_l,\|\eta_l\|-2\, F^*(\eta_l))$.  
Thus the left Bregman bisector is equivalent to a power bisector.

%Let $\hat q_i= \left(\eta_i,\|\eta_i\|^2+ 2\,(F(\theta_i)-\inner{\theta_i}{\eta_i})\right)$ with $\eta_i=\nabla F(\theta_i)$ and 
%$\hat q_j=\left(\eta_j,\|\eta_j\|^2+2(F(\theta_j)-\inner{\theta_j}{\eta_j})\right)$ with $\eta_j=\nabla F(\theta_j)$. 
%Then we can reinterpret the left Bregman bisector as an equivalent power bisector clipped to the primal domain $\Theta$:
%
%\begin{equation}
%\Bi_F(\theta_i,\theta_j) = \Bi_\sigma(\hat q_i,\hat q_j) \cap\Theta.
%\end{equation}
%
%It follows that the left Bregman VD is a PD clipped to the domain $\Theta$ (Figure~\ref{fig:correspondenceBVDPD}):
%
%\begin{Proposition}[\cite{BVD-2010}]\label{prop:equivBVD-PD}
%The left Bregman Voronoi diagram of a set $\{\theta_i\}$ of $n$ parameters  amounts to the power diagram of a set $\{\hat q_i\}$ of $n$ corresponding weighted points  clipped to the domain $\Theta$:
%$$
%\Vor_F(\{\theta_i\})=\Vor_\sigma(\{\hat q_i\})\cap\Theta,
%$$
%where $\hat q_i=(\eta_i,\|\eta_i\|^2+2(F(\theta_i)-\inner{\theta_i}{\eta_i}))$ with $\eta_i=\nabla F(\theta_i)$.
%\end{Proposition}
%
%
%\begin{Corollary}\label{cor:wBVD}
%The weighted Bregman Voronoi diagram of a set of $n$ weighted parameters $\{(\theta_i,w_i)\}$ amounts to a clipped power diagram. 
%\end{Corollary}
 %

\begin{figure}%
\centering
\begin{tabular}{cc}
 $n=2$ (Bisector) & $n=8$ \\
\fbox{\includegraphics[width=0.45\columnwidth]{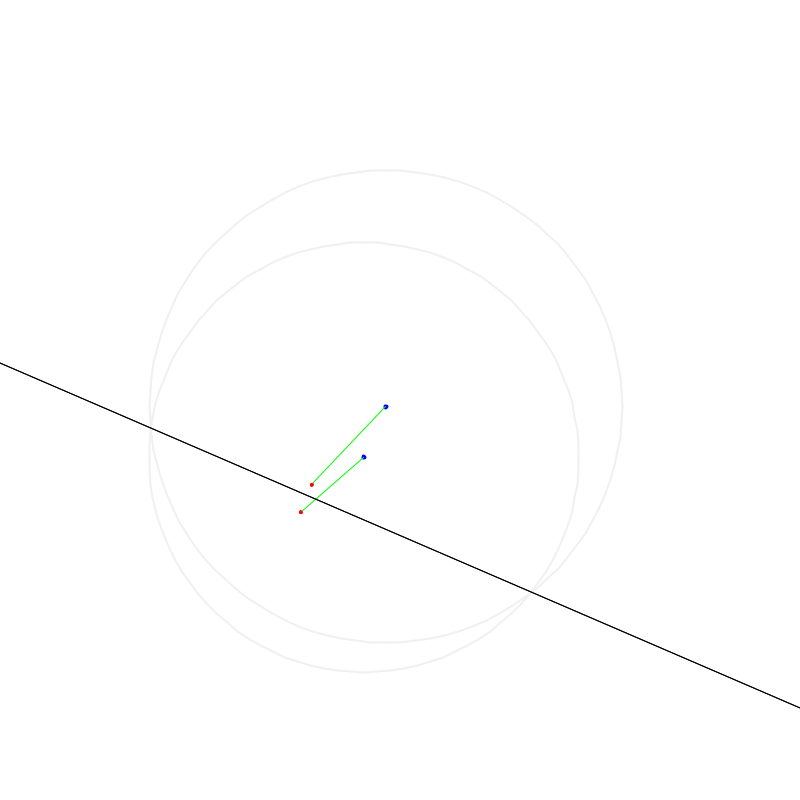}} &
\fbox{\includegraphics[width=0.45\columnwidth]{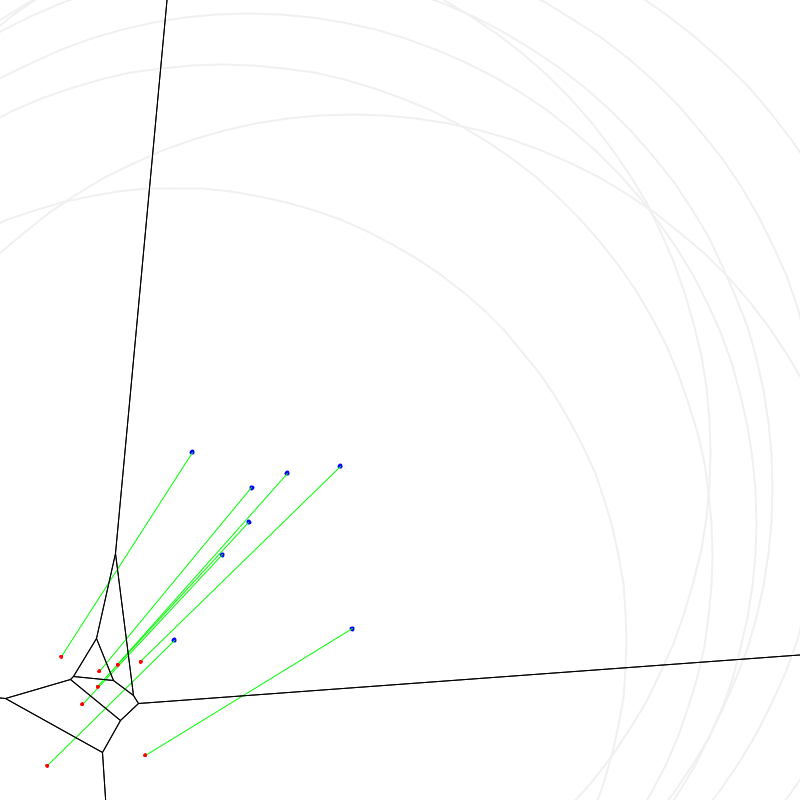}}  \\
\fbox{\includegraphics[width=0.45\columnwidth]{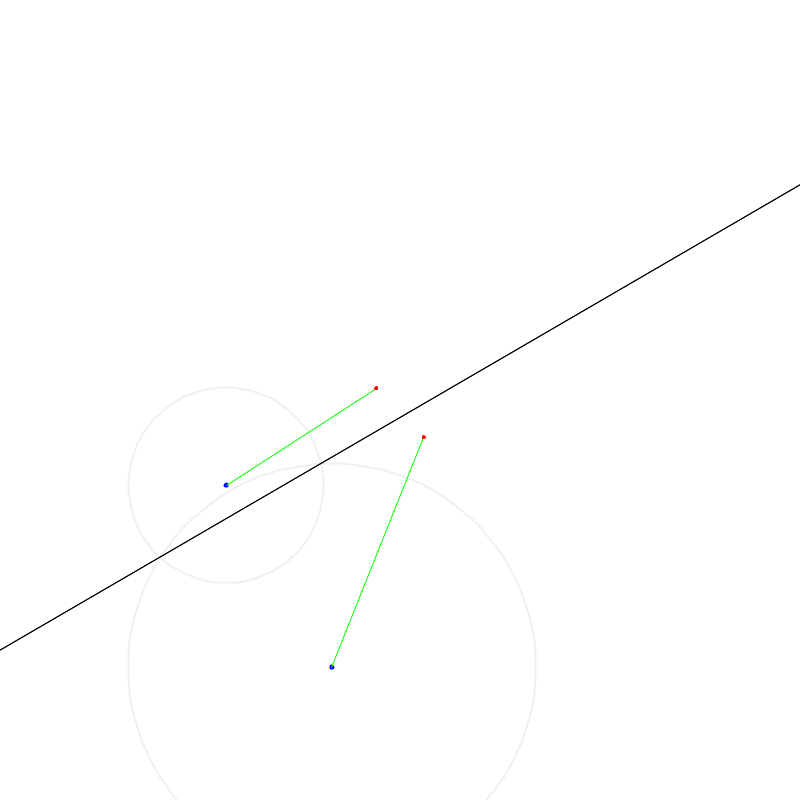}} &
\fbox{\includegraphics[width=0.45\columnwidth]{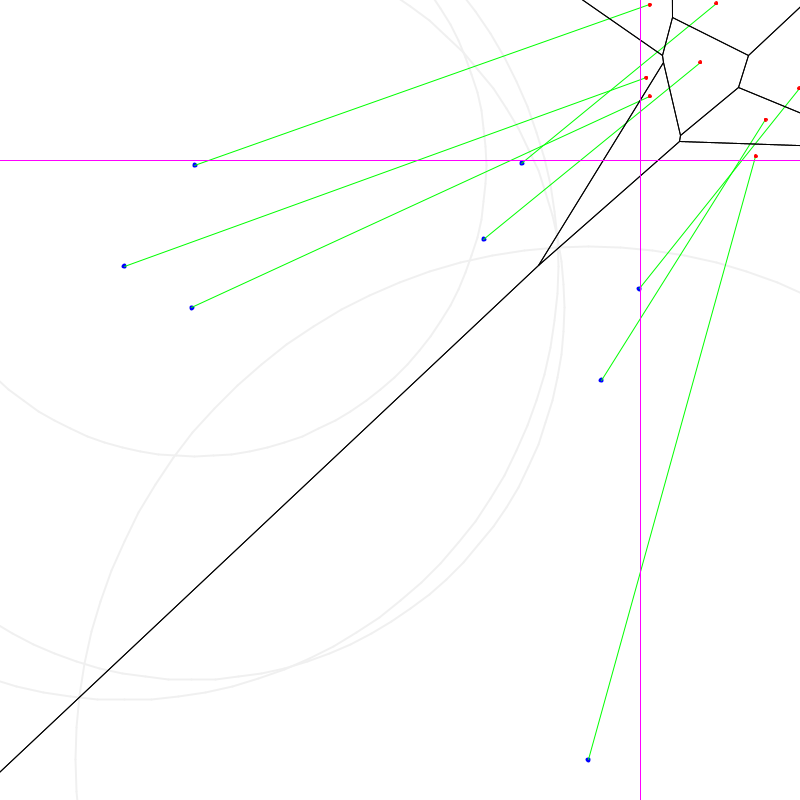}}  
\end{tabular}
\caption{Correspondence between (left) Bregman Voronoi diagrams (BVDs) and power diagrams (PDs).
Left: Bregman bisector corresponds to a power bisector. Right: BVD and equivalent PD for $n=8$. 
Corresponding Bregman Voronoi sites (red) to weighted points (blue) are indicated by (green) edges.
Observe that some PD cells may be empty but all Bregman Voronoi cells are non-empty.
Top: BVD for the exponential divergence $F(\theta)=e^\theta$. Bottom: BVD for the extended KL divergence $F(\theta)=\theta\log\theta-\theta$.
}%
\label{fig:correspondenceBVDPD}%
\end{figure}

In general, the $k$-order weighted Bregman Voronoi diagram is equivalent to a power diagram~\cite{BVD-2010}. 
The farthest Bregman Voronoi diagram is also a weighted Bregman Voronoi diagram and therefore   equivalent to a power diagram~\cite{aurenhammer1987power}.

%%%%%%%%%%%%
\subsection{Bregman M\"obius diagrams}\label{sec:MobVor}
%%%%%%%%%%%%%%
Consider both an additively and multiplicatively weighted parameter set $\{\breve{\theta}_i=(\lambda_i,w_i,\theta_i)\}_i$ where $\lambda_i>0$ and $w_i\in\bbR$.
The divergence between $\theta$ and site $\breve{\theta}_i$ is defined by 
$$
D(\theta:\breve{\theta}_i) \eqdef \lambda_i\, B_F(\theta:\theta_i)-w_i = B_{\lambda_i F}(\theta:\theta_i)-w_i.
$$
The (left-sided) Bregman M\"obius Voronoi bisector between $\breve{\theta}_i$ and $\breve{\theta}_j$ is  expressed as
$$
(\lambda_i-\lambda_j)F(\theta)+\inner{\theta}{\lambda_j\eta_j-\lambda_i\eta_i}+c_{ij}=0,
$$
where 
$$
c_{ij}\eqdef\lambda_i\inner{\theta_i}{\eta_i}-\lambda_j\inner{\theta_j}{\eta_j}+w_j-w_i.
$$

We can linearize the bisector using an extra dimension so that the  bisector becomes an affine bisector in the $\eta$-parameterization.
Let $z=(\theta,F(\theta))$. Then we have 
$$
\Bi_F(\breve{\theta}_i,\breve{\theta}_j): \Inner{z}{\left[\begin{array}{l} \lambda_j\eta_j-\lambda_i\eta_i\cr \lambda_i-\lambda_j\end{array}\right]}+c_{ij}=0.
$$

We can thus compute Bregman M\"obius diagrams similar to the techniques reported for the ordinary M\"obius diagrams~\cite{CVD-2006}.

%%%%%%%%%%%%%
\subsection{Equivalence of potential liftings}\label{sec:equivpotential}
%%%%%%%%%%%%%

Let us consider the link between the paraboloid~\cite{CVD-2006}  and the Bregman potential   transforms~\cite{BVD-2010}.
Let $\calP: \{(x,y) \st y=\inner{x}{x}, x\in\bbR^d\}$ denote the graph of the paraboloid sitting in $\bbR^{m+1}$. 
A weighted point $\hat p=(p,w)$ is mapped into the hyperplane~\cite{aurenhammer1987power} of equation:

\begin{equation}\label{eq:paraboloidlift}
H_{\hat p}^\calP: y=\inner{x}{2p}-\|p\|^2 + w.
\end{equation}

That is, $H_{\hat p}^\calP$ is the supporting hyperplane of $\calP$ at $(p,\inner{p}{p})$ shifted vertically by $w$. 
Denote by $H_{\hat p}^{\uparrow\calP}: \inner{x}{2p}-\|p\|^2 + w\geq 0$ the upper halfspace delimited by $H_{\hat p}^\calP$.
Then the power diagram can be calculated from a polytope obtained by the intersection of upper halfspaces as follows:
$$
\Vor_\sigma\{\hat p_i\}) = \left(\left( \cap_{i=1}^n H_{\hat p_i}^{\uparrow\calP} \right) \cap \calP\right)\downarrow,
$$
where the vertical projection operation $(x,y)\downarrow=x$ is denoted by the down arrow $\downarrow$.

\begin{figure}
\centering 
\includegraphics[width=0.8\columnwidth]{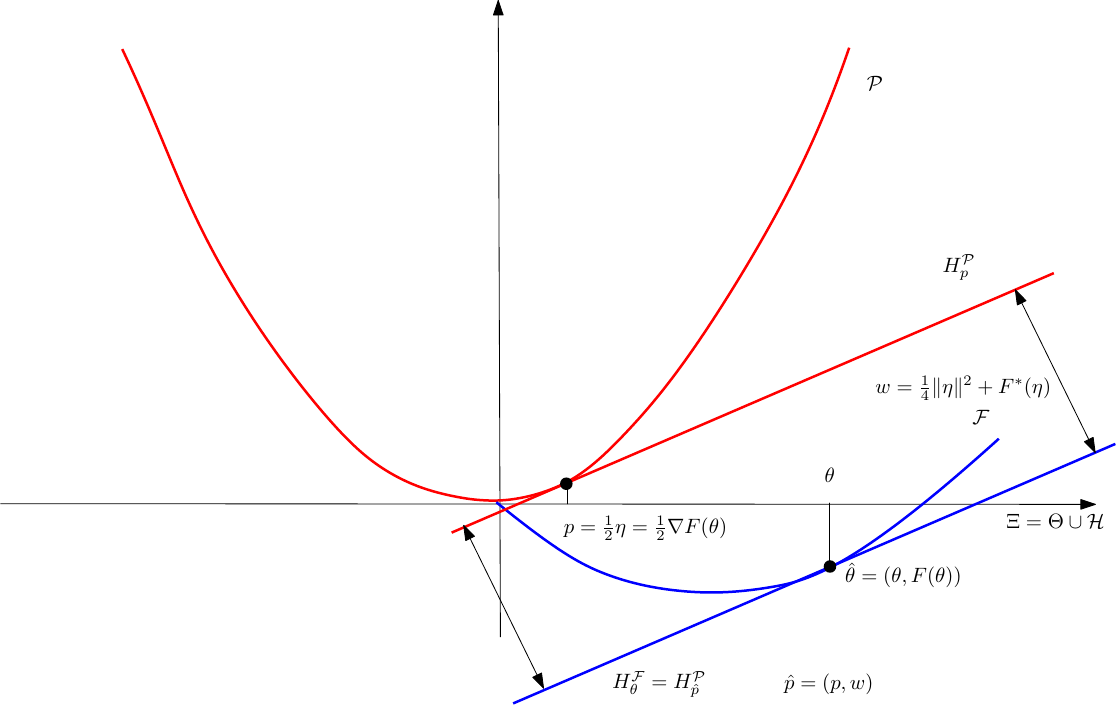}

\caption{The supporting hyperplane $H^\calF_\theta$ of the potential graph $\calF$ at $\theta^\uparrow=(\theta,F(\theta))$ is parallel to the
 lifted hyperplane $H^\calP_{p}$ at $p=\frac{1}{2}\eta$ of the paraboloid graph $\calP$.
By shifting vertically $H^\calP_{p}$ by the amounts $w=\frac{1}{4}\|\eta\|^2-\inner{\theta}{\eta}+F(\theta)$, the hyperplanes match: 
$H^\calP_{\hat p}=H^\calF_\theta$.  
 }
\label{fig:equivalencetangenthyperplane}
\end{figure}

Similarly, let $\calF = \{ (\theta,F(\theta)) \st \theta\in\Theta\subset\bbR^d\}$ be the graph of the Bregman potential function sitting in $\bbR^{m+1}$.
A parameter $\theta\in\Theta$ is lifted to the potential function $\calF$ as $\theta^\uparrow=(\theta,F(\theta))$ and the tangent hyperplane of $\calF$ at $\theta^\uparrow$ has equation:

\begin{equation}\label{eq:potentiallift}
H_\theta^\calF: y=\inner{x-\theta}{\eta}+F(\theta).
\end{equation}
%That is, $H_\theta^\calF$ is the supporting hyperplane of $\calF$ at $\theta^{\uparrow\calF}=(\theta,F(\theta))$.
Notice that when $F(\theta)=Q(\theta)\eqdef\inner{\theta}{\theta}$, $\calF=\calP$ and the supporting hyperplanes match for zero-weighted points $\hat p=(\theta,0)$: $H_{\hat p}^{\calP}=H_\theta^{\calF}$.
Denote by $H_\theta^{\uparrow\calF}: \inner{x-\theta}{\eta}+F(\theta)\geq 0$ the upper halfspace delimited by $H_\theta^\calF$.
Then the left Bregman Voronoi diagram~\cite{BVD-2010} is computed as follows:
$$
\Vor_F(\{\theta_i\})= \left( \left(\left( \cap_{i=1}^n H_{\hat p_i}^{\uparrow\calF} \right) \cap \calF\right)\downarrow \right) \cap\Theta.
$$

To interpret Bregman upper halfspaces $H_\theta^\calF$ as equivalent power upper halfspaces $H_{\hat p}^\calP$, 
 we choose $\hat p$ in order to match Eq.~\ref{eq:potentiallift} with  Eq.~\ref{eq:paraboloidlift}  as follows:
%$p=\frac{1}{2}\eta$ and 
%$w=\frac{1}{4}\|\eta\|^2+\inner{\theta}{\eta}-F(\theta)=\frac{1}{4}\|\eta\|^2+ F^*(\eta)$:
\begin{eqnarray*}
%H_\theta^\calF &=& H_{\hat p}^\calP,\\
\hat p &=& \left(\frac{1}{2}\eta,\frac{1}{4}\|\eta\|^2+\inner{\theta}{\eta}-F(\theta)\right),\\
 &=& \left(\frac{1}{2}\eta,\frac{1}{4}\|\eta\|^2+F^*(\eta)\right).
\end{eqnarray*}

Graphically speaking, the hyperplane tangent at $\calF$ at $\theta^\uparrow=(\theta,F(\theta))$ coincides the hyperplane tangent at $\calP$ at $(\eta,\inner{\eta}{\eta})$ shifted vertically by amount $w$ (Figure~\ref{fig:equivalencetangenthyperplane}). Figure~\ref{fig:equivalencetangenthyperplaneH} illustrates the matching of those hyperplanes when choosing the Bregman generator as the negative Shannon entropy $F_\Shannon(\theta)=\sum_i \theta_i\log\theta_i$.

\begin{Proposition}\label{prop:liftequiv}
The tangent hyperplane at $\theta^\uparrow=(\theta,F(\theta))$ to $\calF$ coincides with the tangent hyperplane at  $(\eta,\inner{\eta}{\eta})$ of the paraboloid $\calP$ 
shifted vertically by amount $w=\frac{1}{4}\|\eta\|^2+F^*(\eta)$.
\end{Proposition}

In particular, when $F(\theta)=Q(\theta)=\inner{\theta}{\theta}$, we check that we have $\eta=2\theta$ and $\hat{p}=(\theta,0)$, and the quadratic Bregman Voronoi diagram amounts to a power diagram with zero weights, i.e., an ordinary Voronoi diagram.

%\begin{Remark}
More generally, instead of taking the paraboloid reference $\calP$ for the lifting transform, we may consider the reference potential $\calF_0=\{(\theta,F_0(\theta)) \st \theta\in\Theta_0\}$ for a Bregman generator $F_0$.
The tangent hyperplane to $\calF_0$ at $(\theta,F_0(\theta))$ shifted by $w$ has equation:
$$
H_{\theta}^{\calF_0}: \inner{x-\theta}{\nabla F_0(\theta)} + F_0(\theta) +w.
$$ 
We then seek to express $\theta_0$ and $w$ in order to match $H_{\theta_0}^{\calF_0}$ with $H_\theta^\calF$.
We get $\theta_0=\nabla F_0^*(\nabla F(\theta))=\nabla F_0^*(\eta)$ (with $\nabla F_0^*=(\nabla F_0)^{-1}$) and 
$w=F^*(\eta)+\inner{\theta_0}{\nabla F_0(\theta_0)}+F_0(\theta_0)$ (Figure~\ref{fig:equivalencetangenthyperplaneH}, right).
Thus the BVD with respect to $F$ can be obtained as an additively weighted BVD with respect to a reference Bregman generator $F_0$ 
provided that $\calH=\{\nabla F(\theta)\st\theta\in\Theta\} \subset \calH_0=\{\nabla F_0(\theta)\st\theta\in\Theta_0\}$.
Notice that in the case when $F_0(\theta)=Q(\theta)$ (paraboloid), we have $\calH_0=\bbR^m$, the full domain.
%\end{Remark}

\begin{Corollary}
The Bregman Voronoi diagram with generator $F$ of a set of  $\calT=\{\theta_i\}$ of $n$ parameters can be calculated from the additively weighted Bregman Voronoi diagram  with generator $F_0$ of a set of weighted parameters $\calT_0=\{(\nabla F_0^*(\nabla F(\theta_i)),w_i)\}$.
\end{Corollary}

\begin{figure}%
\centering
\begin{tabular}{cc}
\fbox{\includegraphics[width=0.4\columnwidth]{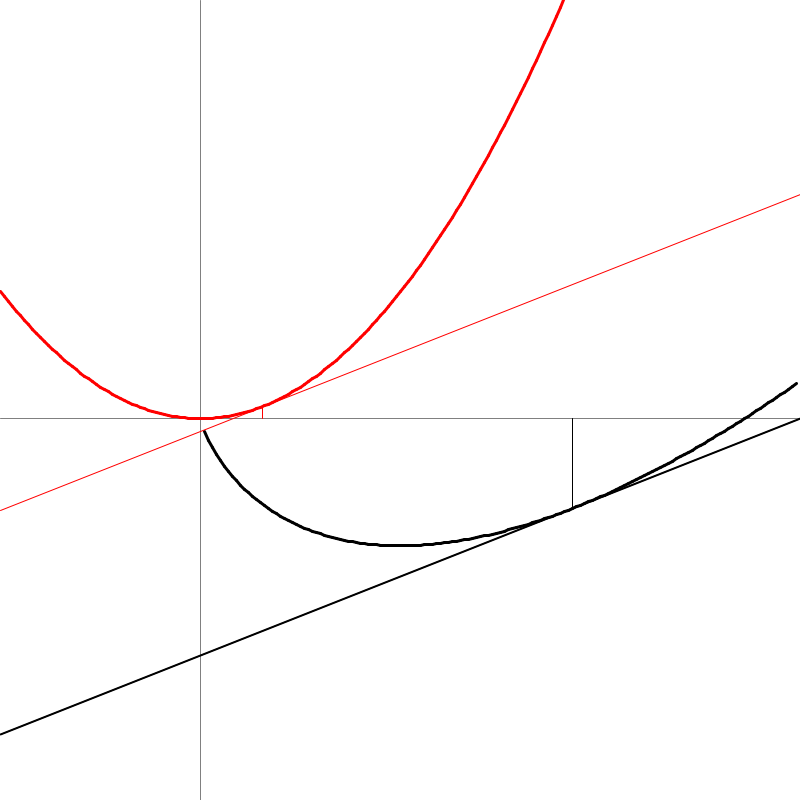}} &
\fbox{\includegraphics[width=0.4\columnwidth]{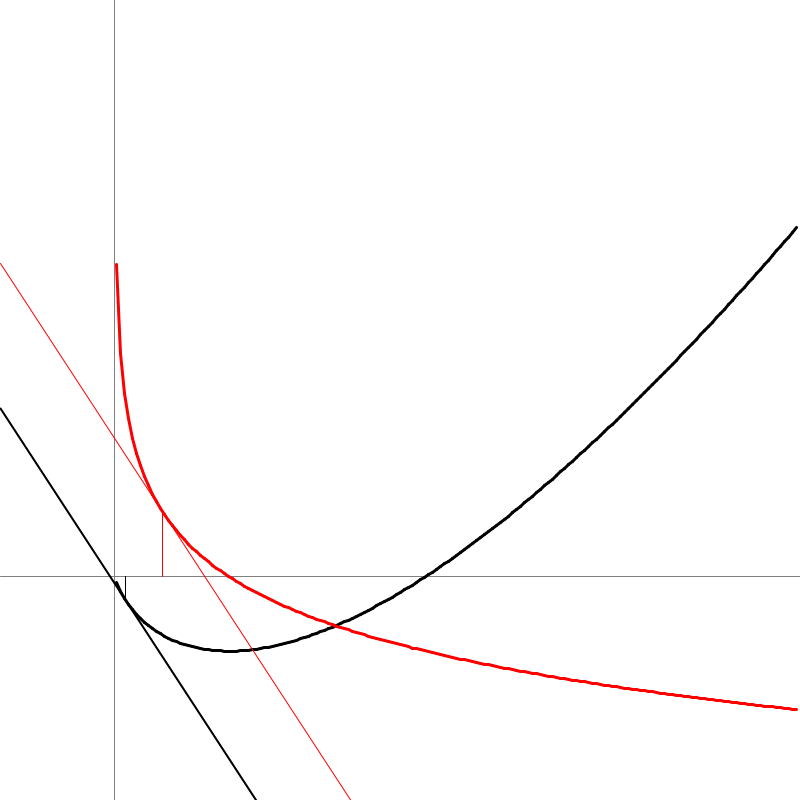}}
\end{tabular}
\caption{Left: Equivalence of hyperplanes $H^\calP_{\hat p}=H^\calF_\theta$ illustrated for the case of the Shannon negentropy Bregman generator: $F(\theta)=\theta\log\theta-\theta$. Right: Equivalent of hyperplanes $H^{\calF_0}_{\hat p}=H^\calF_\theta$ for the case of $F_0(\theta)=-\log\theta$ (Burg negentropy) and $F(\theta)=\theta\log\theta-\theta$. 
 }%
\label{fig:equivalencetangenthyperplaneH}%
\end{figure}

%\begin{Remark}
%The symmetrized Bregman divergence
%\begin{eqnarray}
%S_F(\theta,\theta') &\eqdef& B_F(\theta:\theta')+B_F(\theta':\theta),\\
%&=&\inner{\theta-\theta'}{\nabla F(\theta)-\nabla F(\theta')}
%\end{eqnarray} 
%for a $d$-dimensional Bregman generator $F(\theta)$ is  generally not a Bregman divergence (except for quadratic Bregman generators) but a curved Bregman divergence~\cite{nielsen2025curved} for the $2d$-dimensional Bregman generator $\bar F(\xi)=\bar F(\theta,\eta')=F(\theta)+F^*(\eta')$ for $\xi=(\theta,\eta')\subset\bbR^{2d}$.
%The symmetrized Bregman Voronoi diagram can thus be calculated from a $2d$-dimensional BVD (or equivalent PD) cut by the $d$-dimensional set $\calU=\{\xi=(\theta,\nabla F(\theta))\st \theta\in\Theta\}$:
%$$
%\Vor_F^{\mathrm{sym}}(\{\theta_i\})=\Vor_{\bar{F}}(\{\xi_i\})\cap\calU, \quad \xi_i=(\theta_i,\eta_i).
%$$
%\end{Remark}

%% file: sec-SpaceSphere.tex
%%%%%%%%%%%%%%%%%%%%%
\section{Space of Bregman spheres and orthospheres}\label{sec:spacesphereortho}
%%%%%%%%%%%%%%%%%%%%%

%%%%%%%%%%%%
\subsection{Outline}
%%%%%%%%%%%%

In this section, we first describe the spaces of Bregman spheres~\cite{BVD-2010} in \S\ref{sec:spacesphere}.
We then consider the Bregman and power circumspheres passing through   given sets of support Bregman parameters or power points in \S\ref{sec:circumsphere}: 
We explain how to convert Bregman circumsphere problems as power circumsphere (called orthospheres) problems in Proposition~\ref{prop:equivcircumcenter}.
For the special 1D case, we report exact formula for the Bregman circumcenter and Bregman radius in Proposition~\ref{prop:BMEB1D},  and interpret those quantities geometrically from the potential function graph in \S\ref{sec:BregRad1D}.
We conclude this section by discussing connections between minimum enclosing balls with respect to Bregman divergences and power distances in \S\ref{sec:meb}.

%%%%%%%%%%%%
\subsection{Space of Bregman spheres}\label{sec:spacesphere}
%%%%%%%%%%%%

The  space of spheres was studied  in computational geometry in~\cite{devillers1992space} and further extended to 
the space of Bregman spheres in~\cite{BVD-2010}. 
We recall quickly the basics.
Let $z(\theta)=\theta^{\uparrow} \eqdef (\theta,y(\theta)=F(\theta))$ be the lifted point on $\calF$ and $z(\theta)\downarrow=\theta$ denote the vertical projection.
A (right-sided) Bregman sphere $S=S_F^r(\theta_C,R)=\{\theta \st B_F(\theta:\theta_C)=R\}$ can be lifted on the potential function $\calF$ 
as $S^\uparrow=\{\theta^{\uparrow} \st \theta\in S\}$.
The lifted sphere $S^\uparrow$ is lying on a hyperplane $H_S$ of equation
$$
H_S: y=\inner{\theta-\theta_C}{\eta_C} + F(\theta) + R.
$$
Conversely, any non-vertical hyperplane 
$H:y=\inner{a}{\theta} + b$ for $a\in\bbR^m$ and $b\in\bbR$ intersects the potential function  $\calF$ in a lifted sphere $H\cap\calF$ such that $\downarrow (H\cap\calF)$ is the sphere $S_F(\theta_C,R)$ with:
\begin{eqnarray*}
\theta_C &=& \nabla F^{-1}(a)=\nabla F^*(a),\\
R &=& \inner{a}{\theta_c}-F(\theta_c)+b =F^*(a)+b.
\end{eqnarray*}

Notice that for a sphere $S$ with imaginary radius $R\in\bbC$, the corresponding hyperplane $H_S$ does not intersect the potential function, 
yet the bijection between $m$-dimensional spheres to $(m+1)$-dimensional hyperplanes is still valid~\cite{devillers1992space,cheng2013delaunay}.

%%%%%%%%%%%%
\subsection{Circumscribing Bregman spheres and orthospheres}\label{sec:circumsphere}
%%%%%%%%%%%%

Let $\calV=\{\theta_1,\ldots,\theta_k\}$ be $k$ parameters in general position, with $k\in\{2,\ldots,m+1\}$.
We seek the unique Bregman sphere $S$ passing through those parameters of $\calV$.
Assume wlog. that $k=m+1$ (otherwise, consider the affine space defined by the span of $\calV$).

The right Bregman circumcenter~\cite{BVD-2010} $\theta_C$ can be expressed as a solution of a determinant equation
$\BregmanCircumSphere(\theta_C;\calV)=0$ where

\begin{equation}
\BregmanCircumSphere(\theta;\calV)\eqdef \det \left(
\left[
\begin{array}{cccc}
1 & \ldots & 1 & 1 \cr
\theta_1 & \vdots & \theta_{m+1} & \theta \cr
F(\theta_1) & \ldots & F(\theta_{m+1}) & F(\theta)
\end{array}
\right]
\right).
\end{equation}

%The parameters are said in general position when the matrix
%$$
%\left[
%\begin{array}{lll}
%1 & \vdots & 1  \cr
%\theta_1 & \vdots & \theta_{m+1}  \cr
%F(\theta_1) & \vdots & F(\theta_{m+1})
%\end{array}
%\right]
%$$
%is full rank $m+1$. Otherwise, that means we have Bregman co-spherical parameters in $\calV$.

Once the Bregman circumcenter $\theta_C$ is calculated, we retrieve the Bregman circumcenter radius 
$R=B_F(\theta_C:\theta_i)$ for any $i\in [m+1]$.

We define the orientation  of points by the following  $(m+1)\times (m+1)$ matrix  determinant:
\begin{equation}
\Orient(p_1,\ldots,p_{m+1})=\left|\begin{array}{cc}p_1 & 1\\ \vdots & \vdots \\ p_{m+1} & 1\end{array}\right|.
\end{equation}
A parameter $\theta$ is inside the ball with boundary $S_F^r(\theta_C,R)$ when 
$$
\sign(\Orient(\calV)\, \BregmanCircumSphere(\theta;\calV)) <0,
$$
 and  outside when 
$$
\sign(\Orient(\calV)\, \BregmanCircumSphere(\theta;\calV)) >0.
$$
Parameter $\theta$ lies on the Bregman circumsphere when $\BregmanCircumSphere(\theta;\calV)=0$.

Let us revisit the construction of  Bregman circumscribing spheres from the viewpoint of orthospheres~\cite{cheng2013delaunay} which are circumscribing spheres with respect to the power distance.
Recall that the left-sided Bregman circumcenter $\theta_C$ is defined as
$B_F(\theta_C:\theta_i)=B_F(\theta_C:\theta_j)$ or equivalently using the power distance representation of BDs as 
$\sigma(\theta_C:\weta_i)= \sigma(\theta_C:\weta_j)$.

Let us define the orthosphere  $\Orthosphere(\hat\calV; \hat{p})$ with support vectors in $\hat\calV=\{\hat{p}_1,\ldots,\hat{p}_{m+1}\}$ as follows:

\begin{eqnarray*}
\lefteqn{\Orthosphere(\hat\calV=\{\hat{p}_1=(p_1,w_1),\ldots,\hat{p}_{d+1}=(p_{d+1},w_{m+1}))\}; \hat{p})
= }\\
&&\mathrm{sign}\left(\mathrm{Orient}(p_1,\ldots,p_{m+1}) \, 
\left|\begin{array}{lll}p_1 & \hat p_1^+ & 1\\ \vdots & \vdots &\vdots \\ 
p_{m+1} & \hat p_{m+1}^+ & 1\\
p & \hat{p}^+ & 1\end{array}\right|
\right),
\end{eqnarray*}
where $\hat p_i^+ \eqdef \|p_i\|^2-w_i$ is the power lifting operator (see the  power test in~\cite{devillers2011perturbations,gavrilova2000exact}).
Notice that the determinant equations are deep set formula~\cite{DeepSetFormula-2025}, i.e., formula invariant to permutations of their arguments.

%Then weighted point $\hat{p}$ is orthogonal to all weighted points of $\hat\calV$ when $\Orthosphere}(\calB=\{\hat{p}_1=(p_1,w_1),\ldots,\hat{p}_{d+1}=(p_{d+1},w_{d+1}))\}; \hat{q})=0$, is suborthogonal when $>0$, and superorthogonal when $<0$. 

Hence, we seek equivalently the power circumcenter called the orthocenter of the orthosphere passing through $\check\calV=\{\weta_i=(\eta_i=\nabla F(\theta_i),w_{F^*}(\eta_i))\}_i$.

Let $\hat{c}=(c,\rho)$ be the power circumscribing sphere with power center $c$  and power radius $\rho\in\bbR$ (orthoradius).
Since $\sigma(c,\hatp_i)=\sigma(c,\hatp_j)$, we also have  $\sigma(\hat{c},\hatp_i)=\sigma(\hat{c},\hatp_j)$.
That is, the power circumsphere is orthogonal to the power spheres of $\hat\calV$ (Figure~\ref{fig:orthogonalspheres}).
Henceforth, its name: the orthosphere.
When the orthoradius is negative, the orthosphere has imaginary radius $\rho$ in $\bbC$.

\begin{figure}
%\centering
\includegraphics[width=0.5\textwidth]{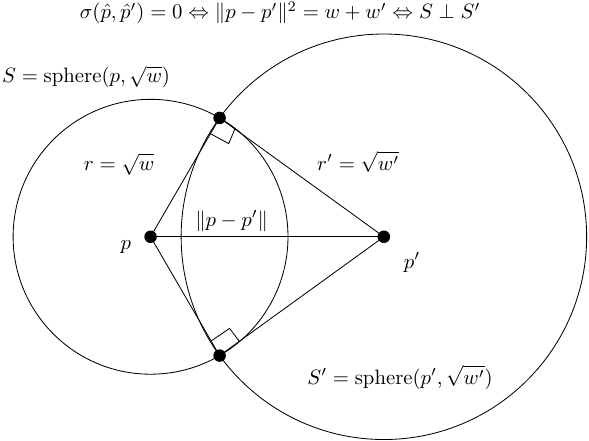}
\caption{Two weighted points $\hat p=(p,w)$ and $\hat p'=(p',w')$ with vanishing power distance have their corresponding spheres $S=\sphere(p,\sqrt{w})$ and $S'=\sphere(p',\sqrt{w'})$ orthogonal to each other (Pythagoras' theorem).}\label{fig:orthogonalspheres}
\end{figure}

\begin{Proposition}\label{prop:equivcircumcenter}
Let $\theta_C$ be the unique Bregman circumcenter passing through the parameters of $\calV=\{\theta_1,\ldots,\theta_{m+1}\}$.
Then  $\theta_C$ coincides with the power circumcenter of the power circumsphere $\hat{c}=(c,\rho)$ passing through the weighted points of  $\check\calV=\{\weta_i=(\eta_i=\nabla F(\theta_i),w_{F^*}(\eta_i))\}_i$.
\end{Proposition}

Figure~\ref{fig:ortho} displays some examples of orthospheres passing through two and three weighted points in the plane.

\begin{figure}
\includegraphics[width=0.2\textheight]{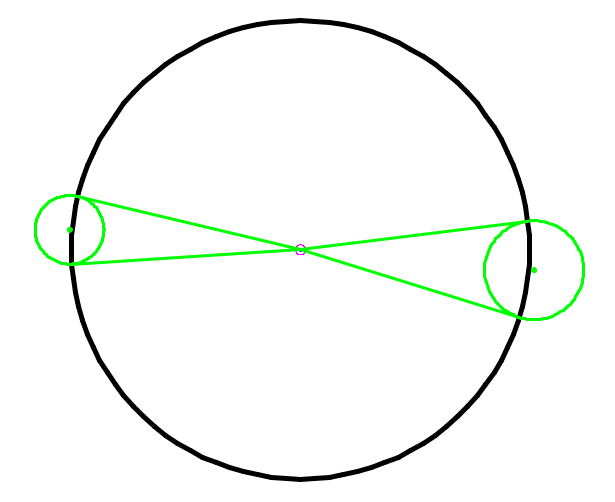}
\includegraphics[width=0.2\textheight]{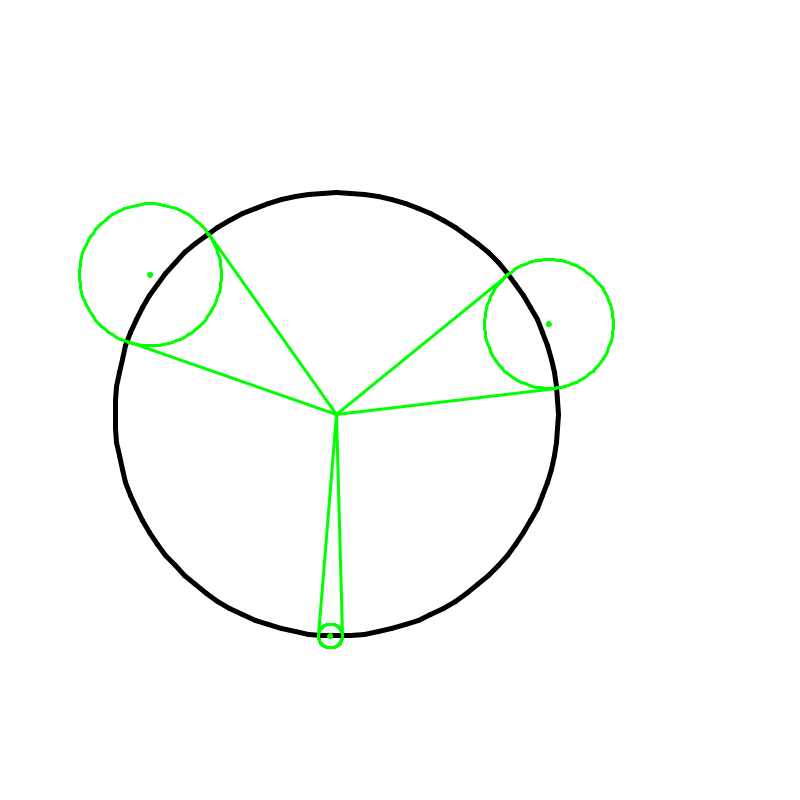}
\caption{Orthosphere of two and three weighted points $\hatp_i$ visualized as green circles $S(\hatp_i)$.
The green line segments joining the power circumcenter (orthocenter) to the circles have all the same Euclidean lengths, and the orthosphere intersects orthogonally the support green circles.
\label{fig:ortho}}
\end{figure}

%%%%%%%%%%%%%
\subsection{Closed-form formula for the 1D case: Bregman and power circumscribing  intervals}\label{sec:cf1d}
%%%%%%%%%%%%%

\subsubsection{A formula for the 1D Bregman radius}\label{sec:BregRad1D}

In 1D, the minimum enclosing Bregman interval of two parameters $\theta_1$ and $\theta_2$
satisfies
 $B_F(\theta:\theta_1)=B_F(\theta:\theta_2)$ (i.e., $Y_F(\theta^*:\eta_1)=Y_F(\theta^*:\eta_2)$) where $\eta_1=\nabla(\theta_1)$ and $\eta_2=\nabla(\theta_2)$.
That is, we have

$$
F^*(\eta_1)-\theta^*\eta_1 = F^*(\eta_2)-\theta^*\eta_2,
$$
from which we obtain the Bregman circumcenter:
\begin{equation}
\theta^*=\frac{F^*(\eta_2)-F^*(\eta_1)}{\eta_2-\eta_1}.
\end{equation}

Let $\theta^*=\theta_1+\alpha^*(\theta_2-\theta_1)=(1-\alpha^*)\theta_1+\alpha^*\theta_2=:(\theta_1\theta_2)_{1-\alpha^*}$ for some interpolation coefficient $\alpha^*\in [0,1]$.
We get
\begin{equation}\label{eq:alphaeta}
\alpha^*= \frac{\frac{F^*(\eta_2)-F^*(\eta_1)}{\eta_2-\eta_1}-\theta_1}{\theta_2-\theta_1}.
\end{equation}

The corresponding Bregman radius is
$$
R^*=B_F(\theta^*:\theta_1)=B_F(\theta^*:\theta_2).
$$

Since $\theta^*-\theta_1=\alpha^*(\theta_2-\theta_1)$ and 
$\theta^*-\theta_2=(1-\alpha^*)(\theta_1-\theta_2)$, we have $B_F(\theta^*:\theta_1) = B_F(\theta^*:\theta_2)$ which yields the equation
$$
F(\theta_2)-F(\theta_1)+(\theta_1-\theta_2)(\alpha^*\eta_1+(1-\alpha^*)\eta_2)=0.
$$

Thus we find an alternative expression for $\alpha^*$ (compare with Eq.~\ref{eq:alphaeta}):

\begin{equation}
\alpha^*=\frac{\frac{F(\theta_1)-F(\theta_2)}{\theta_1-\theta_2}-\eta_2}{\eta_1-\eta_2}
\end{equation}

Geometrically speaking, $\theta^*$ can be found by calculating the unique intersection point of the two tangent lines $T_{\theta_1}$ and $T_{\theta_2}$ of $\calF$ as depicted in Figure~\ref{fig:alphastar}.

\begin{figure}
%\centering
\includegraphics[width=0.4\textwidth]{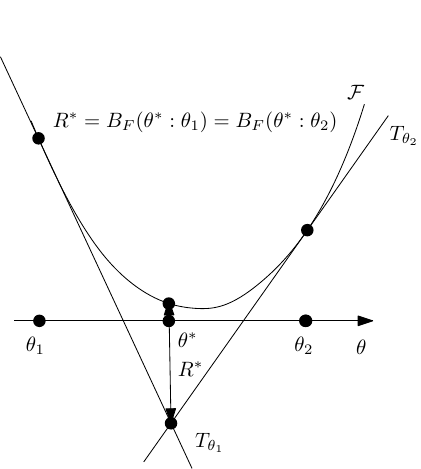}
\caption{Geometric interpretation of $\alpha^*$ and the Bregman radius $R^*$ interpreted from the intersection point of the tangent lines to the potential graph at $\theta_1$ and $\theta_2$.}\label{fig:alphastar}
\end{figure}

Dually, it comes from the equality $B_F(\theta^*:\theta_1)=B_F(\theta^*:\theta_2)$ by using the dual Bregman divergences $B_{F^*}$ that
we have $B_{F^*}(\eta_1:\eta^*)=B_{F^*}(\eta_2:\eta^*)$,
where $\eta^*=F'(\theta^*)=:\eta_1+\beta^\star(\eta_2-\eta_1)$ for some $\beta^\star=\frac{\eta^*-\eta_1}{\eta_2-\eta_1}\in [0,1]$ (in general different from $\alpha^*$).
This derivation is very specific to the 1D case: Indeed, $\eta^*=\nabla F(\theta^\star)$ is in general not lying on the line segment $[\eta_1,\eta_2]$ when the dimension is greater than one.
But this always happens in 1D, hence the introduction of the linear interpolation coefficient $\beta^\star$.

Therefore we have 
\begin{equation}
R^*=B_{F^*}(\eta_1:\eta^\star)=B_{F^*}(\eta_2:\eta^\star)=J_{F^*,\beta^\star}(\eta_1,\eta_2),
\end{equation}
where $J_{F^*,\beta^\star}$ is the following skewed Jensen divergence:

$$
J_{F^*,\beta^\star}(\eta_1,\eta_2) := (1-\beta^\star)F^*(\eta_1)+\beta^\star F^*(\eta_2)-F^*(\eta_1+\beta^\star (\eta_2-\eta_1)).
$$

A geometric interpretation of $\beta^\star$ and the Bregman radius $R^*$ is given in Figure~\ref{fig:betastar}.

\begin{figure}
%\centering
\includegraphics[width=0.6\textwidth]{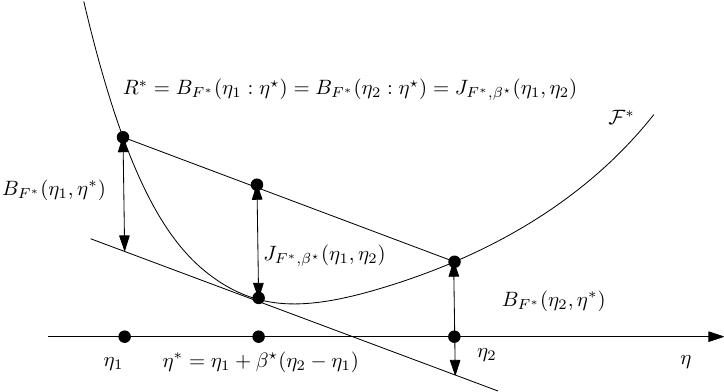}
\caption{Geometric interpretation of $\beta^\star$ and the Bregman radius $R^*$ interpreted as a skewed Jensen divergence.}\label{fig:betastar}
\end{figure}

\begin{Property}\label{prop:BMEB1D}
Let $F$ be a Legendre-type convex function with convex conjugate $F^*$ and $\theta_1',\ldots, \theta_n'$ be $n$ (unsorted) scalar parameters with $\theta_1=\min_i\theta_i'$ and $\theta_2=\max_i \theta_i'$, $\eta_1=F'(\theta_1)$ and $\eta_2=F'(\theta_2)$.
Then the circumcenter $\theta^*$ and the radius $R^*$ of the minimum enclosing Bregman  ball (interval) is given by
\begin{eqnarray*}
\theta^* &=& \frac{F^*(\eta_2)-F^*(\eta_1)}{\eta_2-\eta_1},\\
R^* &=& J_{F^*,\beta^\star}(\eta_1,\eta_2),
\end{eqnarray*}
where $\beta^\star=\frac{\eta^*-\eta_1}{\eta_2-\eta_1}$ for $\eta^*=F'(\theta^*)$ and 
$J_{F^*,\beta^\star}(\eta_1,\eta_2) := (1-\beta^\star)F^*(\eta_1)+\beta^\star F^*(\eta_2)-F^*(\eta_1+\beta^\star (\eta_2-\eta_1))$.
\end{Property}

%%%%%%%%%%%%%%
\subsubsection{Equivalence with power circumcenters}
%%%%%%%%%%%%%%

Now, consider the weighted gradient point $\weta_i'=(\eta_i=F'(\theta_i),w_i^*=w_{F^*}(\eta_i))$ where $w_i^*=\eta_i^2-2F^*(\eta_i)$.
The left Bregman MEB amounts to an equivalent power MEB, and the power MEB circumcenter $\theta^*$ is such that
\begin{eqnarray*}
\sigma(\theta_C,\weta_1) &=& \sigma(\theta_C,\weta_2),\\
(\theta_C-\eta_1)^2-w_1^* &=& (\theta_C-\eta_2)^2-w_2^*.
\end{eqnarray*}

We get
\begin{equation}
\theta_C=\frac{1}{2} \left( \frac{w_1^*-w_2^*}{\eta_2-\eta_1}+ \eta_1+\eta_2\right).
\end{equation}

We have
$$
\rho=\frac{1}{2}\, \sigma(\wtheta_C,\weta_1)=\frac{1}{2}\, \sigma(\wtheta_C,\weta_2)=R.
$$
Notice that this equality of power/Bregman radii holds for $\wtheta_C$ but not for $\theta_C$.

%%%%%%%%%%%%
\subsection{Minimum Enclosing Balls}\label{sec:meb}
%%%%%%%%%%%%

The Euclidean minimum enclosing ball (MEB) of a point set and the Bregman MEBs for parameter sets are always unique~\cite{ASEBB-2005}.
The Euclidean MEB can be solved by quadratic programming and the Bregman MEBs by convex programming~\cite{EIG-2020}.
Bregman MEBs which include the Euclidean MEB as a particular case can be solved exactly using a LP-type algorithm~\cite{SEBB-2008} (practical in moderate dimensions, say up to $m\approx 50$) or approximated numerically using Frank--Wolfe techniques~\cite{clarkson2010coresets} which yield core-sets (and further can be kernelized~\cite{sifa2020novelty}).

Since the Kullback-Leibler (KL) divergence $D_\KL(p_{\theta_1}:p_{\theta_2})$ between two densities of an exponential family $\calE$ amounts to a reverse Bregman divergence $B_F(\theta_2:\theta_1)$ where $F$ is the cumulant function of $\calE$, minimum left/right enclosing KL balls amount to right/left Bregman MEBs.
The following example illustrates the calculation of the KL MEB of Gaussians as corresponding Bregman MEB.

\begin{Example}
The family of univariate Gaussian densities
 $p(\mu,v=\sigma^2)=\frac{1}{\sqrt{2\pi v}}\, \exp\left(\frac{(x-\mu)^2}{2 v}\right)$
 form an exponential family of order $2$ with 
natural parameters $\theta(\mu,v)=\left(\theta_1=\frac{\mu}{v},\theta_2=-\frac{1}{2v}\right)$,
 cumulant function $F(\theta)=-\frac{\theta_1^2}{4\theta_2}+\frac{1}{2}\log(-\pi/\theta_2)$,
 dual moment parameters $\eta(\theta)=\nabla F(\theta)=(-\frac{\theta_1}{2\theta_2},\frac{\theta_1^2}{4\theta_2^2}-\frac{1}{2\theta_2})$ 
($\eta(\mu,v)=(\eta_1=\mu=E[t_1(x)],\eta_2=E[t_2(x)]=\mu^2+v)$  with sufficient statistic vector $t(x)=(t_1(x)=x,t_2(x)=x^2)$),
inverse gradient $\theta(\eta)=\nabla F^*(\eta)=(\frac{\eta_1}{\eta_2-\eta_1^2)},-\frac{1}{2(\eta_2-\eta_1^2)})$ and
convex conjugate $F^*(\eta)=\inner{\theta}{\eta}-F(\theta))=-\frac{1}{2}\log 2\pi e\left(\eta_2-\eta_1^2\right)$.
Thus the right KL MEB of a set of univariate Gaussian distributions amount to the left Bregman MEB.
See Figure~\ref{fig:SEBBGaussian}. % for a left-sided Bregman MEB that amounts to the right-sided KL MEB.
\end{Example}

The Kullback-Leibler MEB is often used machine learning, e.g., in information-geometric reinforcement learning based on maximizing mutual information~\cite{KLMEB-SkillLearning-2022}.
MEBs with respect to $\alpha$-divergences were numerically solved in~\cite{ChebyshevAlphaDiv-2020}.

 %\begin{algorithm}
%\caption{Bregman-Welzl's randomized minimum enclosing ball algorithm}\label{alg:welzl}
%\KwIn{Finite point set $\{\theta_i\}\subset\mathbb{R}^m$, boundary set $\calB$}
%\KwOut{The minimum enclosing ball of $\calB$}
%
%\SetKwFunction{BregmanMiniBall}{BregmanMiniBall}
%\SetKwProg{Fn}{Function}{}{}
%
%\Fn{\Welzl{$\calV,\calB$}}{
%
%\If{$\calP=\emptyset$ \textbf{or} $|\calB|=n+1$}{
    %\Return $\mathrm{BregmanCircumscribingSphere}(\cal{B})$\;
%}
%
%Choose a point $\theta\in \calP$ uniformly at random\;
%
%$B\leftarrow$\Welzl{$\calP\setminus\{p\},\calB$}\;
%
%\If{$p\in B$}{
    %\Return $B$\;
%}
%
%\Return \BregmanMiniBall{$\calP\setminus\{p\},\calB\cup\{p\}$}\;
%
%}
%\end{algorithm}

\def\picw{0.255\textwidth}
\begin{figure}
%\centering

\begin{tabular}{ll}
Natural parameter $\theta$ & Gaussian density $N(\mu,\sigma)$\\
\fbox{\includegraphics[width=\picw]{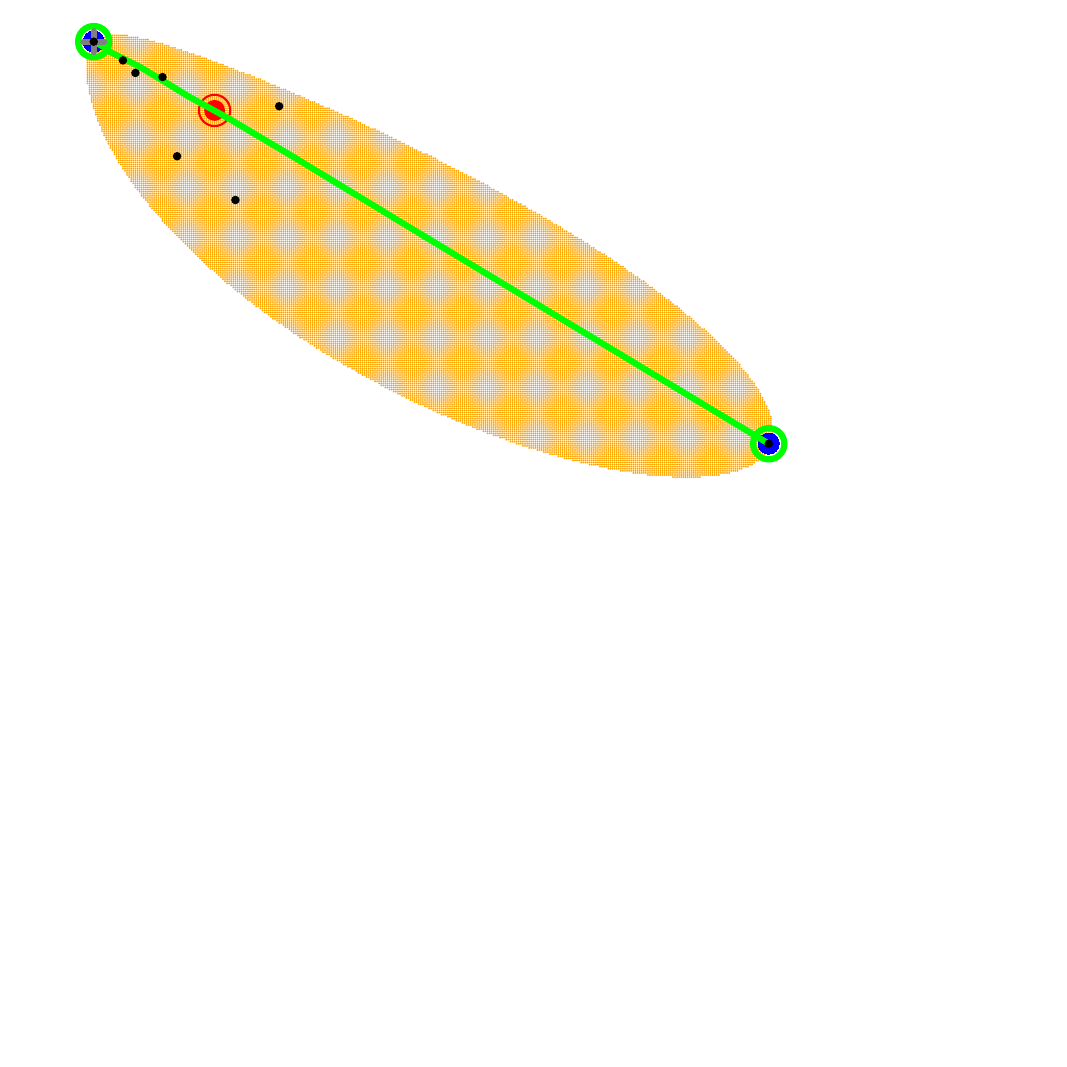}} &
\fbox{\includegraphics[width=\picw]{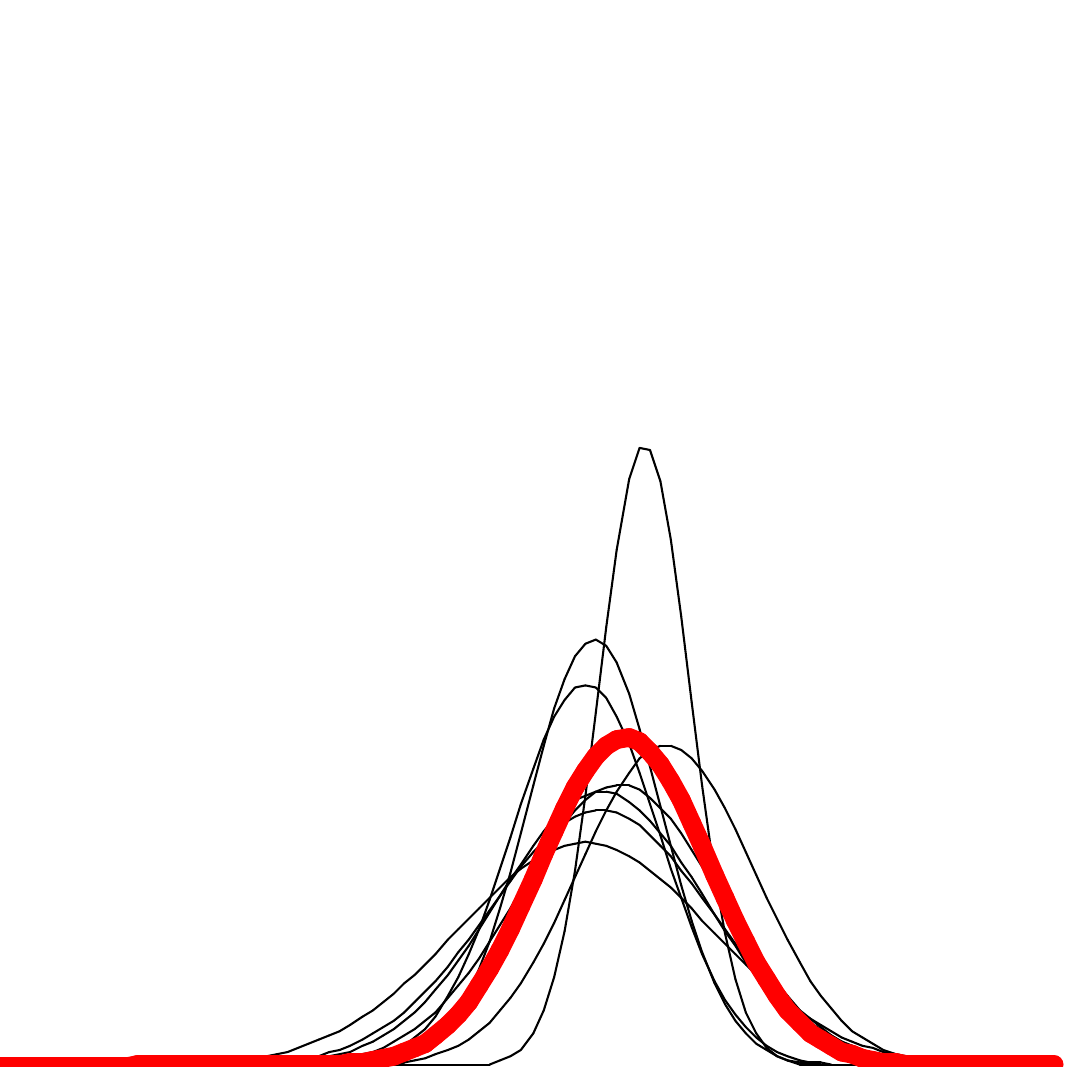}} \\
\fbox{\includegraphics[width=\picw]{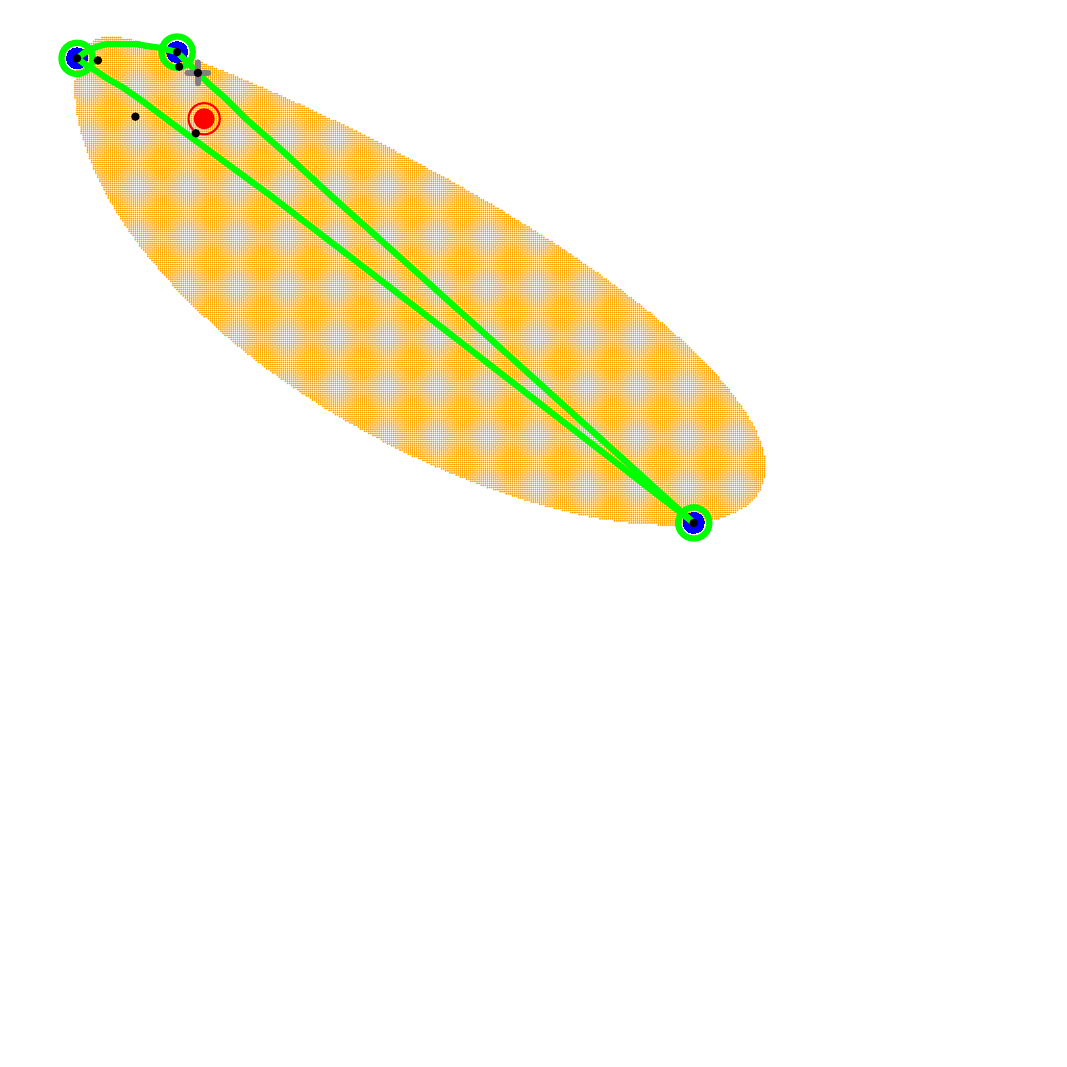}} &
\fbox{\includegraphics[width=\picw]{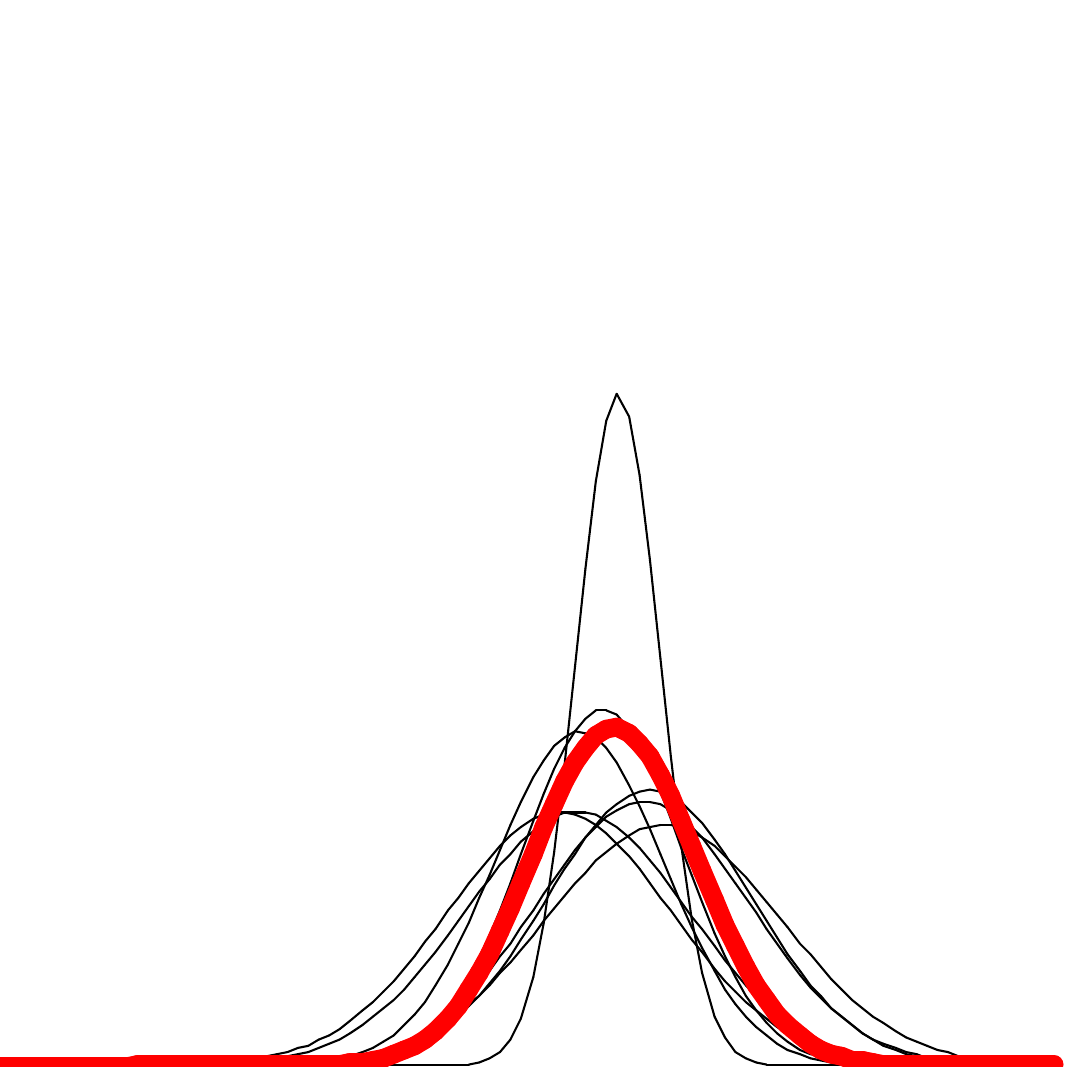}}  
\end{tabular}

\caption{The right KL minimax center of a set of  
$n=8$ univariate Gaussians (right column) is computed equivalently as a left Bregman circumcenter of their corresponding natural parameters $\theta=(\frac{\mu}{\sigma^2},-\frac{1}{2\sigma^2})$.
The top row shows an example where the Bregman MEB is defined by two support distributions while the bottom row displays an example with the Bregman MEB defined by three support distributions.
}\label{fig:SEBBGaussian}
\end{figure}

Figure~\ref{fig:compare} shows some examples of power MEBs (orthospheres) computed by both the exact Welzl LP-type algorithm~\cite{Miniball-1991,SEBB-2008} and by the Frank--Wolfe coreset approximation algorithm~\cite{clarkson2010coresets}.

\begin{figure}
%\centering
\fbox{\includegraphics[width=0.4\textwidth]{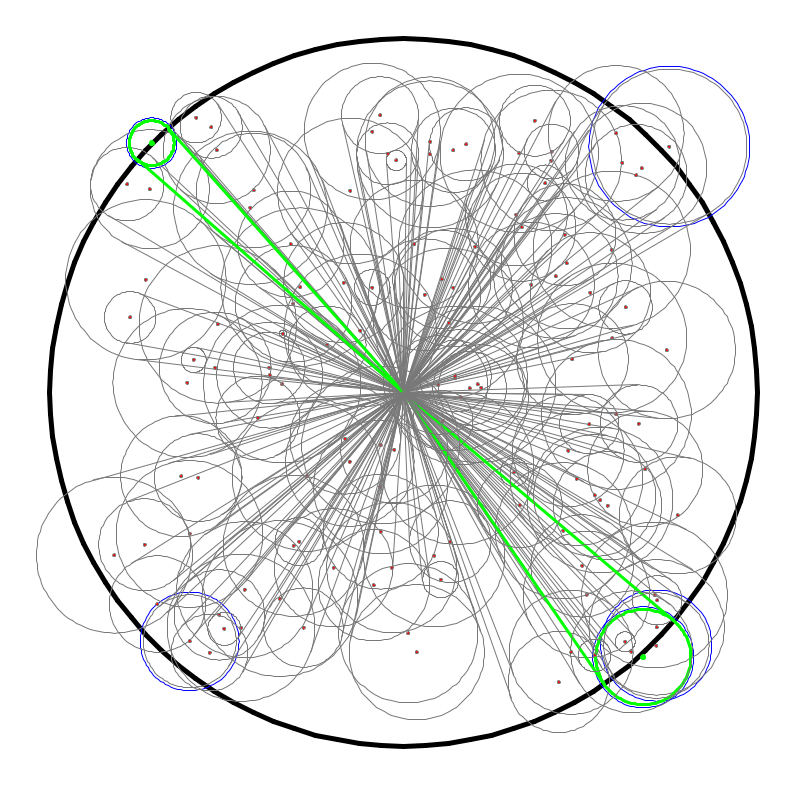}}
\fbox{\includegraphics[width=0.4\textwidth]{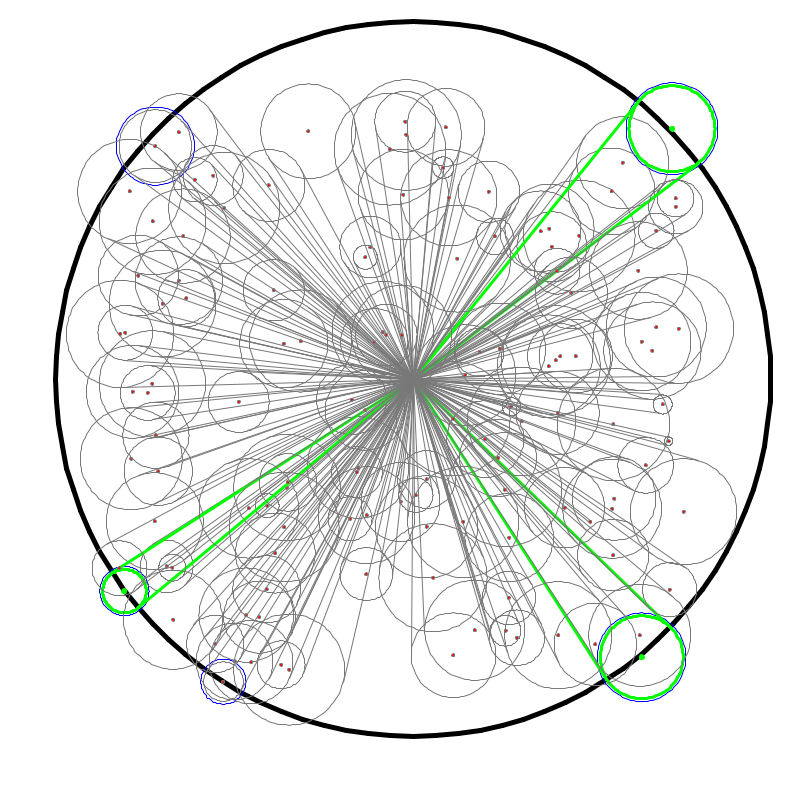}}
\caption{Power MEBs calculated with Welzl/LP-type Miniball algorithm (basis spheres displayed in green) and with Frank--Wolfe approximation algorithm (coreset spheres displayed in blue) for a set of $n=128$ spheres/weighted points.\label{fig:compare}}
\end{figure}

However, let us notice that the Bregman MEB problem (convex programming) is different from the power MEB problem (quadratic programming) because
$$
\argmin_\theta B_F(\theta:\theta_i) = \argmin_\theta \sigma(\wtheta:\weta_i) \not=    \argmin_\theta \sigma(\theta:\weta_i).
$$
Indeed, in the right-hand-side term, we dropped the non-linear $w_F(\theta)$ term.
However, we can relax the Bregman MEB problem to the power MEB problem and calculate $\tilde{\theta}_C= \argmin_\theta \sigma(\theta:\weta_i)$.
This allows us to find an upper bound on the Bregman radius $R\leq \max_i B_F(\tilde{\theta}_C:\theta_i)$.
See Figure~\ref{fig:MEBex3} which illustrates an example of connections between the Bregman MEB and the power MEB with dually parameterized weighted points shown as red/blue circles.

\begin{figure}

\begin{tabular}{lll}
% $|\calB|=2$ & $|\calB|=3$  & $|\calB|=3$ \\ \hline
\fbox{\includegraphics[width=\picw]{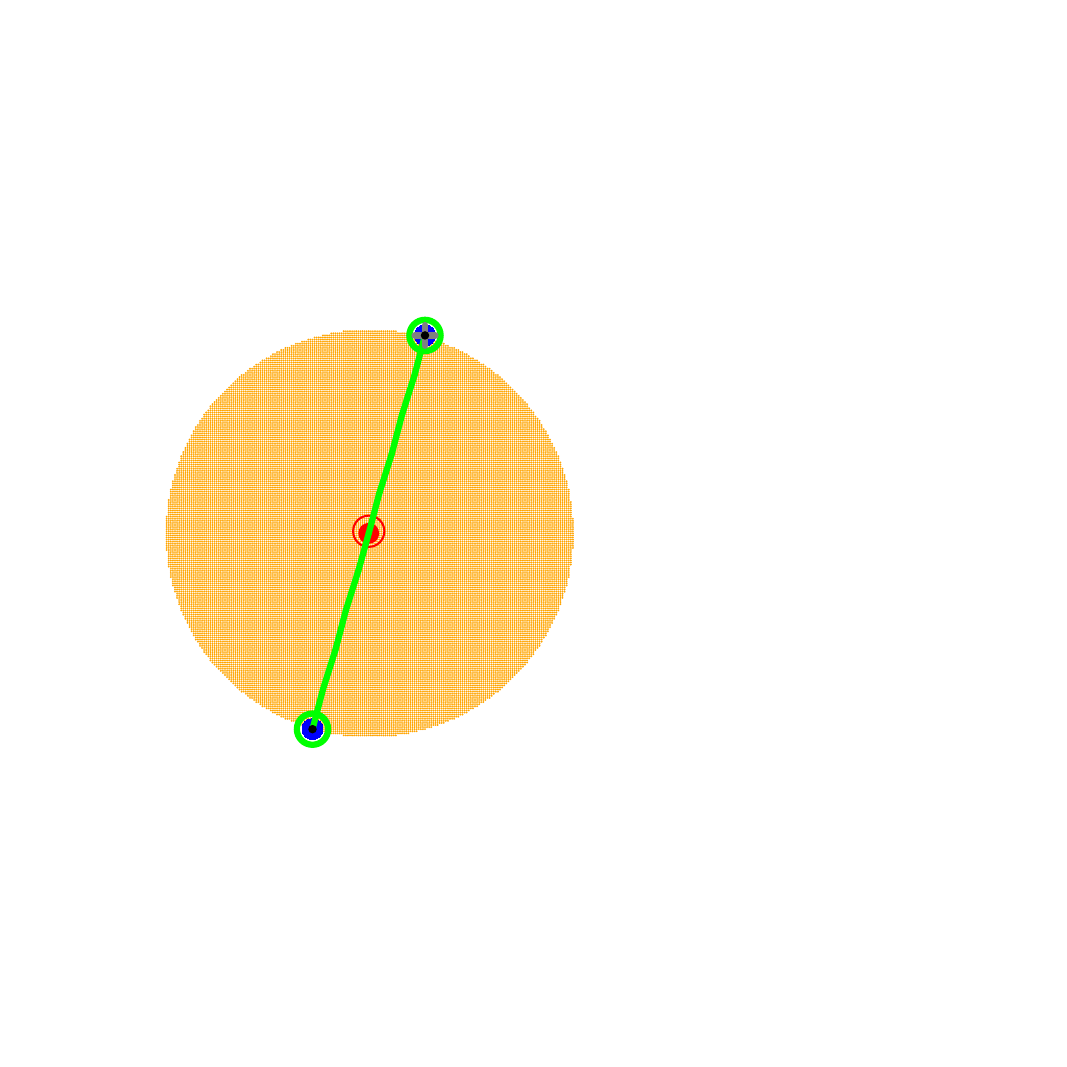}} & 
\fbox{\includegraphics[width=\picw]{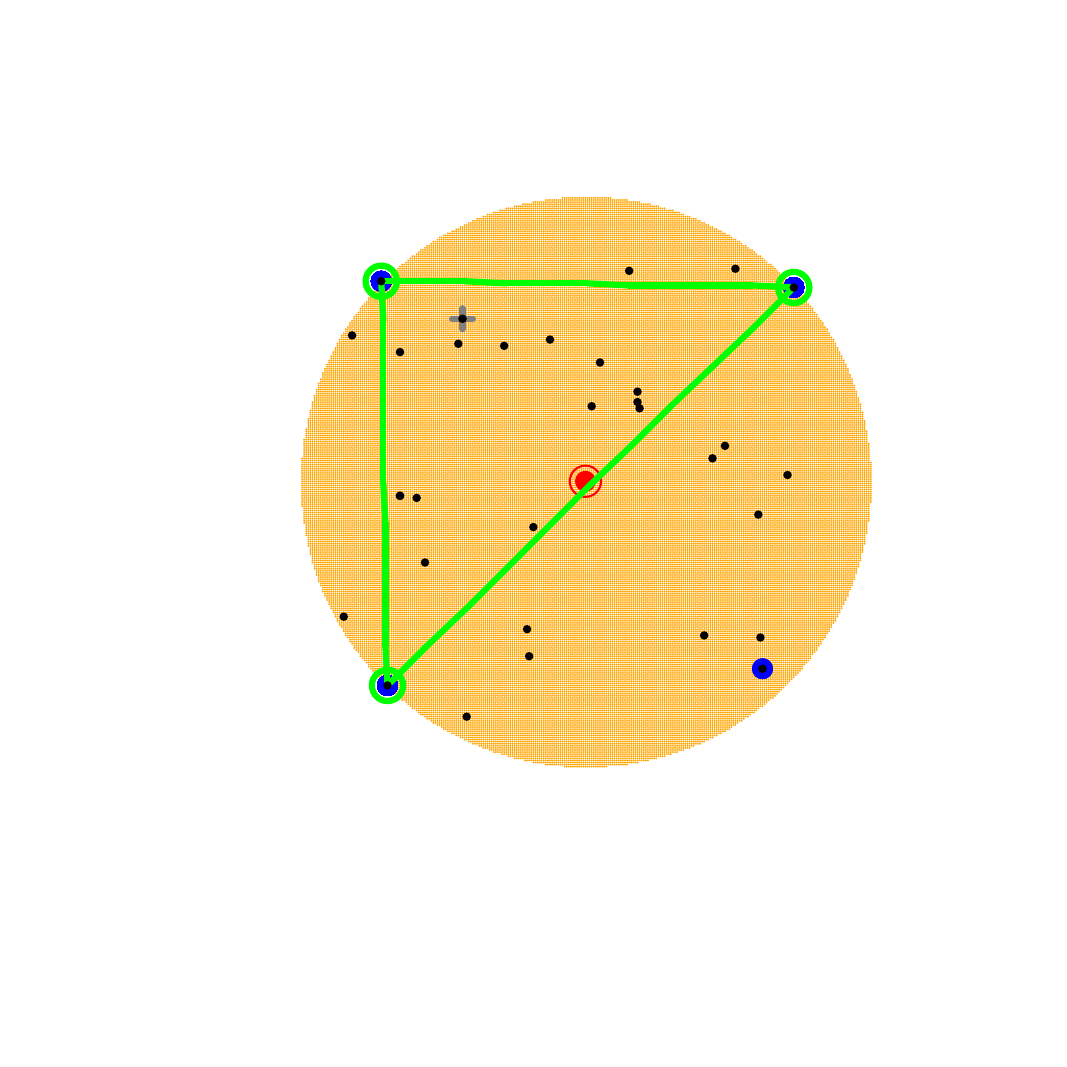}} &
\fbox{\includegraphics[width=\picw]{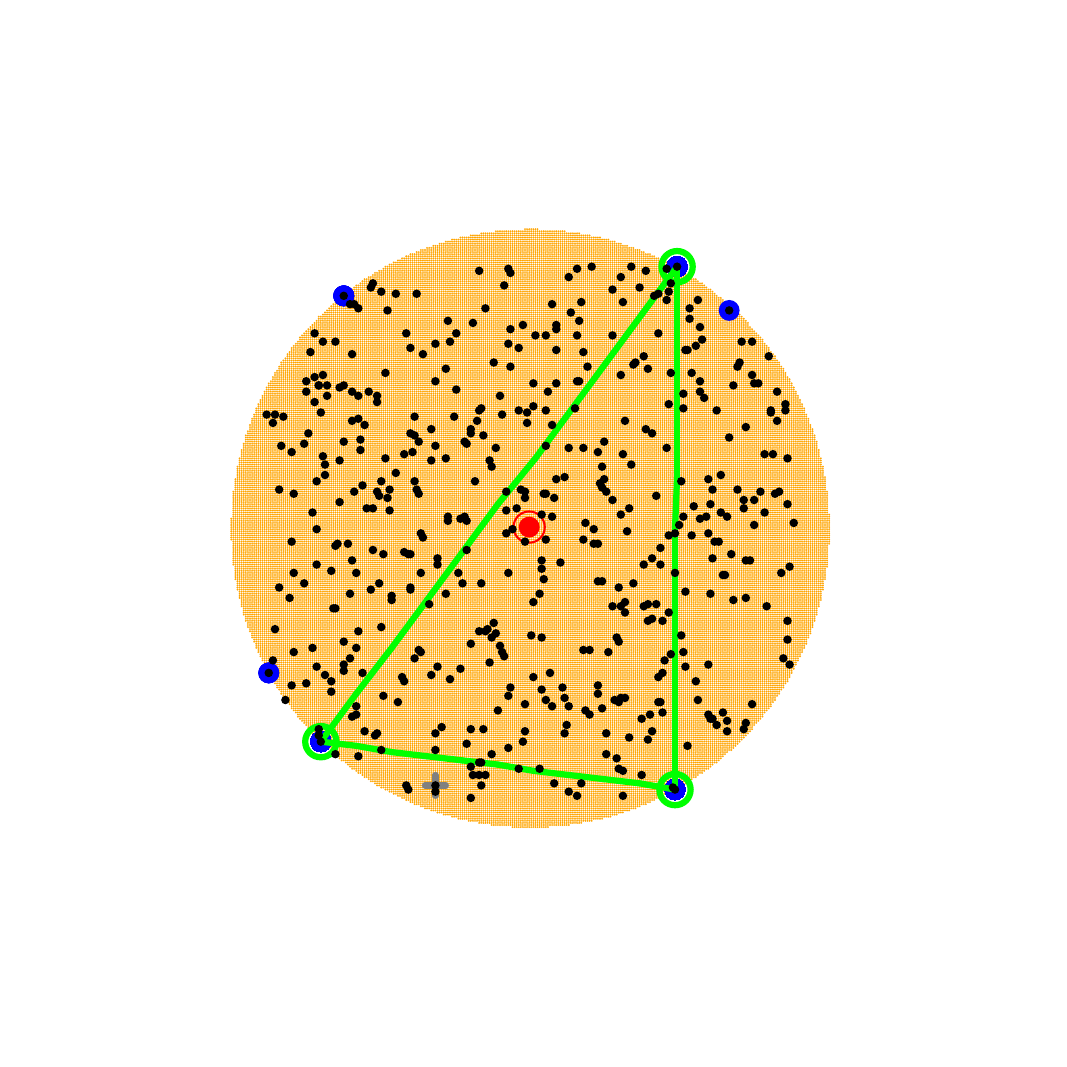}} \\
\fbox{\includegraphics[width=\picw]{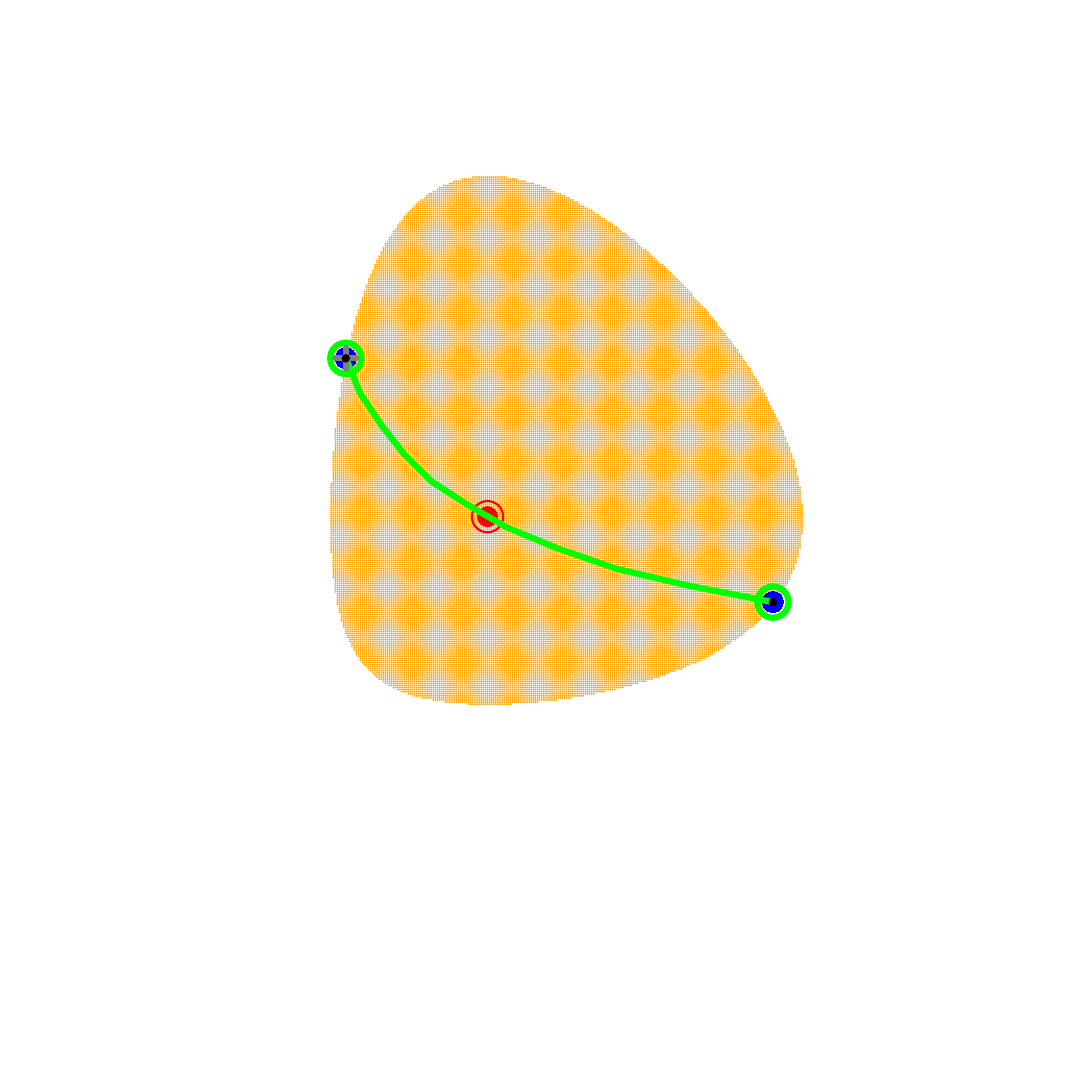}} &
\fbox{\includegraphics[width=\picw]{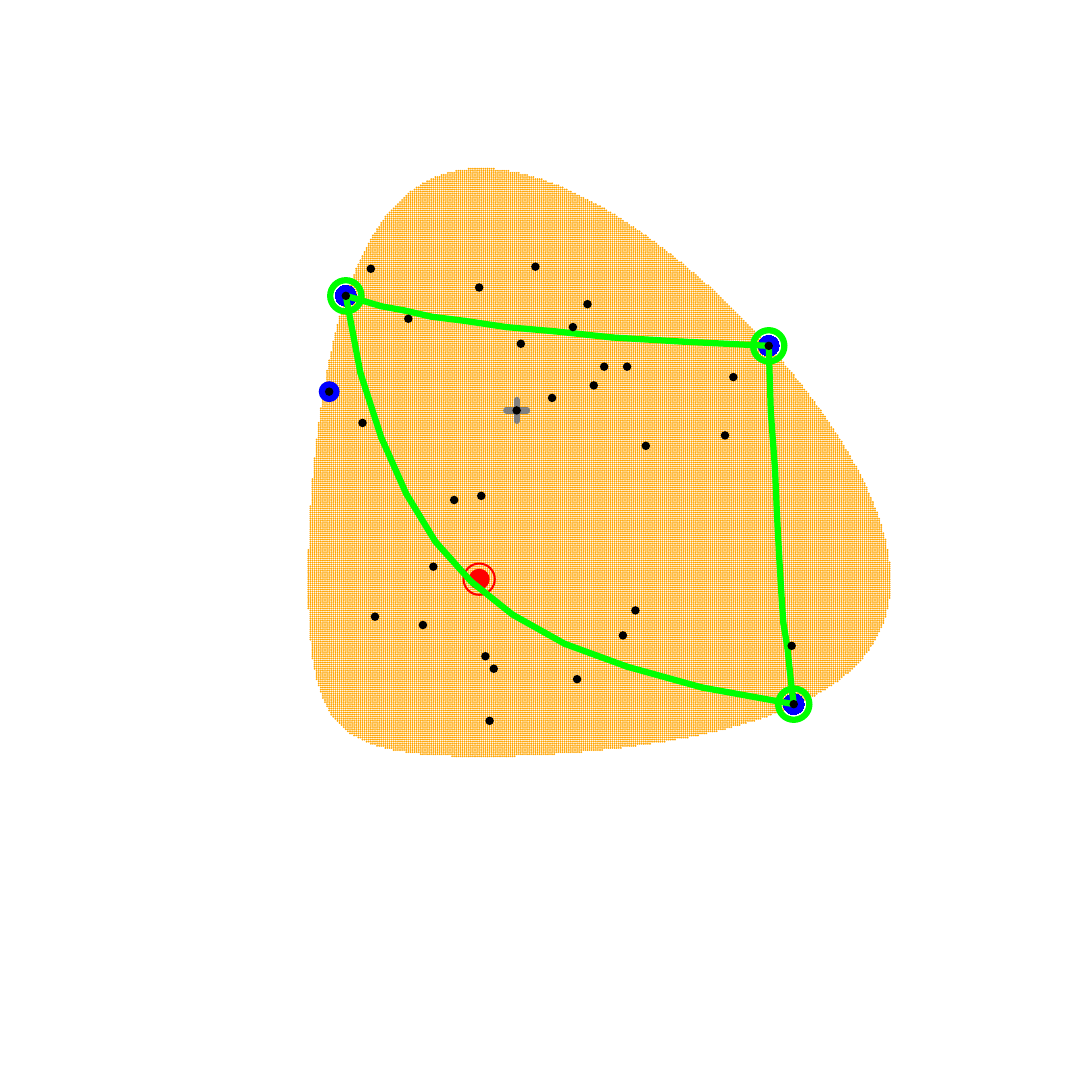}} &
\fbox{\includegraphics[width=\picw]{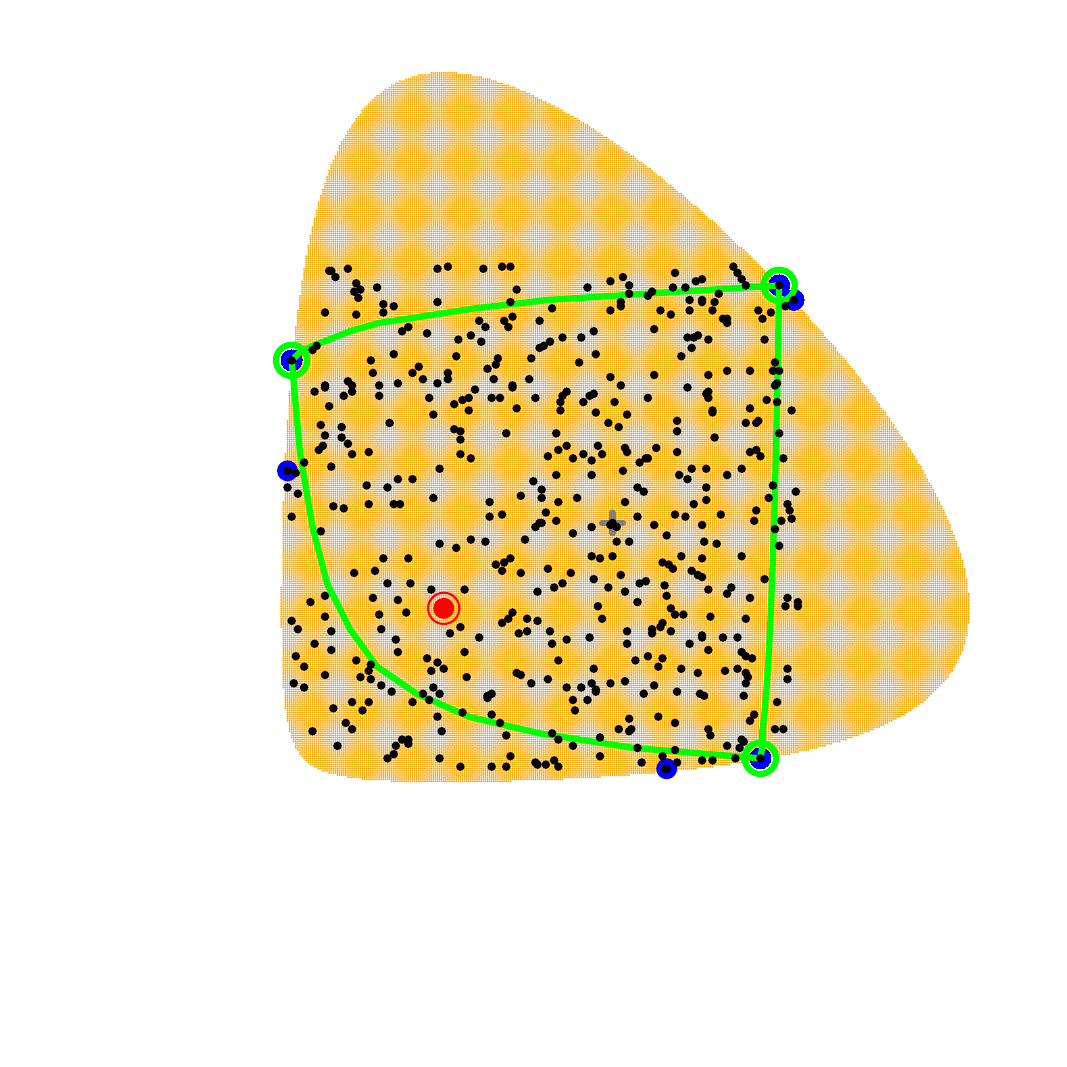}} \\
\fbox{\includegraphics[width=\picw]{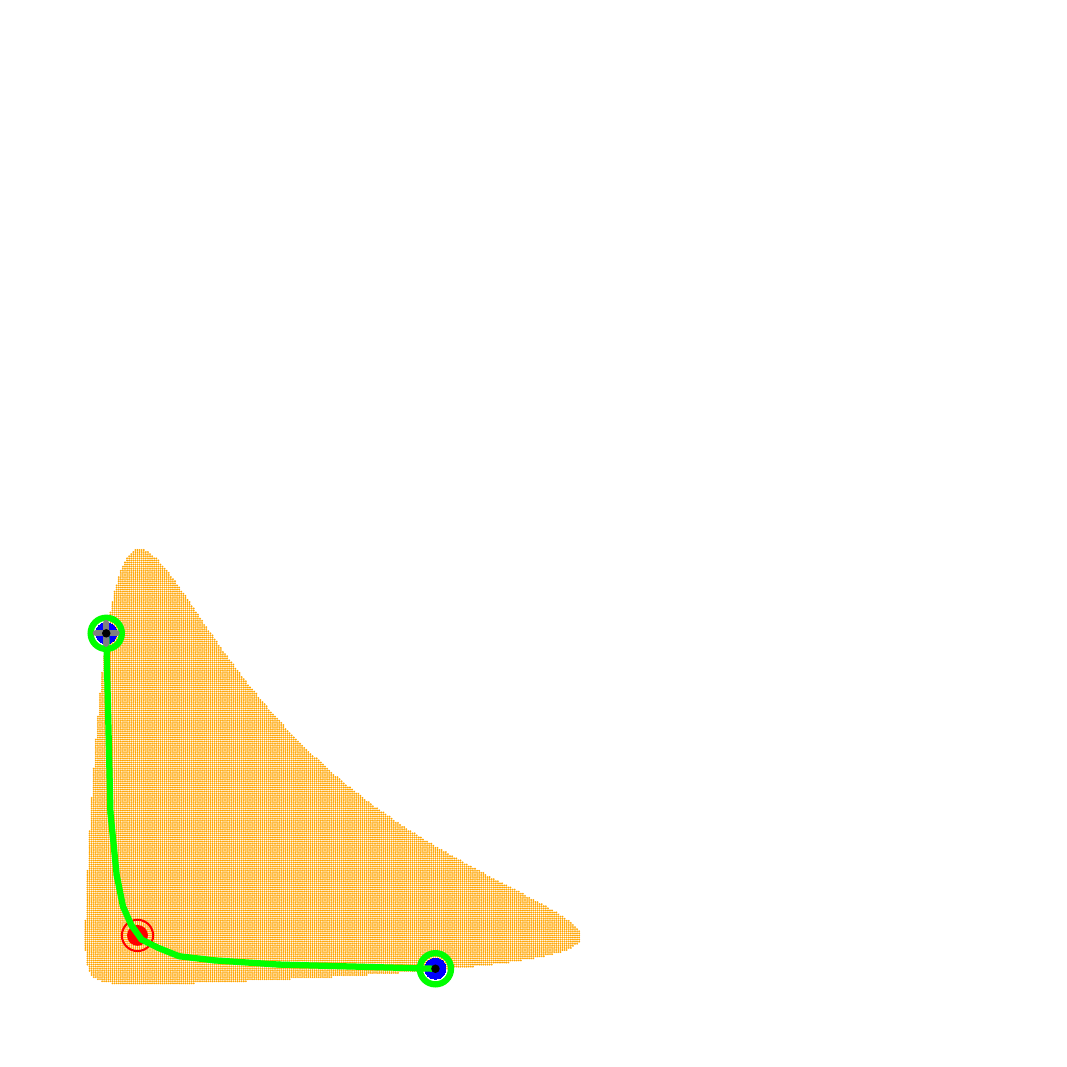}}  
 &
\fbox{\includegraphics[width=\picw]{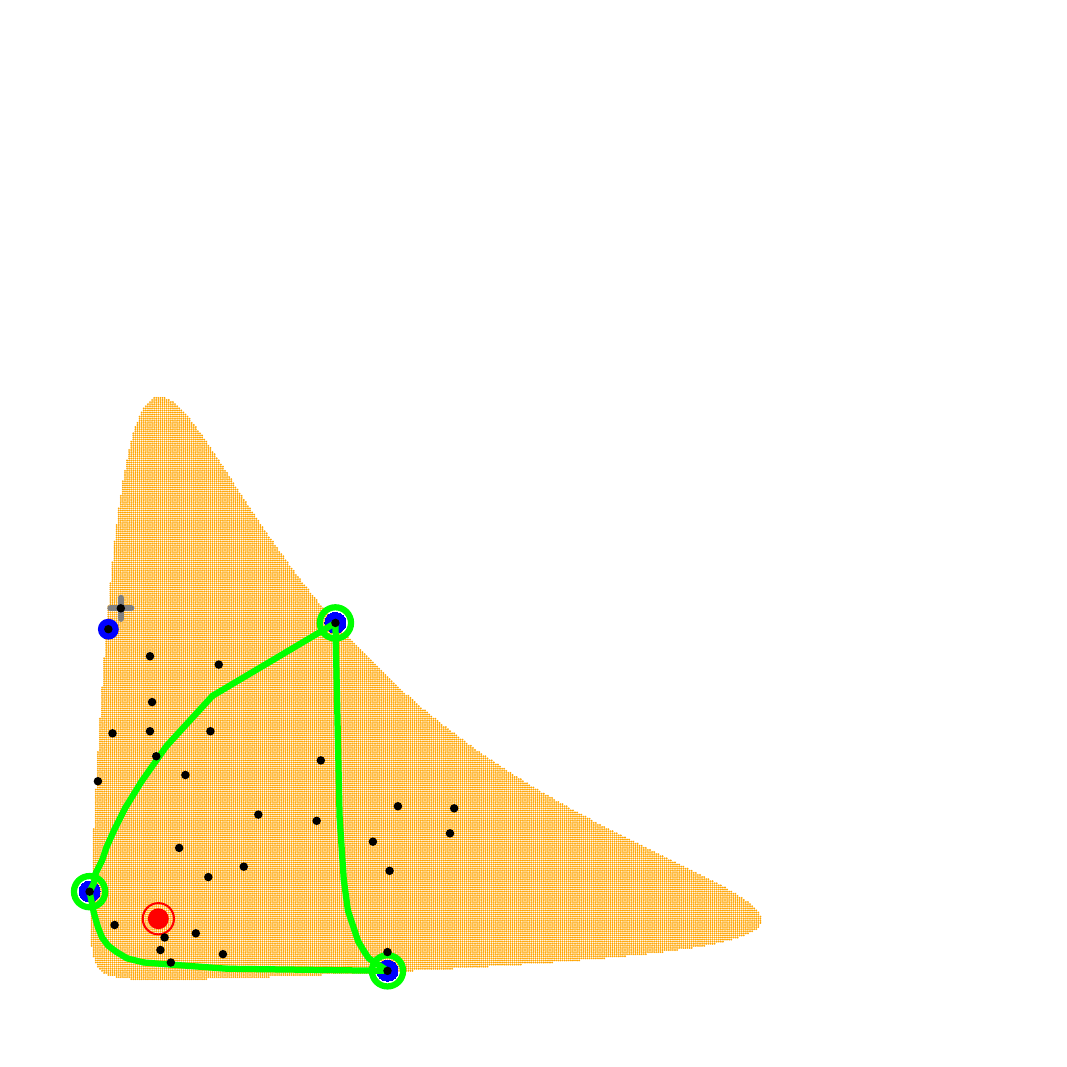}}  
 &
\fbox{\includegraphics[width=\picw]{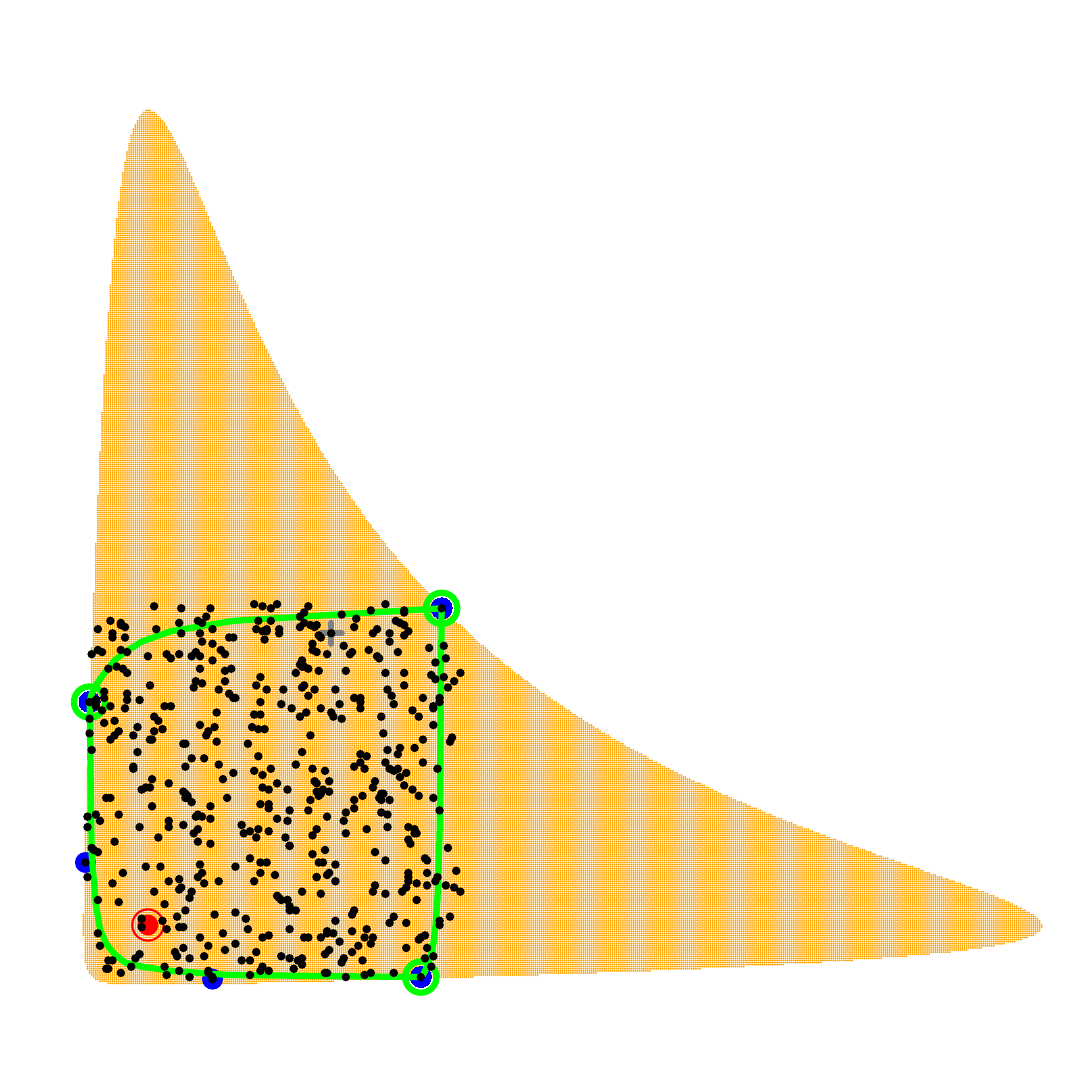} } 
\end{tabular}

\caption{Examples of Bregman exact/approximate MEBs of $n$ points for $n\in\{2,32,512\}$.
From top to bottom: squared Euclidean divergence, extended Kullback-Leibler divergence, and Itakura-Saito divergence.
The exact Bregman MEB is induced by a  basis $\calV$ of size $2$ or $3$ points with corresponding geodesic simplex shown in green.
The Frank--Wolfe core-set of algorithm   is indicated in blue.  
}\label{fig:SEBB}
\end{figure}

\begin{figure}%
\begin{tabular}{ccc}\includegraphics[width=0.3\columnwidth]{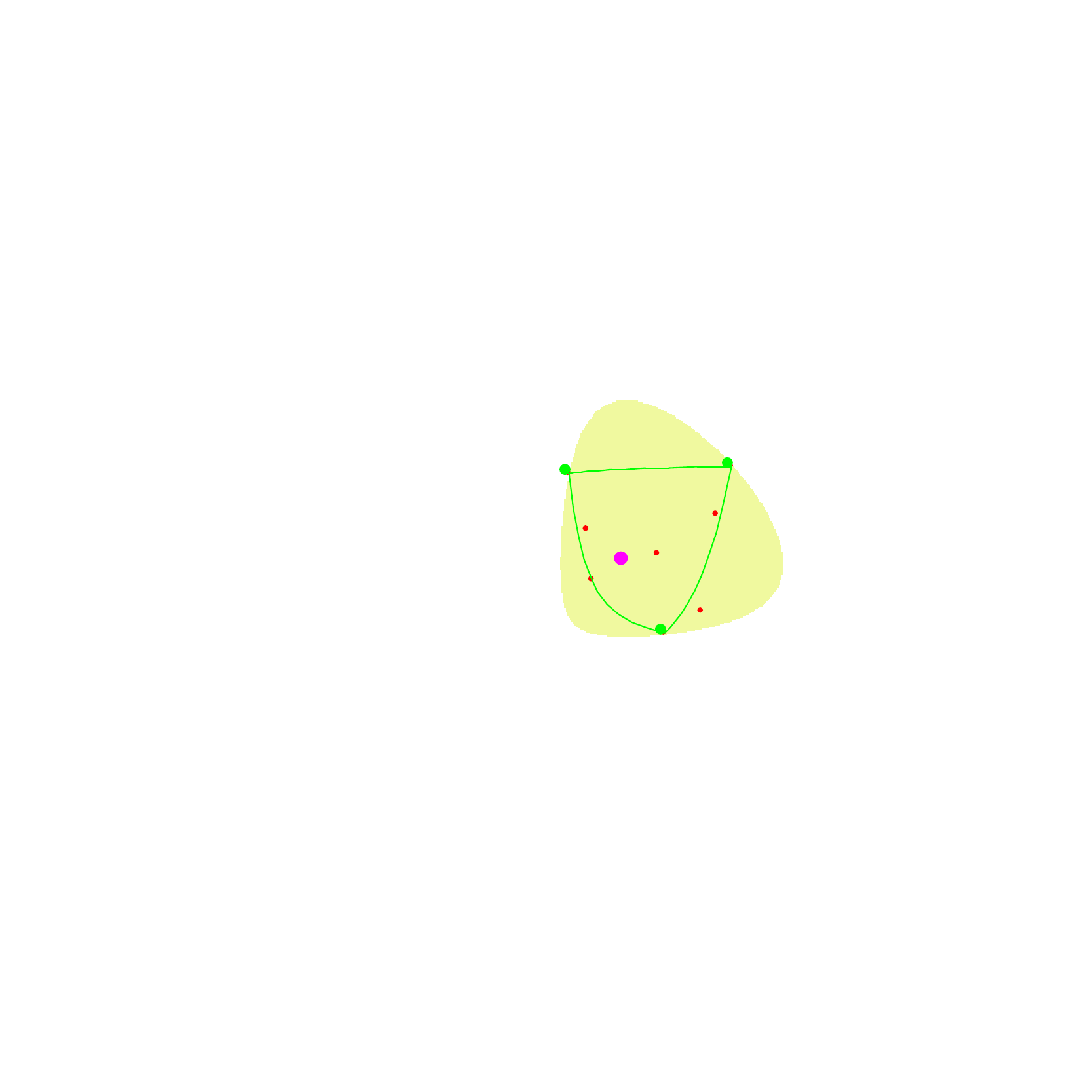} &
\includegraphics[width=0.3\columnwidth]{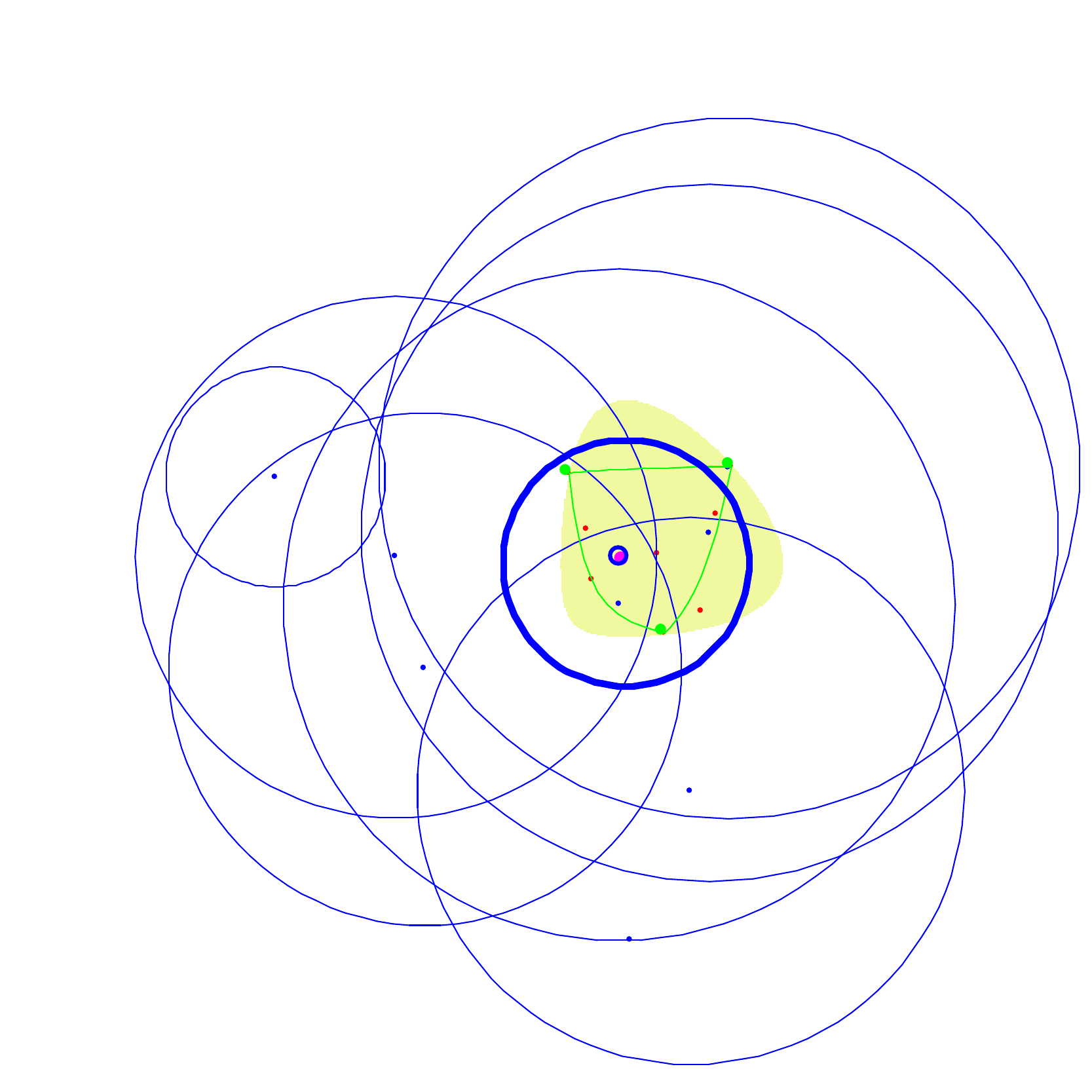} &
\includegraphics[width=0.3\columnwidth]{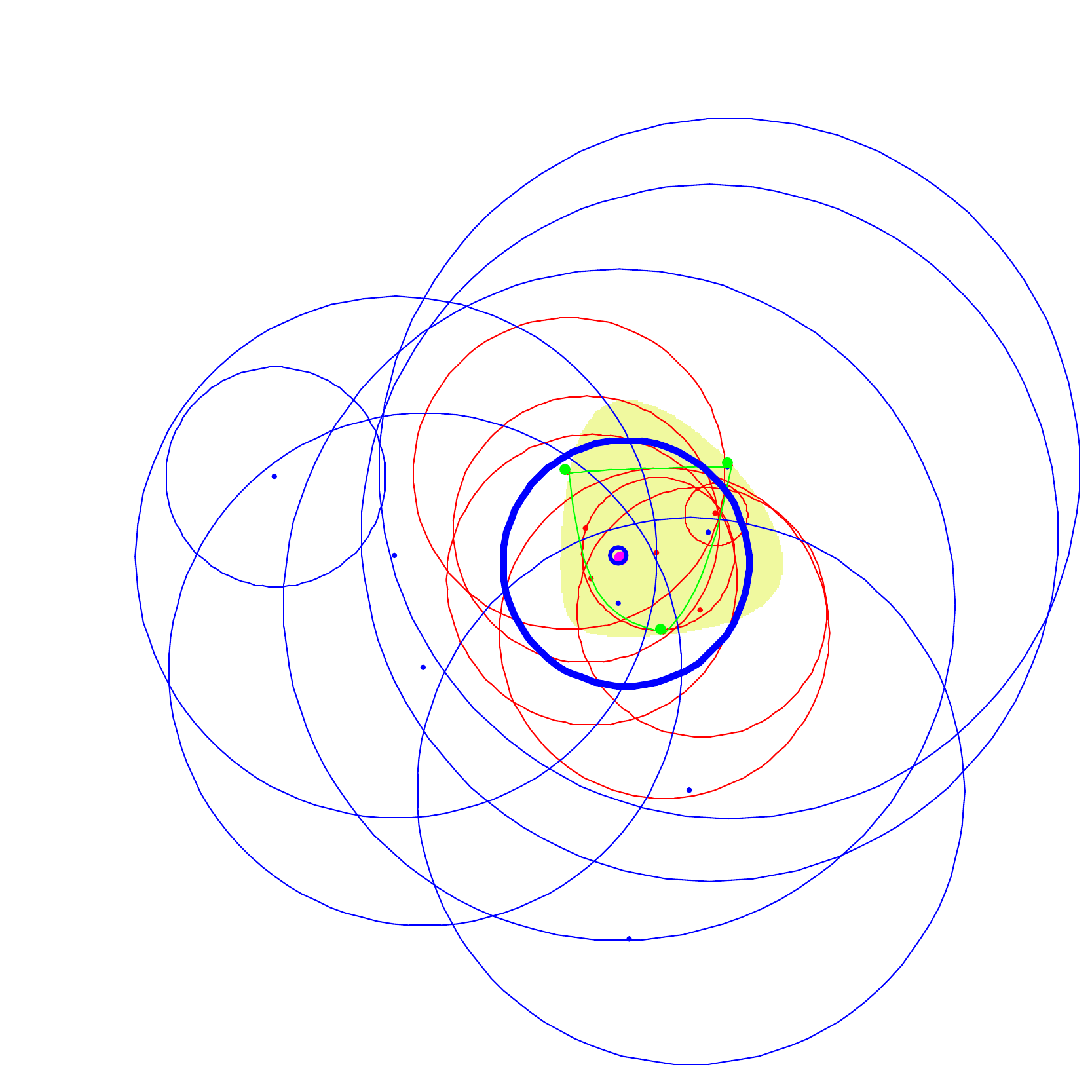}\\
(a) & (b) & (c)
\end{tabular}
\caption{Bregman minimum enclosing ball of a set of $n=8$ points (red) for the extended Kullback-Leibler divergence (a).
(b) Bregman MEB relaxed to a power MEB (thick blue circle with blue circumcenter enclosing the Bregman circumcenter)  for the dual weighted points (blue) shown as circles (after offsetting the weights as  $\sqrt{B-b_i}$ where $B=\max_i b_i$ so they all become non-negative).
(c) Visualizing the primal weighted points $\wtheta_i$ (red) and dual weighted points $\eta_i$ (blue) as circles after offsetting their weights so that all radii are positive.}%
\label{fig:MEBex3}%
\end{figure}

Power MEBs can be solved exactly using the LP-type paradigm.
We describe the Frank--Wolfe numerical procedure for  approximating the power MEB of a generic weighted point set in algorithm~\ref{alg:FWpowerMEB}.
The algorithm {\sc FrankWolfePowerMEB} is similar to the Bad\u oiu-Clarkson   algorithm~\cite{clarkson2010coresets} except that the farthest point is selected according to the power distance depending on the point weights instead of the Euclidean distance.
The analysis of convergence of this procedure is classic and postponed in  Appendix~\ref{sec:FWapp}.

\begin{algorithm}
\caption{{\sc FrankWolfePowerMEB}($\hat\calS=\{\hat p_i=(p_i,w_i)\}$, $T$)}\label{alg:FWpowerMEB}
Choose at random $c_1\in \{\xi_i\}$\;
\For{$t=1,2,\ldots,T-1$}{
    $f_t= \leftarrow \arg\max_{i\in [n]} \sigma(c_t,\hat{p}_i)$\;
    $c_{t+1}=
    \frac{t}{t+1}c_t
    + \frac{1}{t+1}p_{f_t}
    $\;
}
Return $\tilde{c}=c_T$\;
\end{algorithm}

%% file: sec-conclusion.tex
%%%%%%%%%%%%%%%%%%%%%%%%%%%%
\section{Summary and Discussion}\label{sec:concl}
%%%%%%%%%%%%%%%%%%%%%%%%%%%% 

In this work, we defined Bregman power divergences which unifies power distances with Bregman divergences and reported two geometric interpretations of Bregman power divergences:

\begin{itemize}
\item A first interpretation is illustrated by a vertical gap built from the graph of the Bregman potential function (Figure~\ref{fig:bpd}).

\item A second interpretation relies on geodesic duality and the Pythagoras' theorem of information geometry which generalizes the classic interpretation of power distances (Figure~\ref{fig:bpdeta}).
\end{itemize}

 We explained how Bregman power divergences can be represented by ordinary double-sided power distances using the mixed parameterization (Proposition~\ref{prop:equivydpd}) and 
show how to perform Bregman proximity queries as equivalent Euclidean proximity queries using an extra dimension.
 We put this in practice by showing how to  convert  Bregman diameter and Bregman closest pair problems as Euclidean counterparts.
	The power distance representation of Bregman divergences can be used to solve other Bregman problems as corresponding Euclidean problems on bichromatic point sets that can be efficiently implemented using legacy data-structures and software libraries of computational geometry~\cite{de2008computational}.

Power distances are used to build power diagrams~\cite{aurenhammer1987power} also called Laguerre Voronoi diagrams or Laguerre tesselations~\cite{lautensack2008random}.
Laguerre geometry is a geometry based on oriented spheres and hyperplanes which considers oriented tangencies. 
(Laguerre geometry can also be studied in hyperbolic and spherical spaces~\cite{bobenko2021non}.)

We also explained how to remap hyperplanes supporting lifted Bregman spheres  to hyperplanes supporting Euclidean spheres (i.e., defined via the paraboloid lifting transform). 
This brought another complementary viewpoint of the property that Bregman power divergences can be expressed as double-sided power distances.
 	
We showed to convert the calculation of Bregman circumspheres as equivalent power circumspheres (called orthospheres) in Proposition~\ref{prop:equivcircumcenter}.
The calculation of circumscribing spheres rely on determinant calculations which are rooted in the spaces of Bregman spheres~\cite{BVD-2010}.
In 1D, we reported exact formula for the Bregman circumcenter and Bregman radius (Proposition~\ref{prop:BMEB1D}) and the orthocenter and orthoradius.

From the viewpoint of geometry the equivalence of Bregman divergences with double-sided power distances using the mixed parameterization emphasizes the role of the inner product and can be geometrically interpreted as depicted in Figure~\ref{fig:charts}:
Namely, the two canonical dual $\theta/\eta$ charts of Bregman manifolds in information geometry~\cite{IG-2016} are superposed in $\bbR^m$.
We may call this viewpoint of Bregman divergences/Bregman manifolds ``superposed duality''.
It is interesting to further study interactions of Laguerre geometry with Hessian manifolds~\cite{shima2007geometry} also called dually flat spaces~\cite{IG-2016} in information geometry.

\begin{figure}
\centering
\includegraphics[width=\textwidth]{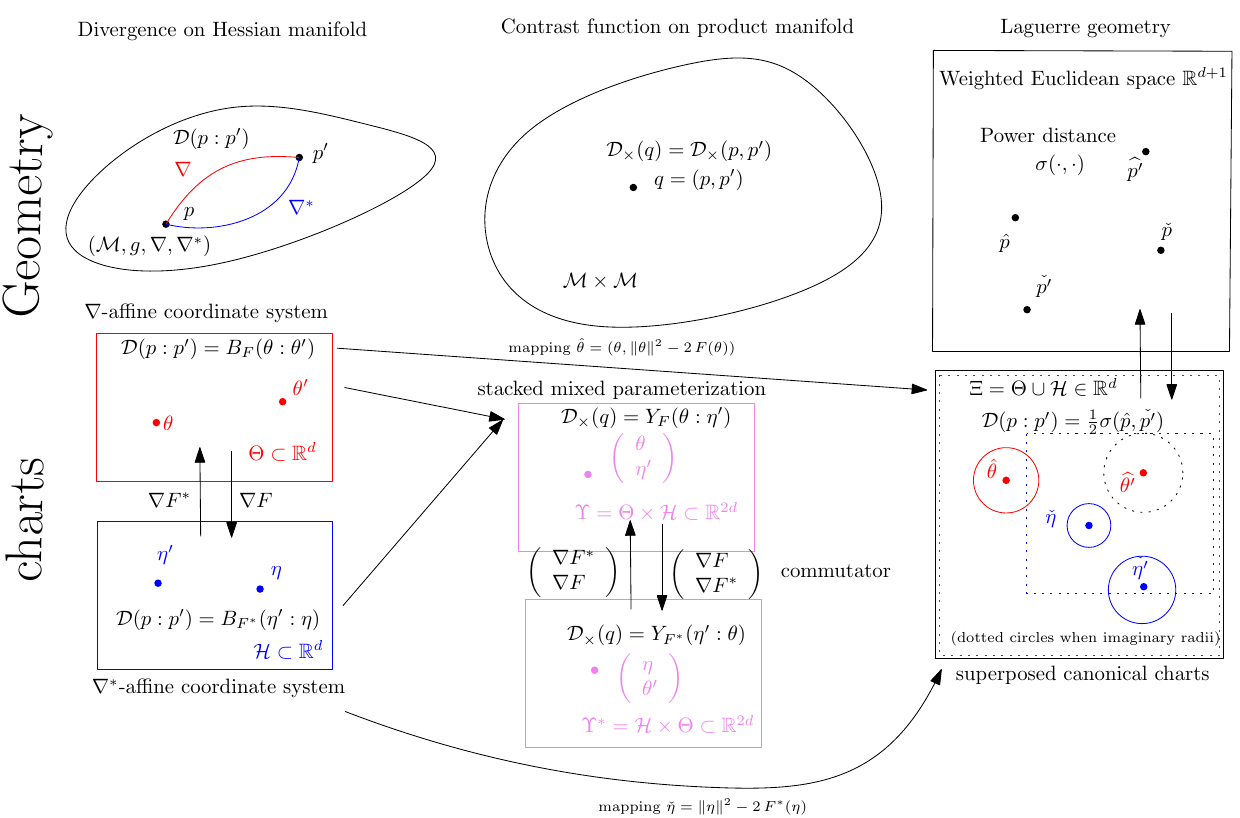}
\caption{Dual Bregman divergences are coordinatized divergences of a dually flat divergence between points on a Bregman Hessian manifold. 
The Fenchel-Young divergences are the corresponding contrast functions on the product manifold using mixed parameterization.
The Bregman/Fenchel-Young divergence amounts to a power distance in the Laguerre geometry of spheres using the weighted point  mappings $\hat{\cdot}$ and $\check{\cdot}$.
In the Euclidean coordinate system, weighted points $\hat{p}=(p,w)$ are visualized as either real circles (plain) when $w\geq 0$ or imaginary circles when $w<0$ (dotted).
}\label{fig:charts}
\end{figure}

%To conclude, this work has emphasized how to remap  Bregman generators to the canonical  paraboloid  which is the generator of the (squared) Euclidean distance, and vice-versa.
We note that power distances have been defined more generally in Hilbert spaces~\cite{PowerMEB-TR-2023} with applications in topological data analysis.
A promising direction for future work is to consider how power distances can play a role in studying and implementing pointwise and functional Bregman divergences~\cite{frigyik2008functional,tsuchida2026generalising}.

%% file: notations2.tex
% notations 
%%%%

\begin{tabular}{ll}
$m$ or $d$, $n$, $k$ & dimension, number of points, support size of circumscribing spheres\\
$u\cdot v$ & scalar product (dot product) between vectors $u$ and $v$\\
$\hat p=(p,w)$, $\hatp=(p,\pi)$ $\hat\theta$, $(\theta,\varpi)$ & a weighted point or parameter\\ %, $p\in\bbR^m$, $w\in\bbR$\\
$\sigma(x,\hat p)$, $\sigma(\hat p,\hat p')$ & single-sided and double-sided power distances\\
$B(\hat p)$, $S(\hat p)$ & Representation of the weighted point $\hat p$ by a ball or a sphere\\
$\Theta$ & primal parameter space\\
$B_F(\theta:\theta')$ & Bregman divergence between parameters $\theta$ and $\theta'$\\
$\eta=\nabla F(\theta)$ & dual parameter\\
$\calH$ & dual parameter space\\
$F^*(\eta)$ & Fenchel-Young convex conjugate\\
$Y_F(\theta:\eta')$ & Fenchel-Young divergence between mixed parameters $\theta$ and $\eta'$\\
$\hat\theta=(\theta,w_F(\theta))$ & weighted point corresponding to primal parameter $\theta$\\
$\check\eta=(\eta,w_{F^*}(\eta))$ & weighted point corresponding to dual parameter $\eta$\\
$\hat{\hat\theta}$ & coupled weighted parameter on primal parameter: $(\theta,w+w_F(\theta))$\\
$\check{\check\eta}$ & coupled weighted parameter on dual parameter: $(\eta,w+w_{F^*}(\eta))$ \\
$\calW_2(\mu,\mu')$ & Wasserstein distance between probability measures $\mu$ and $\mu'$\\
$\calW_F(\mu,\mu')$ & Bregman-Wasserstein distance between measures $\mu$ and $\mu'$ on $\Theta$\\
$Q(\theta)=\frac{1}{2}\inner{\theta}{\theta}$ & half paraboloid function\\
$Q_2(\theta)=2Q(\theta)$ &   paraboloid function\\
$B_{F_1,F_2}(\theta:\theta')$ & duo Bregman pseudo-divergence\\
$Y_{F_1,F_2}(\theta:\eta')$ & duo Fenchel-Young pseudo-divergence\\
$\cal F$ & graph of the function $\calF$\\
$\hat\theta^F \eqdef F(\theta)-w$ & $y$-coordinate of the lifted weighting point to $\calF$\\
$h_{\hat\theta}^\calF(x)$, $h_{\hat\eta}^\calF(x)$ & $y$-function of the tangent plane at $\hat\theta$ or $\hat\eta$\\
$\calA=\{\theta_1,\ldots,\theta_n\}$ & set of primal parameters\\
$\calB=\{\eta_1,\ldots, \eta_n\}$ & set of dual parameters\\
$\check\calB=\{\check\eta_i=(\eta_i,b_i)\}$ & set of dual weighted parameters\\
$B=\max_j b_j$ & maximum weight of dual weighted parameters $\hat\calB$\\
$S_F(\theta,\theta')$ & symmetrized Bregman divergence (Jeffreys-Bregman)\\
$\calD(p:q)$ & Dually flat divergence on a Bregman manifold $M$\\
$\eta_i^{\shortparallel}$ & lifted parameter $\eta_i\in\bbR^m$ to $\bbR^{m+1}$ for proximity queries\\
$\theta^{=}$ & lifted query parameter $\theta$ for proximity queries \\
$\calK$ & Query set\\
$\Delta_F(\calA)$, $\delta_F(\calA)$ & Bregman diameter and closest pair for a parameter set $\calA$\\
$z=(\theta,y=F(\theta))$ & Embedding using an extra dimension\\
$\breve{p}=(p,\lambda,w)$ & A point $p$ with a multiplicative weight $\lambda>0$ and additive weight $w\in\bbR$\\
$\calV$, $\hat\calV$ & Support vectors of Bregman and power circumspheres\\
$\Bi_\sigma(\hatp_i,\hatp_j)$ & power bisector of weighted points $\hatp_i$ and $\hatp_j$\\
$\Bi_F(\theta_i,\theta_j)$, $\Bi_F(\hat\theta_i,\hat\theta_j)$ & Bregman bisector, power Bregman bisector\\
$\Vor_\sigma$, $\Vor_F$ & power diagram and Bregman Voronoi diagram\\
$\hat\calA$, $\hat\calB$ & primal and dual weighted parameter sets\\
$\wwcalA$, $\wwcalB$ & primal and dual weighted parameter sets augmented by coupled weights\\
$\calE=\{p_\theta\st\theta\in\Theta\}$ & exponential family with cumulant function $F$
\end{tabular}